\documentclass[11pt]{article}
\usepackage[left=3cm, right=3cm, top=3cm, bottom=3.5cm]{geometry}
\usepackage{tikz}
\usetikzlibrary{decorations.pathmorphing,patterns}
\usetikzlibrary{calc,angles,quotes}
\usetikzlibrary{intersections,arrows}
\usepackage{amsfonts,amsmath,amsthm}
\usepackage{epsfig,graphicx}
\usepackage{psfrag}
\usepackage{comment}
\usepackage{scalerel}

\def\qmax{q_{\rm max}}

\def\2toro{\mathbb{T}^2}
\def\R{\mathbb{R}}
\def\N{\mathbb{N}}
\def\dnod{d_{\rm nod}}
\def\dmin{d_{\rm min}}
\def\deltanod{\delta_{\rm nod}}
\def\deltaint{\delta_{\rm int}}
\def\deltaext{\delta_{\rm ext}}
\def\deltalink{\delta_{\rm link}}
\def\lowintom{\ell^\omega_{\rm int}}
\def\lowextom{\ell^\omega_{\rm ext}}

\def\uppintom{u^\omega_{\rm int}}
\def\uppextom{u^\omega_{\rm ext}}
\def\upplinkom{u^\omega_{\rm link}}
\def\lowinte{\ell^{\:e}_{\rm int}}
\def\lowexte{\ell^{\:e}_{\rm ext}}

\def\uppinte{u^{\:e}_{\rm int}}
\def\uppexte{u^{\:e}_{\rm ext}}
\def\upplinke{u^{\:e}_{\rm link}}

\def\lowextomp{\ell^{\omega'}_{\rm ext}}

\def\uppintomp{u^{\omega'}_{\rm int}}
\def\uppextomp{u^{\omega'}_{\rm ext}}
\def\upplinkomp{u^{\omega'}_{\rm link}}
\def\rpiu{r_{\scaleto{+}{5pt}}}
\def\rmeno{r_{\scaleto{-}{5pt}}}
\def\rpiup{r_{\scaleto{+}{5pt}}'}
\def\rmenop{r_{\scaleto{-}{5pt}}'}

\newtheorem{proposition}{\bf Proposition}
\newtheorem{lemma}{\bf Lemma}

\newtheorem{remark}{\bf Remark}
\newtheorem{definition}{\bf Definition}

\title{On the nodal distance between two Keplerian trajectories with
  a common focus}

\begin{document}
\author{Giovanni F. Gronchi\footnote{Dipartimento di Matematica,
    Universit\`a di Pisa} \ and Laurent
  Niederman\footnote{D\'epartment de Math\'ematiques d'Orsay,
    Universit\'e Paris-Sud}}

\maketitle

\begin{abstract}
 We study the possible values of the nodal distance $\deltanod$
 between two non-coplanar Keplerian trajectories ${\cal A}, {\cal A}'$
 with a common focus.  In particular, given ${\cal A}'$ and assuming
 it is bounded, we compute optimal lower and upper bounds for
 $\deltanod$ as functions of a selected pair of orbital elements of
 ${\cal A}$, when the other elements vary. This work arises in the
 attempt to extend to the elliptic case the optimal estimates for the
 orbit distance given in \cite{GV2013} in case of a circular
 trajectory ${\cal A}'$.  These estimates are relevant to understand
 the observability of celestial bodies moving (approximately) along
 ${\cal A}$ when the observer trajectory is (close to) ${\cal A}'$.
\end{abstract}

\section{Introduction}

The computation of the distance $\dmin$ between two Keplerian
trajectories ${\cal A}$, ${\cal A}'$ with a common focus, also called
orbit distance, is relevant for different purposes in Celestial
Mechanics.  Several authors introduced efficient methods to compute
$\dmin$, e.g. \cite{sitarski1968}, \cite{KV1999b}, \cite{G2002},
\cite{G2005}.  Small values of $\dmin$ are relevant for the assessment
of the hazard of near-Earth asteroids with the Earth
\cite{milani2006}, \cite{FCMGC}, or for the detection of conjunctions
between satellites of the Earth \cite{hoots1984}, \cite{CTL2014}.  On
the other hand, we may wish to check whether $\dmin$ can assume large
values, because in this case it is more difficult to observe a small
celestial body moving along ${\cal A}$ from a point following ${\cal
  A}'$.

In \cite{GV2013} the authors studied the range of the values of the
orbit distance $\dmin$ between the trajectory ${\cal A}'$ of the
Earth, assumed to be circular, and the possible trajectory ${\cal A}$
of a near-Earth asteroid, as a function of selected pairs of orbital
elements.  The results have been used to detect some observational
biases in the known population of near-Earth asteroids (NEAs).  We
would like to extend these results to the case of an elliptic
trajectory ${\cal A}'$.  This generalization seems to be difficult
because $\dmin$ is implicitely defined, and because two local minima
of the distance between a point of ${\cal A}$ and a point of ${\cal
  A}'$ may exchange their role as the orbit distance, see
\cite{GV2013}.  Therefore, as a first step in this direction, we
investigate the range of the values of the nodal distance $\deltanod$,
which is defined explicitely by equation \eqref{deltanod}. The
distance $\deltanod$ is defined only when the two trajectories are not
coplanar, and is similar to $\dmin$ for some aspects: $\deltanod = 0$
if and only if $\dmin=0$, moreover the absolute values of the
ascending and descending nodal distances may exchange their role as
the nodal distance.  We also have $\dmin\leq \deltanod$, thus the
nodal distance gives us an upper bound to the orbit distance.

The ascending and descending nodal distances have also been used in
\cite{KV1999a} to define linking coefficients as functions of the
orbital elements and to estimate the orbit distance. A lower bound for 
the orbit distance is also given in \cite{MB2019}.

This paper is organized as follows. In Section~\ref{s:prelim} we
introduce the nodal distance $\deltanod$ and show some basic
properties.  In Section~\ref{s:bounds} we present the main results,
that is optimal bounds for $\deltanod$: first we deal with the case of
an eccentric trajectory ${\cal A}'$, with $e'\in(0,1)$, then we
consider the particular case $e'=0$ and compare the results with the
ones in \cite{GV2013}.  In Section~\ref{s:appl} we show an application
of the results to the known population of NEAs. Finally, in
Section~\ref{s:dmin}, we discuss the analogies and the differences
between the optimal upper bounds of $\deltanod$ and $\dmin$ on the
basis of numerical computations.

\section{Preliminary definitions and basic properties}
\label{s:prelim}

\subsection{Mutual orbital elements}

Given two non-coplanar Keplerian trajectories ${\cal A}, {\cal A}'$
with a common focus, we define the {\em cometary mutual elements}
\[
{\cal E}_M = (q,e,q',e',I_M,\omega_M,\omega_M')
\]
as follows: $q,e$ and $q',e'$ are the pericenter distance and the
eccentricity of the two trajectories, $I_M$ is the mutual inclination
between the two orbital planes and $\omega_M,\omega_M'$ are the angles
between the ascending mutual node\footnote{defined by assigning an
  orientation to both trajectories.} and the pericenters of ${\cal A}$
and ${\cal A}'$, see Figure~\ref{fig:mutelems}.

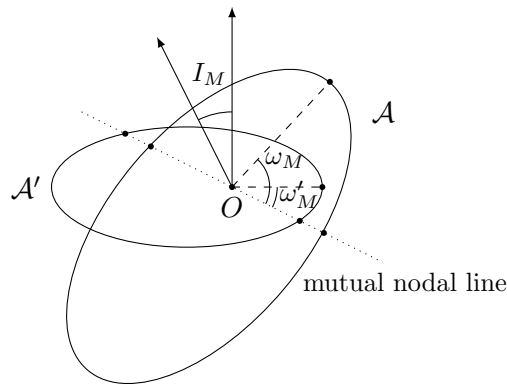
\begin{figure}[h!]
\centering
\begin{tikzpicture}
  \coordinate (F) at (0,0);
  \node[below] at (F) {$O$};
  \coordinate (Lm) at (-2,1);
  \coordinate (Lp) at (2,-1);
  \coordinate (Nm1) at (-2*0.71,1*0.71);
  \coordinate (Np1) at (2*0.45,-1*0.45);
  \coordinate (Nm2) at (-2*0.54,1*0.54);
  \coordinate (Np2) at (2*0.61,-1*0.61); 
  \coordinate (q) at (1.3,1.4);
  \coordinate (qp) at (1.2,0);
  \coordinate (v) at (-1,2);
  \coordinate (vp) at (0,2.4);
  \draw (-0.6,0) ellipse (1.8cm and 0.8cm);
  \draw[rotate=50] (-0.6,-0.1) ellipse (2.5cm and 1.3cm);
  \node at (2.3,1) [left] {${\cal A}$};
  \node at (-2.4,0) [left] {${\cal A}'$};
  \fill (F) circle (0.4mm);
  \fill (q) circle (0.4mm);
  \fill (qp) circle (0.4mm);
  \fill (Nm1) circle (0.4mm);
  \fill (Np1) circle (0.4mm);
  \fill (Nm2) circle (0.4mm);
  \fill (Np2) circle (0.4mm);
  \draw[dotted] (Lm)--(Lp);
  \node at (2.3,-1) [below] {\small mutual nodal line};
\draw[-latex] (F)--(v);\draw[-latex] (F)--(vp);
  \draw[dashed] (F)--(q);
  \draw[dashed] (F)--(qp);
  \pic [draw, angle eccentricity=1.5, angle radius = 0.5cm] {angle=Np1--F--q};
  \node at (0.71,0.36) {\small $\omega_M$};
  \pic[draw, angle eccentricity=0.9, angle radius=0.6cm]{angle=Np1--F--qp};
  \node at (0.9,-0.1) {\small $\omega_M'$};
  \pic[draw, angle eccentricity=0.6, angle radius=1.cm]{angle=vp--F--v};
\node at (-0.3,1.5) {\small $I_M$};
  
\end{tikzpicture}
\caption{The mutual orbital elements $I_M, \omega_M, \omega_M'$.}
  \label{fig:mutelems}
\end{figure}

The map
\[
\Phi: (E,E') \rightarrow {\cal E}_M,
\]
from the usual cometary elements
\[
E = (q,e,I,\Omega,\omega),\qquad E' = (q',e',I',\Omega',\omega')
\]
of ${\cal A}$, ${\cal A}'$, to the mutual elements, is not injective:
there are infinitely many configurations leading to the same mutual
position of the two orbits.  We can select a unique set of orbital
elements $(E,E')$ in each counter-image $\Phi^{-1}\left({\cal
  E}_M\right)$ as follows:
\[
  E = (q,e,I_M,0,\omega_M), \qquad E' = (q',e',0,0,\omega_M') .
\]
This corresponds to computing the usual cometary elements with respect
to the mutual reference frame $Oxyz$, with the $x$-axis along the
mutual nodal line, oriented towards the ascending mutual node,
assuming that ${\cal A}'$ lies on the $xy$ plane.
Another possible choice is
\[
E = (q,e,I_M,-\omega_M',\omega_M), \qquad E' = (q',e',0,0,0),
\]
where we choose the reference $Oxyz$ with the $x$-axis along the
apsidal line of ${\cal A}'$, oriented towards its pericenter. In this
way,
for a given choice of the pericenter distance $q'$ and the
eccentricity $e'$ of ${\cal A}'$, we can vary all the other mutual
elements by changing only the elements of ${\cal A}$.

For simplicity, from now on we shall drop the subscript in $I_M,
\omega_M, \omega_M'$ and the adjective `mutual' referred to the
nodes and to the nodal distances.
We assume that $q'>0$ and $e'\in[0,1)$ are given, and let the
other mutual elements vary in the following ranges:
\[
0 < q\leq \qmax, \qquad 0\leq
e\leq 1, \qquad 0<I<\pi, \qquad 0\leq \omega,\omega'< 2\pi,
\]
for a given $\qmax>0$.

Moreover, we admit that the considered functions of the mutual orbital
elements attain the values $+\infty$ and $-\infty$, when
there exists an infinite limit for the value of such functions.

\subsection{The nodal distance}
\label{s:deltanod}

Let us set
\[
\begin{split}
&\rpiu = \frac{q(1+e)}{1+e\cos\omega},\qquad
  \rmeno = \frac{q(1+e)}{1-e\cos\omega},\\
  &\rpiup = \frac{q'(1+e')}{1+e'\cos\omega'},\quad \ \ 
  \rmenop = \frac{q'(1+e')}{1-e'\cos\omega'}\\
\end{split}
\]
and introduce the ascending and descending nodal distances:
\[
\dnod^+ = \rpiup - \rpiu, \qquad \dnod^- = \rmenop - \rmeno.
\]
\begin{definition}
  We define the (minimal) {\bf nodal distance} $\deltanod$ as the minimum
between the absolute values of the ascending and descending nodal
distances:
\begin{equation}
  \deltanod = \min\bigl\{|\dnod^+|, |\dnod^-|\bigr\}.
\label{deltanod}
\end{equation}
\end{definition}
Note that $\deltanod$ does not depend on the mutual inclination $I$.

\begin{remark}
 The transformations
 \[
 \begin{split}
&(\omega,\omega')\mapsto (\pi-\omega,\pi-\omega'),
\qquad
(\omega,\omega')\mapsto (\pi+\omega,\pi-\omega'),\\
&(\omega,\omega')\mapsto (2\pi-\omega,\omega'),
\qquad\quad\ 
(\omega,\omega')\mapsto (\omega,2\pi-\omega')\\
\end{split}
\]
leave the values of $\deltanod$ unchanged. 
\label{rem:sym}
\end{remark}
By the previous remark we get all the possible values of $\deltanod$
even if we restrict $\omega,\omega'$ to the following ranges:
\begin{equation}
0\leq \omega \leq \pi/2, \qquad 0\leq \omega'\leq\pi.
\label{deltanod_domain}
\end{equation}


\noindent We prove the following elementary facts:
%
\begin{lemma} 
Assuming $(\omega,\omega')\in[0,\pi]\times[0,\pi]$, the ascending
nodal distance $\dnod^+$ is a non-increasing function of $\omega$ and
a non-decreasing function of $\omega'$.  In the same domain the
descending nodal distance $\dnod^-$ is a non-decreasing function of
$\omega$ and a non-increasing function of $\omega'$.  Moreover, both
$\dnod^+$ and $\dnod^-$ are non-increasing functions of $e$.
\label{der_om_omp_e}
\end{lemma}

\begin{proof}
    We only need to compute the following derivatives:
  \[
  \begin{split}
  &\frac{\partial \dnod^+}{\partial\omega} = -\frac{e\sin\omega}{(1+e\cos\omega)}\rpiu, \qquad
  \frac{\partial \dnod^+}{\partial\omega'} = \frac{e'\sin\omega'}{(1+e'\cos\omega')}\rpiup,\\
   &\frac{\partial \dnod^-}{\partial\omega} = \frac{e\sin\omega}{(1-e\cos\omega)}\rmeno, \qquad
  \frac{\partial \dnod^-}{\partial\omega'} = -\frac{e'\sin\omega'}{(1-e'\cos\omega')}\rmenop,\\
  &\frac{\partial \dnod^+}{\partial e} = -\frac{q(1-\cos\omega)}{(1+e\cos\omega)^2}, \qquad
  \frac{\partial \dnod^-}{\partial e} = -\frac{q(1+\cos\omega)}{(1-e\cos\omega)^2}.\\
  \end{split}
  \]
  \end{proof}

We shall use this notation for the semi-latus rectum and for
the apocenter distance:
\[
p=q(1+e), \qquad p'=q'(1+e'), \qquad Q=\frac{q(1+e)}{1-e}, \qquad
Q'=\frac{q'(1+e')}{1-e'}.
\]
Moreover, we shall employ the variables
\[
\xi = e\cos\omega, \qquad \xi' = e'\cos\omega'.
\]

\begin{definition}
  We consider the following {\bf linking configurations} between the
  trajectories ${\cal A}, {\cal A}'$:
\begin{itemize}
  \item[-]{\bf internal nodes:} the nodes of ${\cal A}$ are internal
    to those of ${\cal A}'$, that is $\dnod^+, \dnod^- > 0$. A
    sufficient condition for this case is $Q<q'$;

  \item[-]{\bf external nodes:} the nodes of ${\cal A}$ are external
    to those of ${\cal A}'$ (possibly located at infinity), that is
    $\dnod^+, \dnod^- < 0$. A sufficient condition for this case is
    $q>Q'$;
    
  \item[-]{\bf linked orbits:} ${\cal A}$ and ${\cal A}'$ are
    topologically linked, that is
    $\dnod^+<0<\dnod^-$, or $\dnod^-<0<\dnod^+$;
  \item[-]{\bf crossing orbits:} ${\cal A}$ and ${\cal A}'$ have at
    least one point in common, that is $\dnod^+\dnod^- = 0$.
\end{itemize}
\label{linkcond}
\end{definition}
Assume $q'>0$ and $e'\in[0,1)$ are given. We introduce the functions
\[
\begin{split}
  &\deltaint(q,e,\omega,\omega') = \min\{\dnod^+, \dnod^-\},\\
  &\deltaext(q,e,\omega,\omega') = \min\{-\dnod^+, -\dnod^-\},\\
  &\deltalink^{(i)}(q,e,\omega,\omega') = \min\{-\dnod^+, \dnod^-\},\\
  &\deltalink^{(ii)}(q,e,\omega,\omega') = \min\{\dnod^+, -\dnod^-\},\\
  &\deltalink(q,e,\omega,\omega') = \max\{\deltalink^{(i)}, \deltalink^{(ii)}\}.\\
\end{split}
\]
The linking configurations
depend on the sign of these functions as described below.
\begin{lemma}
  Given the vector $(q,e,\omega,\omega')$, we have
\begin{itemize}
\item[a)] internal nodes if and only if $\deltaint(q,e,\omega,\omega')>0$,
\item[b)] external nodes if and only if $\deltaext(q,e,\omega,\omega')>0$,
\item[c)] linked orbits if and only if $\deltalink(q,e,\omega,\omega')>0$,
\item[d)] crossing orbits if and only if $\deltaint=\deltaext=\deltalink=0$ at $(q,e,\omega,\omega')$.
\end{itemize}
Moreover,
\begin{equation}
\deltanod = \max\{\deltaint,\deltaext,\deltalink\}.
\label{deltanod_max}
\end{equation}
\label{lem:delta}
\end{lemma}

\begin{proof}
 Properties $a)$ - $d)$ follow immediately from
 Definition~\ref{linkcond}.  Relation~\eqref{deltanod_max} follows
 from the fact that the linking configurations are mutually exclusive,
 therefore at least one of the expressions $\deltaint, \deltaext,
 \deltalink$ must be non-negative, and if one of these is strictly
 positive, then the other two are strictly negative.

\end{proof}

\section{Optimal bounds for the nodal distance}
\label{s:bounds}

In this section we state and prove optimal bounds for $\deltanod$ as
functions of selected pairs of orbital elements. For the case $e'=0$
we also compare the results with the ones obtained in \cite{GV2013}
for the orbit distance $\dmin$.

\subsection{Bounds for $\deltanod$ when $e'\in(0,1)$}

Assume $q'>0$ and $e'\in(0,1)$ are given. First we present the optimal lower and upper
bounds for $\deltanod$ as functions of $(q,\omega)$.

\begin{proposition}
Let ${\cal D}_1 =
\{(e,\omega'): 0\leq e\leq 1, 0\leq \omega'\leq \pi\}$, ${\cal D}_2 =
\{(q,\omega): 0< q\leq \qmax, 0\leq \omega\leq {\pi}/{2}\}$.  For
each choice of $(q,\omega)\in{\cal D}_2$ we have
  \begin{align}
  &\displaystyle \min_{(e,\omega')\in{\cal D}_1} \deltanod =
  \max\bigl\{0, \lowintom, \lowextom \bigr\}, \label{lbom}\\
%
  &\displaystyle \max_{(e,\omega')\in{\cal D}_1} \deltanod =
  \max\bigl\{
  \uppintom, \uppextom, \upplinkom\bigr\},
  \label{ubom}
\end{align}
where\footnote{we admit infinite values for the considered functions,
  e.g. $\lowintom(q,0) = -\infty$.}
\[
\lowintom(q,\omega) = q'-\frac{2q}{1-\cos\omega},\qquad
\lowextom(q,\omega) = q-Q',
\]
\[
\uppintom(q,\omega) = p' - q,
\]
\[
\uppextom(q,\omega) = \min\Bigl\{ \frac{2q}{1-\cos\omega} -
\frac{p'}{1-\hat{\xi}'_*},\ \frac{2q}{1+\cos\omega} - q'\Bigr\},
\]
with
\[
\hat{\xi}'_* = \min\{ \xi_*',e'\},
\]
where
\[
\xi'_*(q,\omega) = \frac{4q\cos\omega}
  {p'\sin^2\omega + \sqrt{p'^2\sin^4\omega+16q^2\cos^2\omega}},
\]
and
\begin{equation}
\upplinkom(q,\omega) = \min\left\{ Q' -
\frac{q(1+\hat{e}_*)}{1+\hat{e}_*\cos\omega},
\frac{2q}{1-\cos\omega}-q' \right\},
\label{def:upplinkom}
\end{equation}
with
\[
\hat{e}_* = \max\bigl\{0,\min\{e_*,1\}\bigr\},
\]
where
\[
e_*(q,\omega) =
\frac{2(p'-q(1-e'^2))}{q(1-e'^2) + \sqrt{q^2(1-e'^2)^2 +
    4p'\cos^2\omega(p'-q(1-e'^2))}}.
\]
\label{prop:deltanod_qom_gen}
\end{proposition}

\begin{proof}

We prove some preliminary facts.
\begin{lemma}
The following properties hold:
\begin{itemize}
\item[i)] for each $(q,\omega)\in{\cal D}_2$ and $(e,\omega')\in{\cal D}_1$
we have
  \begin{eqnarray}
  &&\deltaint(q,e,\omega,\omega') \geq \deltaint(q,1,\omega,\pi) = \lowintom(q,\omega),\label{b1om}\\ 
  &&\deltaext(q,e,\omega,\omega') \geq  \deltaext(q,0,\omega,\pi) = \lowextom(q,\omega),\label{b2om} 
  \end{eqnarray}
therefore, given $(q,\omega)\in{\cal D}_2$, we have internal
(resp. external) nodes for each $(e,\omega')\in{\cal D}_1$ if and only if
$\lowintom(q,\omega)>0$ (resp. $\lowextom(q,\omega)>0$);

  
\item[ii)] if $(q,\omega)$ is such that $\lowintom(q,\omega)\leq 0$
  and $\lowextom(q,\omega)\leq 0$, then there exists
  $(e,\omega')\in{\cal D}_1$ such that $\dnod^+\dnod^- = 0$,
  i.e. there exists $(e,\omega')$ corresponding to a crossing
  configuration.
\end{itemize}
\label{lem:omegabounds}
\end{lemma}
\begin{proof}
  We prove the bounds \eqref{b1om}, \eqref{b2om} by observing that 
  for each $(q,\omega)\in{\cal D}_2$ and $(e,\omega')\in{\cal D}_1$ we have
  \[\begin{split}
  \deltaint&\geq \min\biggl\{\min_{\omega'\in[0,\pi]}\rpiup - \max_{e\in[0,1]}\rpiu,
  \min_{\omega'\in[0,\pi]}\rmenop - \max_{e\in[0,1]}\rmeno\biggr\}\cr
  &=\min\{\rpiup|_{\omega'=0} - \rpiu|_{e=1}, \rmenop|_{\omega'=\pi} - \rmeno|_{e=1}\} =
  q' - \frac{2q}{1-\cos\omega}\cr
  \end{split}\]
  and
  \[\begin{split}
  \deltaext&\geq \min\biggl\{\min_{e\in[0,1]}\rpiu - \max_{\omega'\in[0,\pi]}\rpiup,
  \min_{e\in[0,1]}\rmeno - \max_{\omega'\in[0,\pi]}\rmenop\biggr\}\cr
  &=
  \min\{\rpiu|_{e=0} - \rpiup|_{\omega'=\pi},
  \rmeno|_{e=0} - \rmenop|_{\omega'=0}\} = q - Q'.\cr
  \end{split}\]
  We conclude the proof of i) using
  properties {\em a)}, {\em b)} in Lemma~\ref{lem:delta}.
To prove {\em ii)} we note that
\[
\deltaint(q,0,\omega,\pi/2) = p'-q, \qquad
\deltaext(q,0,\omega,\pi/2) = q-p'
\]
for each $(q,\omega)\in{\cal D}_2$.  Therefore, either they are both
zero and there is a crossing for $(e,\omega')=(0,\pi/2)$, or they are
different from zero and opposite and, since we are assuming that
$\lowintom,\lowextom\leq 0$ at $(q,\omega)$, by continuity there
exists $(e,\omega')\in{\cal D}_1$ corresponding to a crossing configuration.

\end{proof}

\noindent We continue the proof of Proposition~\ref{prop:deltanod_qom_gen}.

{\em Lower bound:} we prove relation \eqref{lbom} by observing that,
by {\em i)} of Lemma~\ref{lem:omegabounds}, if $\lowintom(q,\omega)>0$
we can have only internal nodes for each $(e,\omega')\in{\cal D}_1$.
Therefore $\min_{(e,\omega')\in{\cal D}_1}\deltanod(q,\omega) =
\min_{(e,\omega')\in{\cal D}_1}\deltaint(q,\omega) =
\lowintom(q,\omega)$ and $\deltaext(q,e,\omega,\omega')$,
$\deltalink(q,e,\omega,\omega') <0$ for each $(e,\omega')\in{\cal
  D}_1$. In particular we have $\lowextom(q,\omega)<0$.
In a similar way, if $\lowextom(q,\omega)>0$ we can have only external
nodes for each $(e,\omega')\in{\cal D}_1$. Therefore $\min_{(e,\omega')\in{\cal
    D}_1}\deltanod(q,\omega) = \min_{(e,\omega')\in{\cal D}_1}\deltaext(q,\omega) =
\lowextom(q,\omega)$ and $\deltaint(q,e,\omega,\omega')$,
$\deltalink(q,e,\omega,\omega') <0$ for each $(e,\omega')\in{\cal
  D}_1$. In particular we have $\lowintom(q,\omega)<0$.
Finally, if $\lowintom(q,\omega)\leq 0$ and $\lowextom(q,\omega)\leq
0$, by {\em ii)} of Lemma~\ref{lem:omegabounds} there exists
$(e,\omega')\in{\cal D}_1$ corresponding to a crossing configuration,
therefore $\min_{(e,\omega')\in{\cal D}_1}\deltanod(q,\omega)=0$.
%
The previous discussion yields relation \eqref{lbom}.

      
{\em Upper bound:}
by Lemma~\ref{der_om_omp_e} both $\dnod^+$ and $\dnod^-$ are
non-increasing functions of $e$, therefore also $\deltaint$ is, and we
have
\[
\deltaint(q,e,\omega,\omega')\leq \deltaint(q,0,\omega,\omega')
\]
for each $(q,\omega)\in{\cal D}_2$ and $(e,\omega')\in{\cal D}_1$.
By the same lemma, $\dnod^+$ is non-decreasing
with $\omega'$, while $\dnod^-$ is non-increasing, whatever the value
of $e$.
Moreover, for $e=0$ we have $\dnod^+=\dnod^-$ if and only if
\[
\frac{p'\xi'}{1-\xi'^2} = 0
\]
with $\xi'=e'\cos\omega'$, that is for $\omega'={\pi}/{2}$.  We
conclude that for each $(q,\omega)\in{\cal D}_2$ the maximal value of
$\deltaint$ over ${\cal D}_1$ is
\[
\uppintom(q,\omega) = \deltaint(q,0,\omega,\pi/2) = p' - q.
\]
We also observe that by
Lemma~\ref{der_om_omp_e} we have
\[
\deltaext(q,e,\omega,\omega')\leq \deltaext(q,1,\omega,\omega').
\]
Moreover, by the same lemma,
$-\dnod^+$ is non-increasing with $\omega'$ while $-\dnod^-$ is
non-decreasing, whatever the value of $e$.
Let us set
\[
\begin{split}
  &D_+(0) = \left.\dnod^+\right|_{e=1,\omega'=0} = q' - \frac{2q}{1+\cos\omega},
  \\
  &D_-(0) = \left.\dnod^-\right|_{e=1,\omega'=0} = Q' - \frac{2q}{1-\cos\omega},
  \\
  &D_+(\pi) = \left.\dnod^+\right|_{e=1,\omega'=\pi} = Q' - \frac{2q}{1+\cos\omega},
  \\
  &D_-(\pi) = \left.\dnod^-\right|_{e=1,\omega'=\pi} = q' - \frac{2q}{1-\cos\omega}.
  \\
\end{split}
\]
\begin{figure}[t!]
\begin{center}
  \begin{tikzpicture}
  \coordinate (O) at (0,0);
  \coordinate (O1) at (0+5,0);
  \coordinate (O2) at (0+5*2,0);
  \coordinate (x) at (3.5,0);
  \coordinate (x1) at (3.5+5,0);
  \coordinate (x2) at (3.5+5*2,0);
  \coordinate (y) at (0,2.5);
  \coordinate (y1) at (0+5,2.5);
  \coordinate (y2) at (0+5*2,2.5);
  \draw [-latex] (O) -- (x) node [below] {$\omega'$};
  \draw [-latex] (O) -- (y);
  \draw [-latex] (O1) -- (x1) node [below] {$\omega'$};
  \draw [-latex] (O1) -- (y1);
  \draw [-latex] (O2) -- (x2) node [below] {$\omega'$};
  \draw [-latex] (O2) -- (y2);
  \coordinate (Ap) at (0,0.5); \coordinate (Bp) at (3,1);
  \coordinate (Am) at (0,2); \coordinate (Bm) at (3,1.5);
  \draw[line width=0.2mm, color=black!70!white] (Am)--(Bm);
  \draw[line width=0.2mm, color=black!30!white] (Ap)--(Bp);
  \draw [dashed] (3,0)--(3,2.3);
  \node at (0,0) [below] {$0$};
  \node at (3,0) [below] {$\pi$};
  \draw (1.3,0.5) node {$-\dnod^-$};
  \draw (0.8,1.6) node {$-\dnod^+$};
  \fill (3,1) circle (0.4mm);
  \fill (3,1.5) circle (0.4mm);
  \draw (3.,0.7) node {\small $-D_-(\pi)$};
  \draw (3.,1.8) node {\small $-D_+(\pi)$};

  \coordinate (A1p) at (5,0.5+0.5); \coordinate (B1p) at (3+5,1+0.5);
  \coordinate (A1m) at (5,2-0.5); \coordinate (B1m) at (3+5,1.5-0.5);
  \draw[line width=0.2mm, color=black!70!white] (A1m)--(B1m);
  \draw[line width=0.2mm, color=black!30!white] (A1p)--(B1p);
  \draw [dashed] (3+5,0)--(3+5,2.3);
  \node at (0+5,0) [below] {$0$};
  \node at (3+5,0) [below] {$\pi$};
  \draw [dashed] (1.5+5,0)--(1.5+5,1.3);
  \node at (1.5+5,0) [below] {$\omega_*'$};
  \draw (1.8+5,0.8) node {$-\dnod^+$};
  \draw (1.8+5,1.7) node {$-\dnod^-$};
  \fill (0+5,1) circle (0.4mm);
  \fill (0+5,1.5) circle (0.4mm);
  \draw (-0.3+5,0.7) node {\small $-D_-(0)$};
  \draw (-0.3+5,1.8) node {\small $-D_+(0)$};
  \fill (3+5,1) circle (0.4mm);
  \fill (3+5,1.5) circle (0.4mm);
  \draw (3.2+5,0.7) node {\small $-D_+(\pi)$};
  \draw (3.2+5,1.8) node {\small $-D_-(\pi)$};

  \coordinate (A2p) at (5*2,0.5+1); \coordinate (B2p) at (3+5*2,1+1);
  \coordinate (A2m) at (5*2,2-1); \coordinate (B2m) at (3+5*2,1.5-1);
  \draw[line width=0.2mm, color=black!70!white] (A2m)--(B2m);
  \draw[line width=0.2mm, color=black!30!white] (A2p)--(B2p);
  \draw [dashed] (3+5*2,0)--(3+5*2,2.3);
  \node at (0+5*2,0) [below] {$0$};
  \node at (3+5*2,0) [below] {$\pi$};
  \draw (2+5*2,1.6) node {$-\dnod^-$};
  \draw (1.5+5*2,0.4) node {$-\dnod^+$};
  \fill (0+5*2,1) circle (0.4mm);
  \fill (0+5*2,1.5) circle (0.4mm);
  \draw (-0.1+5*2,0.7) node {\small $-D_+(0)$};
  \draw (-0.1+5*2,1.8) node {\small $-D_-(0)$};
  
  \node at (1.5,2.5) {a)};
  \node at (1.5+5,2.5) {b)};
  \node at (1.5+5*2,2.5) {c)};
  
  \end{tikzpicture}
  \caption{Possible behavior of $-\dnod^+$ and $-\dnod^-$ as functions
    of $\omega'$.}
  \label{fig:upperextqom}
\end{center}
\end{figure}
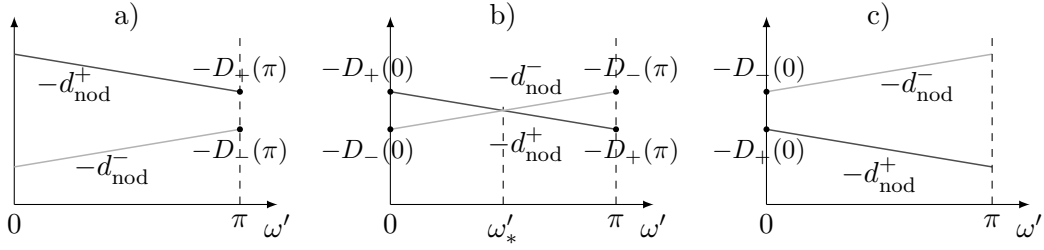
We consider the three cases depicted in Figure~\ref{fig:upperextqom}:
\begin{itemize}
\item[a)] $-D_+(\pi) > -D_-(\pi)$, which corresponds to
  \[
  Q'-q' < -4q\frac{\cos\omega}{\sin^2\omega};
  \]
\item[b)] $-D_+(0) \geq -D_-(0)$ and
  $-D_+(\pi) \leq -D_-(\pi)$, which correspond to
  \[
  Q'-q' \geq 4q\frac{\cos\omega}{\sin^2\omega};
  \]
\item[c)] $-D_+(0) < -D_-(0)$, which corresponds to
  \[
  Q'-q' < 4q\frac{\cos\omega}{\sin^2\omega}.
  \]
\end{itemize}
Indeed case a) is impossible, because $Q'\geq q'$ and $\omega\in[0,\pi/2]$.
In case b) the maximal value of $\deltaext$ is attained for $e=1$ and
$\omega'$ such that $\dnod^+=\dnod^-$, that is when
\begin{equation}
\frac{p'\xi'}{1-\xi'^2} =
\frac{2q\cos\omega}{\sin^2\omega},
  \label{dnodeq_e1}
\end{equation}
with $\xi'=e'\cos\omega'$.
The solution of \eqref{dnodeq_e1} is\footnote{we discard the solution
  giving a value of $\xi'$ which is $<-1$.}
\begin{equation}
\xi_*'(q,\omega) = 
\frac{4q\cos\omega}{p'\sin^2\omega + \sqrt{p'^2\sin^4\omega+16q^2\cos^2\omega}},
\label{xipstar}
\end{equation}
for which we find
\[
\omega_*' = \arccos\frac{\xi_*'}{e'},
\]
if $\xi_*'\leq e'$.
In case c) equation $\dnod^+=\dnod^-$ has no real solution for
$\omega'$, that is $\xi_*'>e'$, and the maximal value of $\deltaext$ is
given by $-\dnod^+$ with $e=1$ and $\omega'=0$.

We introduce the cut-off
\begin{equation}
\hat{\xi}_*' = \min\{\xi_*',e'\}
\label{hatxistarp}
\end{equation}
and define
\[
  \hat{\omega}_*'= \arccos(\hat{\xi}_*'/e'). 
\]
From the previous discussion we obtain that the maximal value of
$\deltaext$ over ${\cal D}_1$ is given by
\[
\begin{split}
\uppextom(q,\omega) &= \min\left\{
\left.-\dnod^-\right|_{e=1,\omega'=\hat{\omega}'_*},
\left.-\dnod^+\right|_{e=1,\omega'=0}
 \right\}\cr
&= \min\Bigl\{\frac{2q}{1-\cos\omega} -
\frac{p'}{1-\hat{\xi}_*'}, \frac{2q}{1+\cos\omega}-q' \Bigr\}.\cr
\end{split}
\]

Finally, we consider the function $\deltalink$ and examine
$\deltalink^{(i)}$ and $\deltalink^{(ii)}$ separately.  We can not
select {\em a priori} a value of the eccentricity $e$ that maximizes
$\deltalink$, as we did before. However, we can do this for $\omega'$,
in fact by Lemma~\ref{der_om_omp_e} both $-\dnod^+$ and $\dnod^-$ are
non-increasing functions of $\omega'$, therefore also
$\deltalink^{(i)}$ is. Therefore, for each fixed value of $q,e,\omega$
the maximal value of $\deltalink^{(i)}$ is attained for $\omega'=0$.
By a similar argument we obtain that for each fixed value
of $q,e,\omega$ the maximal value of $\deltalink^{(ii)}$ is attained for
$\omega'=\pi$.

We observe that $\dnod^+ = -\dnod^-$ if and only if
\begin{equation}
q(1+e)(1-\xi'^2) = p'(1-e^2\cos^2\omega)
\label{dnp_eq_dnm_linked}
\end{equation}
where $\xi' = e'\cos\omega'$.
Both for $\omega'=0$ and for $\omega'=\pi$ equation
\eqref{dnp_eq_dnm_linked} becomes
\[
q(1+e)(1-e'^2) = p'(1-e^2\cos^2\omega),
\]
that gives the eccentricity\footnote{we discard the solution giving
  a negative value of $e$.}
\begin{equation}
e_*(q,\omega) = \frac{2(p'-q(1-e'^2))}
{q(1-e'^2) + \sqrt{q^2(1-e'^2)^2 + 4p'\cos^2\omega(p'-q(1-e'^2))}}
\label{estar}
\end{equation}
We observe that $e_*$ can attain negative values, or values larger than 1.
For this reason we introduce the cut-off
\begin{equation}
  \hat{e}_*(q,\omega)=\max\{0,\min\{e_*(q,\omega),1\}\}.
  \label{hatestar}
\end{equation}

\noindent By Lemma~\ref{der_om_omp_e}, $-\dnod^+$ is non-decreasing
with $e$, while $\dnod^-$ is non-increasing, whatever the value of
$\omega'$.
Let us set
\[
\begin{split}
  &D^{(i)}_+(0) = \left.\dnod^+\right|_{\omega'=0,e=0} = q' - q,
  \\
  &D^{(i)}_-(0) = \left.\dnod^-\right|_{\omega'=0,e=0} = Q' - q,
  \\
  &D^{(i)}_+(1) = \left.\dnod^+\right|_{\omega'=0,e=1} = q' - \frac{2q}{1+\cos\omega},
  \\
  &D^{(i)}_-(1) = \left.\dnod^-\right|_{\omega'=0,e=1} = Q' - \frac{2q}{1-\cos\omega}.
  \\
\end{split}
\]

\begin{figure}[h!]
\begin{center}
  \begin{tikzpicture}
  \coordinate (O) at (0,0);
  \coordinate (O1) at (0+5,0);
  \coordinate (O2) at (0+5*2,0);
  \coordinate (x) at (3.5,0);
  \coordinate (x1) at (3.5+5,0);
  \coordinate (x2) at (3.5+5*2,0);
  \coordinate (y) at (0,2.5);
  \coordinate (y1) at (0+5,2.5);
  \coordinate (y2) at (0+5*2,2.5);
  \draw [-latex] (O) -- (x) node [below] {$e$};
  \draw [-latex] (O) -- (y);
  \node [below] at (O) {$0$};
  \draw [-latex] (O1) -- (x1) node [below] {$e$};
  \draw [-latex] (O1) -- (y1);
  \node [below] at (O1) {$0$};
  \draw [-latex] (O2) -- (x2) node [below] {$e$};
  \draw [-latex] (O2) -- (y2);
  \node [below] at (O2) {$0$};

  \coordinate (Ap) at (5*2,0.5); \coordinate (Bp) at (3+5*2,1);
  \coordinate (Am) at (5*2,2); \coordinate (Bm) at (3+5*2,1.5);
  \draw[line width=0.2mm, color=black!30!white] (Am)--(Bm);
  \draw[line width=0.2mm, color=black!70!white] (Ap)--(Bp);
  \draw [dashed] (3+5*2,0)--(3+5*2,2.3);
  \node at (3+5*2,0) [below] {$1$};
  \draw (1.8+5*2,0.5) node {$-\dnod^+$};
  \draw (1.5+5*2,1.5) node {$\dnod^-$};
  \fill (3+5*2,1) circle (0.4mm);
  \fill (3+5*2,1.5) circle (0.4mm);
  \draw (3.1+5*2,0.7) node {\small $-D^{(i)}_+(1)$};
  \draw (3.2+5*2,1.8) node {\small $D^{(i)}_-(1)$};

  \coordinate (A1p) at (5,0.5+0.5); \coordinate (B1p) at (3+5,1+0.5);
  \coordinate (A1m) at (5,2-0.5); \coordinate (B1m) at (3+5,1.5-0.5);
  \draw[line width=0.2mm, color=black!30!white] (A1m)--(B1m);
  \draw[line width=0.2mm, color=black!70!white] (A1p)--(B1p);
  \draw [dashed] (3+5,0)--(3+5,2.3);
  \node at (3+5,0) [below] {$1$};
  \draw [dashed] (1.5+5,0)--(1.5+5,1.3);
  \node at (1.5+5,0) [below] {$e_*$};
  \draw (2.+5,0.8+0.9) node {$-\dnod^+$};
  \draw (2.1+5,1.6-0.8) node {$\dnod^-$};
  \fill (0+5,1) circle (0.4mm);
  \fill (0+5,1.5) circle (0.4mm);
  \draw (-0.1+5,0.7) node {\small $-D^{(i)}_+(0)$};
  \draw (-0.+5,1.8) node {\small $D^{(i)}_-(0)$};
  \fill (3+5,1) circle (0.4mm);
  \fill (3+5,1.5) circle (0.4mm);
  \draw (3.3+5,0.7) node {\small $D^{(i)}_-(1)$};
  \draw (3.2+5,1.8) node {\small $-D^{(i)}_+(1)$};

  \coordinate (A2p) at (0,0.5+1); \coordinate (B2p) at (3,1+1);
  \coordinate (A2m) at (0,2-1); \coordinate (B2m) at (3,1.5-1);
  \draw[line width=0.2mm, color=black!30!white] (A2m)--(B2m);
  \draw[line width=0.2mm, color=black!70!white] (A2p)--(B2p);
  \draw [dashed] (3,0)--(3,2.3);
  \node at (3,0) [below] {$1$};
  \draw (2,1.6) node {$-\dnod^+$};
  \draw (1,0.6) node {$\dnod^-$};
  \fill (0,1) circle (0.4mm);
  \fill (0,1.5) circle (0.4mm);
  \draw (0,0.7) node {\small $D^{(i)}_-(0)$};
  \draw (-0.1,1.8) node {\small $-D^{(i)}_+(0)$};

  \node at (1.5,2.5) {a)};
  \node at (1.5+5,2.5) {b)};
  \node at (1.5+5*2,2.5) {c)};
  
  \end{tikzpicture}
  \caption{Possible behavior of $-\dnod^+$ and $\dnod^-$ as functions
    of $e$.}
  \label{fig:upperlinkiqom}
\end{center}
\end{figure}

We consider the three cases
\begin{itemize}
\item[a)] $-D^{(i)}_+(0)>D^{(i)}_-(0)$, which corresponds to
  \[
  Q'+q' < 2q;
  \]
\item[b)] $-D^{(i)}_+(0)\leq D^{(i)}_-(0)$ and
  $-D^{(i)}_+(1)\geq D^{(i)}_-(1)$, which correspond to
  \[
  2q\leq Q'+q' \leq \frac{4q}{\sin^2\omega};
  \]
\item[c)] $-D^{(i)}_+(1)<D^{(i)}_-(1)$, which corresponds to
  \[
  Q'+q' > \frac{4q}{\sin^2\omega}.
  \]
\end{itemize}
Therefore, the maximal value of $\deltalink^{(i)}$ over ${\cal D}_1$ is given by
\[
\min\left\{
\left.\dnod^-\right|_{\omega'=0,e=\hat{e}_*},
\left.-\dnod^+\right|_{\omega'=0,e=1}
\right\}
=\min\left\{
Q' - \frac{q(1+\hat{e}_*)}{1-\hat{e}_*\cos\omega},
\frac{2q}{1+\cos\omega}-q'  \right\}.
\]

To compute a bound for $\deltalink^{(ii)}$ we observe that $\dnod^+$
is non-increasing with $e$, while $-\dnod^-$ is non-decreasing.  Let
us set
\[
\begin{split}
  &D^{(ii)}_+(0) = \left.\dnod^+\right|_{\omega'=\pi,e=0} = Q' - q,
  \\
  &D^{(ii)}_-(0) = \left.\dnod^-\right|_{\omega'=\pi,e=0} = q' - q,
  \\
  &D^{(ii)}_+(1) = \left.\dnod^+\right|_{\omega'=\pi,e=1} = Q' - \frac{2q}{1+\cos\omega},
  \\
  &D^{(ii)}_-(1) = \left.\dnod^-\right|_{\omega'=\pi,e=1} = q' - \frac{2q}{1-\cos\omega}.
  \\
\end{split}
\]

\begin{figure}[h!]
\begin{center}
  \begin{tikzpicture}
  \coordinate (O) at (0,0);
  \coordinate (O1) at (0+5,0);
  \coordinate (O2) at (0+5*2,0);
  \coordinate (x) at (3.5,0);
  \coordinate (x1) at (3.5+5,0);
  \coordinate (x2) at (3.5+5*2,0);
  \coordinate (y) at (0,2.5);
  \coordinate (y1) at (0+5,2.5);
  \coordinate (y2) at (0+5*2,2.5);
  \draw [-latex] (O) -- (x) node [below] {$e$};
  \draw [-latex] (O) -- (y);
  \node [below] at (O) {$0$};
  \draw [-latex] (O1) -- (x1) node [below] {$e$};
  \draw [-latex] (O1) -- (y1);
  \node [below] at (O1) {$0$};
  \draw [-latex] (O2) -- (x2) node [below] {$e$};
  \draw [-latex] (O2) -- (y2);
  \node [below] at (O2) {$0$};
  \coordinate (Ap) at (0,0.5+1); \coordinate (Bp) at (3,1+1);
  \coordinate (Am) at (0,2-1); \coordinate (Bm) at (3,1.5-1);
  \draw[line width=0.2mm, color=black!70!white] (Am)--(Bm);
  \draw[line width=0.2mm, color=black!30!white] (Ap)--(Bp);
  \draw [dashed] (3,0)--(3,2.3);
  \node at (3,0) [below] {$1$};
  \draw (2.,1.6) node {$-\dnod^-$};
  \draw (1.6,0.5) node {$\dnod^+$};
  \fill (0,1) circle (0.4mm);
  \fill (0,1.5) circle (0.4mm);
  \draw (0,0.7) node {\small $D^{(ii)}_+(0)$};
  \draw (-0.1,1.8) node {\small $-D^{(ii)}_-(0)$};

  \coordinate (A1p) at (5,0.5+0.5); \coordinate (B1p) at (3+5,1+0.5);
  \coordinate (A1m) at (5,2-0.5); \coordinate (B1m) at (3+5,1.5-0.5);
  \draw[line width=0.2mm, color=black!70!white] (A1m)--(B1m);
  \draw[line width=0.2mm, color=black!30!white] (A1p)--(B1p);
  \draw [dashed] (3+5,0)--(3+5,2.3);
  \node at (3+5,0) [below] {$1$};
  \draw [dashed] (1.5+5,0)--(1.5+5,1.3);
  \node at (1.6+5,0) [below] {$e_*$};
  \draw (1.+5,0.9+0.8) node {$\dnod^+$};
  \draw (0.9+5,1.6-0.7) node {$-\dnod^-$};
  \fill (0+5,1) circle (0.4mm);
  \fill (0+5,1.5) circle (0.4mm);
  \draw (-0.3+5,0.7) node {\small $-D^{(ii)}_-(0)$};
  \draw (-0.2+5,1.8) node {\small $D^{(ii)}_+(0)$};
  \fill (3+5,1) circle (0.4mm);
  \fill (3+5,1.5) circle (0.4mm);
  \draw (3.3+5,0.7) node {\small $D^{(ii)}_+(1)$};
  \draw (3.2+5,1.8) node {\small $-D^{(ii)}_-(1)$};

  \coordinate (A2p) at (5*2,0.5); \coordinate (B2p) at (3+5*2,1);
  \coordinate (A2m) at (5*2,2); \coordinate (B2m) at (3+5*2,1.5);
  \draw[line width=0.2mm, color=black!70!white] (A2m)--(B2m);
  \draw[line width=0.2mm, color=black!30!white] (A2p)--(B2p);
  \draw [dashed] (3+5*2,0)--(3+5*2,2.3);
  \node at (3+5*2,0) [below] {$1$};
  \draw (1.5+5*2,0.5) node {$-\dnod^-$};
  \draw (1+5*2,1.6) node {$\dnod^+$};
  \fill (3+5*2,1) circle (0.4mm);
  \fill (3+5*2,1.5) circle (0.4mm);
  \draw (3.1+5*2,0.7) node {\small $-D^{(ii)}_-(1)$};
  \draw (3.2+5*2,1.8) node {\small $D^{(ii)}_+(1)$};

  \node at (1.5,2.5) {a)};
  \node at (1.5+5,2.5) {b)};
  \node at (1.5+5*2,2.5) {c)};
  
  \end{tikzpicture}
  \caption{Possible behavior of $\dnod^+$ and $-\dnod^-$ as functions
    of $e$.}
  \label{fig:upperlinkiiqom}
\end{center}
\end{figure}

We consider the three cases
\begin{itemize}
\item[a)] $D^{(ii)}_+(0)<-D^{(ii)}_-(0)$, which corresponds to
  \[
  Q'+q'<2q;
  \]
\item[b)] $D^{(ii)}_+(0)\geq -D^{(ii)}_-(0)$ and
  $D^{(ii)}_+(1)\leq -D^{(ii)}_-(1)$, which correspond to
  \[
  2q\leq Q'+q' \leq \frac{4q}{\sin^2\omega};
  \]
\item[c)] $D^{(ii)}_+(1)>-D^{(ii)}_-(1)$, which corresponds to
  \begin{equation}
  Q'+q' > \frac{4q}{\sin^2\omega}.
  \label{estargt1}
  \end{equation}
\end{itemize}
Therefore, the maximal value of $\deltalink^{(ii)}$ over ${\cal D}_1$ is given by
\[
\min\left\{ \left.\dnod^+\right|_{\omega'=\pi,e=\hat{e}_*},
\left.-\dnod^-\right|_{\omega'=\pi,e=1}
\right\}
=\min\left\{ Q' - \frac{q(1+\hat{e}_*)}{1+\hat{e}_*\cos\omega},
\frac{2q}{1-\cos\omega}-q'
\right\},
\]
where $\hat{e}_*$ is defined as in \eqref{hatestar}.
We conclude that
the maximal value of $\deltalink$ over ${\cal D}_1$ is given by
\small
\[
\begin{split}
  \upplinkom(q,\omega) &= \max\left\{ \min\Bigl\{
  Q' - \frac{q(1+\hat{e}_*)}{1-\hat{e}_*\cos\omega},
  \frac{2q}{1+\cos\omega}-q' 
  \Bigr\}, \ \min\Bigl\{
  Q' - \frac{q(1+\hat{e}_*)}{1+\hat{e}_*\cos\omega},
  \frac{2q}{1-\cos\omega}-q' \Bigr\} \right\}\\
  &=\min\left\{  Q' - \frac{q(1+\hat{e}_*)}{1+\hat{e}_*\cos\omega},
  \frac{2q}{1-\cos\omega}-q' \right\},
  \end{split}
\]
\normalsize 
where the last equality holds because
$\omega\in[0,{\pi}/{2}]$.
In particular, the maximal value is attained by $\deltalink^{(ii)}$.

We conclude the proof of relation \eqref{ubom} using \eqref{deltanod_max}
and the optimal bounds
\[
  \deltaint(q,e,\omega,\omega')\leq \uppintom(q,\omega),\quad
  \deltaext(q,e,\omega,\omega')\leq \uppextom(q,\omega),\quad
  \deltalink(q,e,\omega,\omega')\leq \upplinkom(q,\omega).
\]

\end{proof}

In Figure~\ref{maxdeltanodqom} we show the graphic of
$\max_{(e,\omega')\in{\cal D}_1}\deltanod(q,\omega)$ for different
values of $e'$, with $q'=1$.  Using Remark~\ref{rem:sym} we can extend
by symmetry the graphic of $\max_{(e,\omega')\in{\cal
    D}_1}\deltanod(q,\omega)$ to the set $(0,\qmax]\times[0,2\pi)$.

\begin{figure}[t!]
  \centerline{\includegraphics[width=7.5cm]{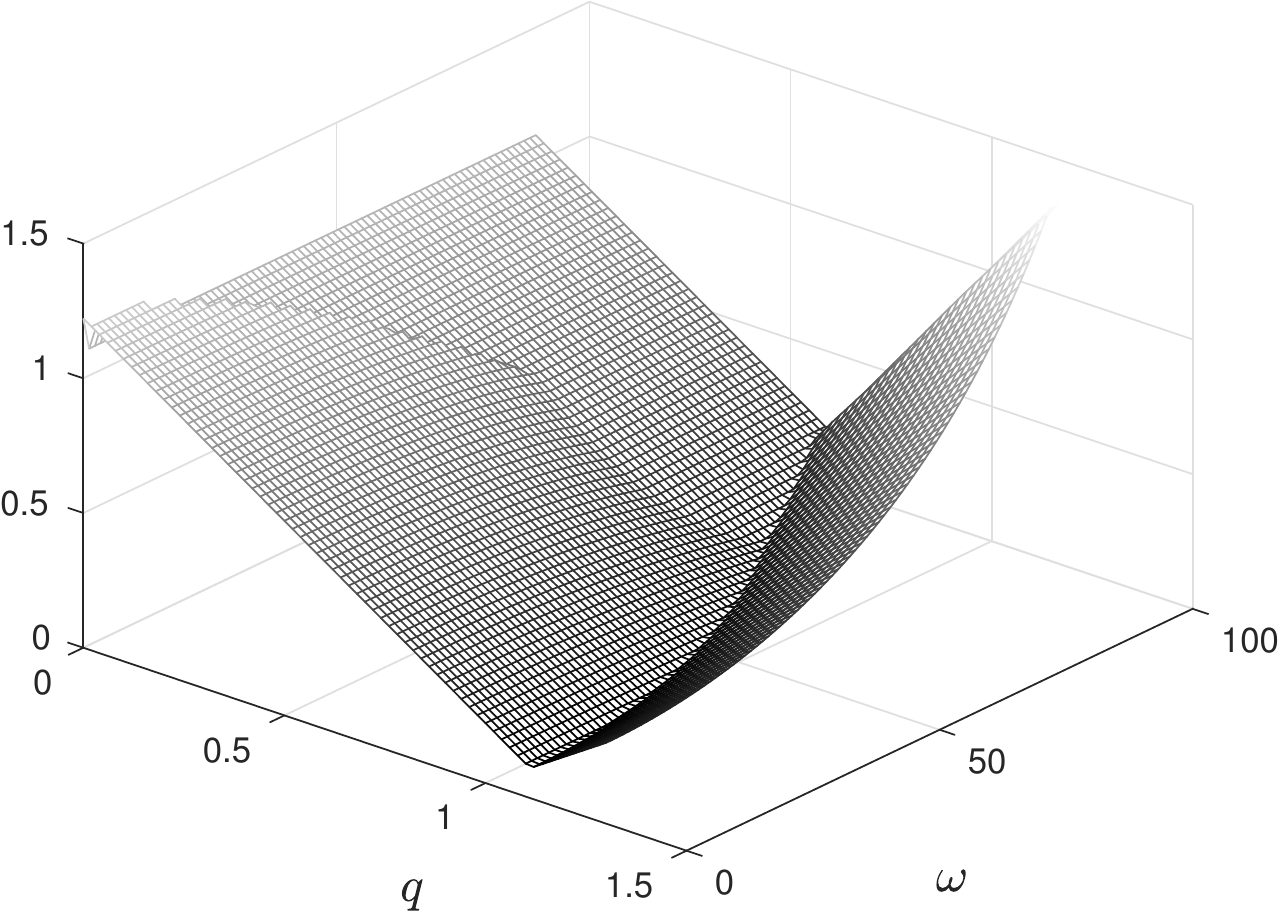}
    \hskip 0.2cm
    \includegraphics[width=7.5cm]{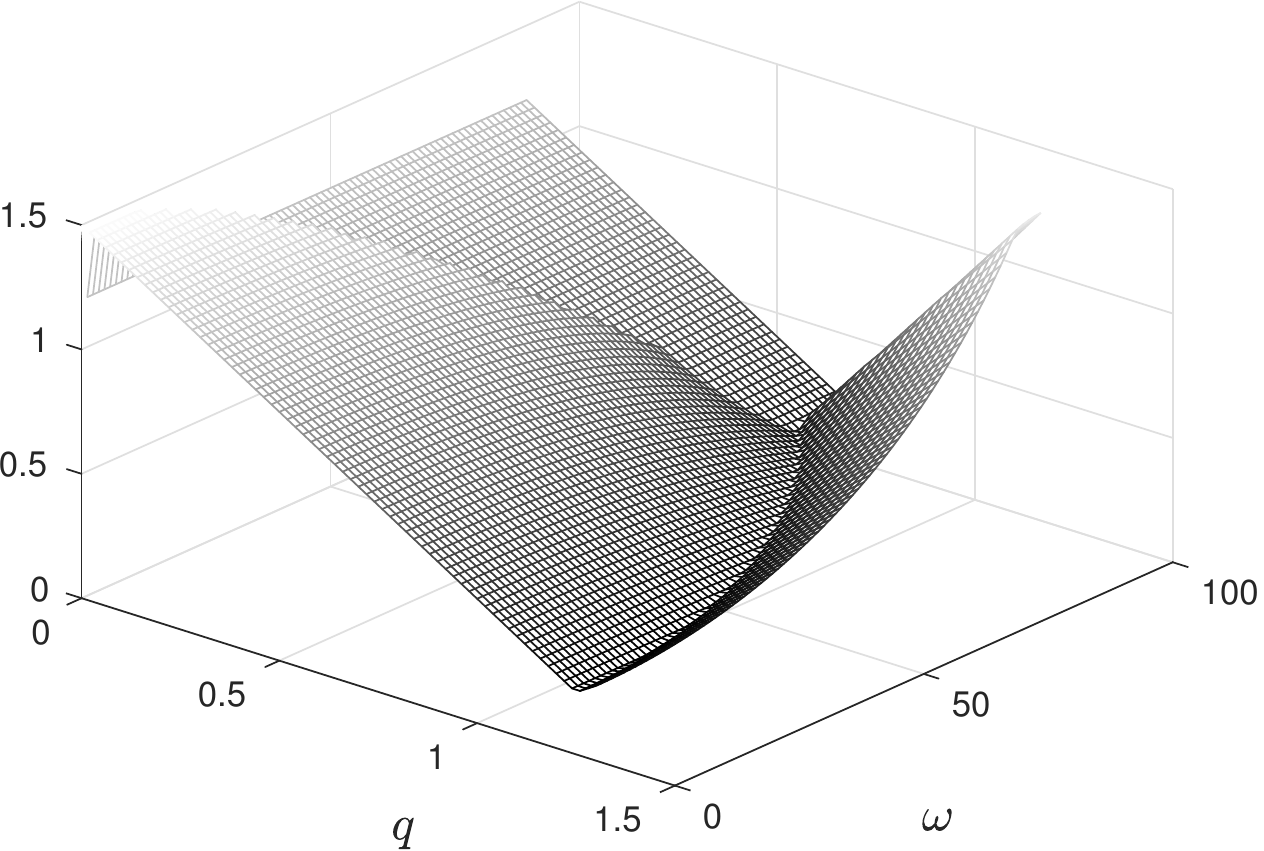}}
    \centerline{\includegraphics[width=7.5cm]{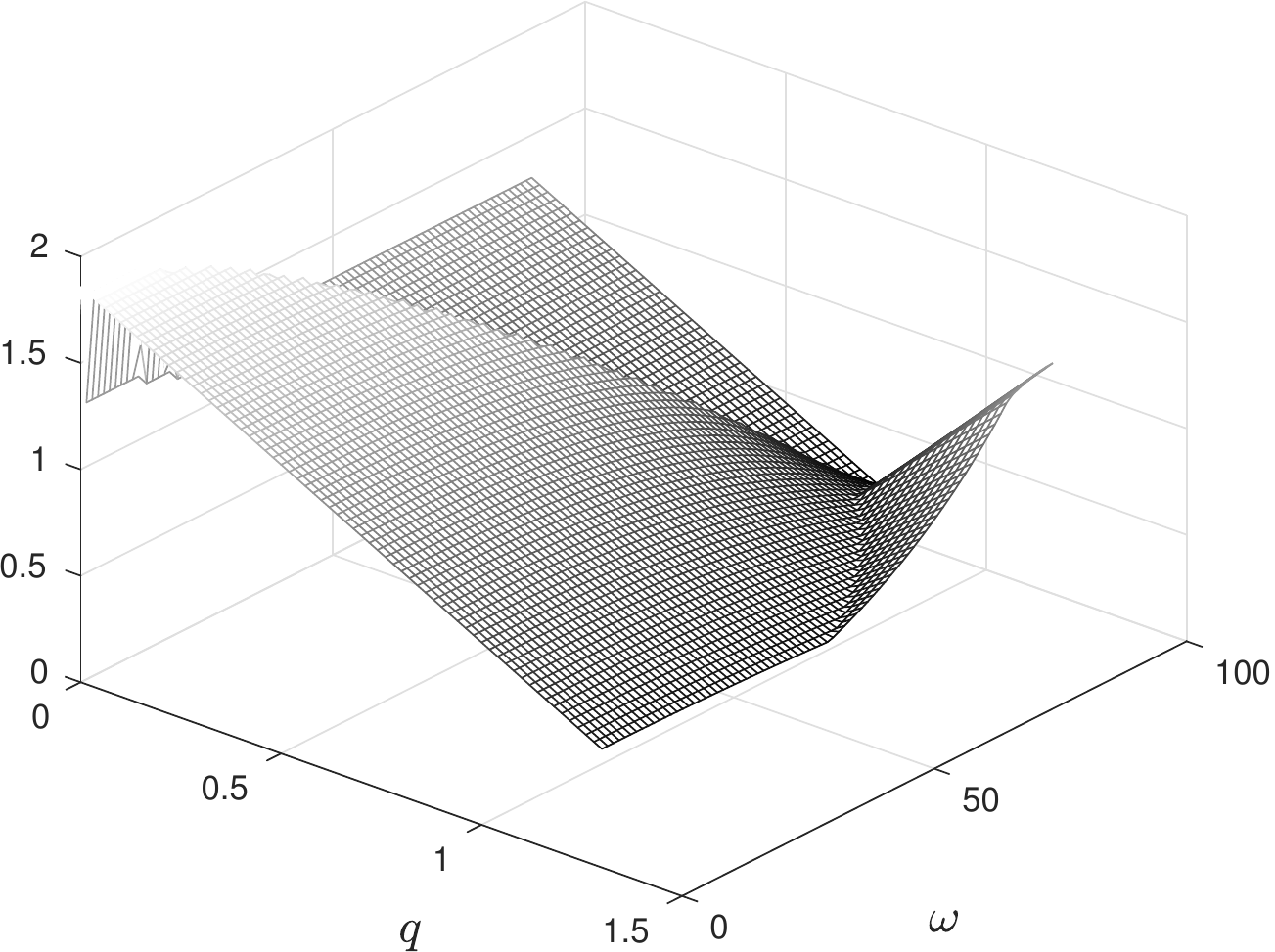}
    \hskip 0.2cm
    \includegraphics[width=7.5cm]{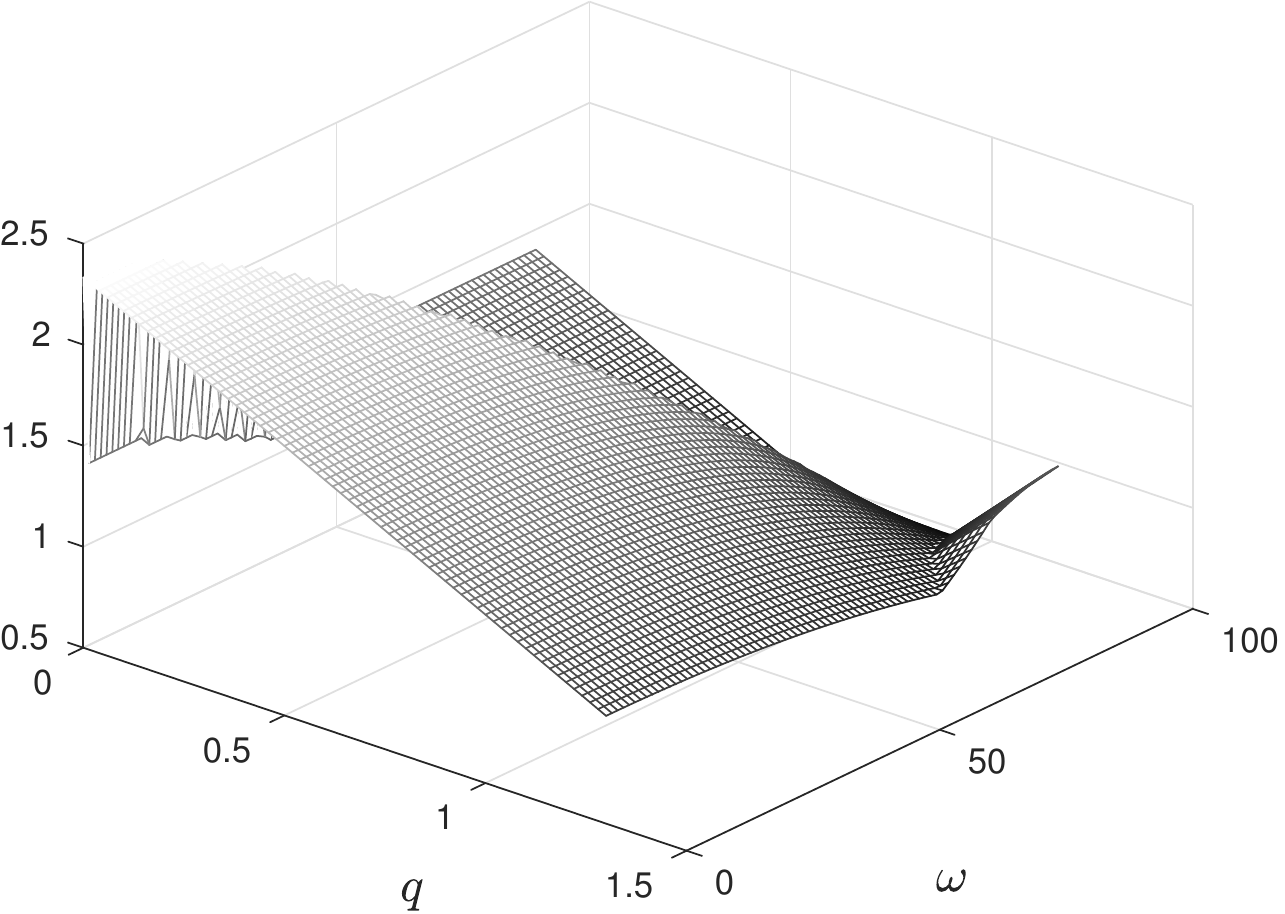}}
\caption{Graphics of $(q,\omega)\mapsto \max_{(e,\omega')\in{\cal D}_1}
  \deltanod(q,\omega)$ for $e'=0.1$ (top left), $e'= 0.2$ (top right),
  $e'=0.3$ (bottom left), $e'=0.4$ (bottom right). Here we set $q'=1$.}
\label{maxdeltanodqom}
\end{figure}

\begin{proposition}
The zero level curves of $\lowintom, \lowextom, \uppintom, \uppextom$
divide the plane $(q,\omega)$ into regions where different linking
configurations are allowed.  Moreover, $\uppextom(q,\omega)=0$ is a
piecewise smooth curve with only one component, a portion of which is
a vertical segment with $q={p'}/{2}$.
\label{prop:ILEC_qom}
\end{proposition}
\begin{figure}[h!]
  \begin{center}
    \includegraphics[width=10cm]{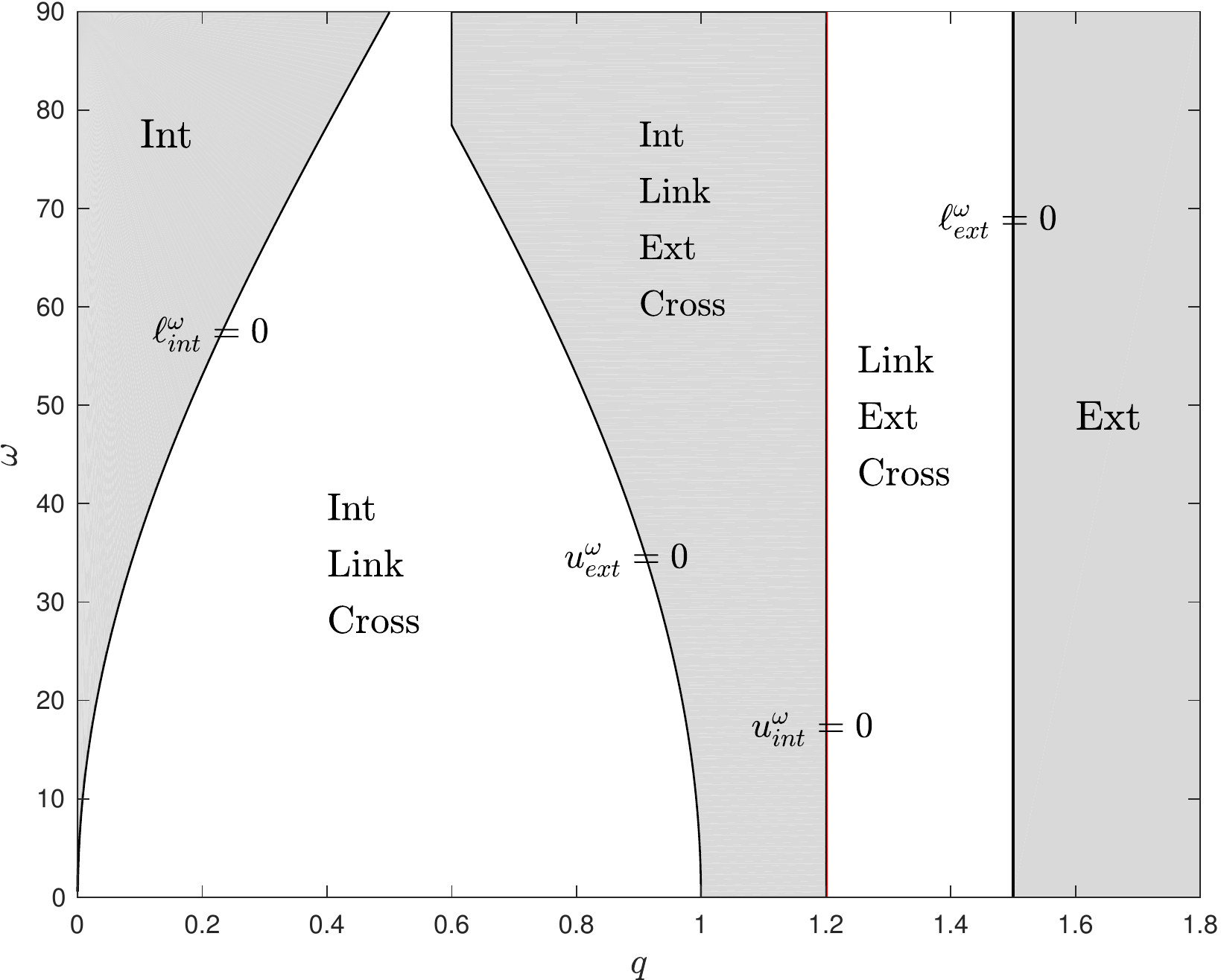}
  \end{center}
  \caption{Regions with different linking configurations in the plane
    $(q,\omega)$ for $q'=1$ and $e'=0.2$.}
  \label{fig:qom}
\end{figure}

\begin{proof}
  By Lemma~\ref{lem:omegabounds}, given $(q,\omega)\in{\cal D}_2$, we
  have internal nodes for each $(e,\omega')\in{\cal D}_1$ if and only
  if $\lowintom(q,\omega)>0$, therefore the region where only internal
  nodes are possible is delimited on the right by the curve
  $\lowintom(q,\omega)=0$.  In a similar way, the region with only
  external nodes is delimited on the left by $\lowextom(q,\omega)=0$.
  
  Moreover, we have internal nodes for some choice of $(e,\omega')$ if
  and only if $\uppintom(q,\omega) > 0$. In a similar way, we have
  external nodes (resp. linked orbits) for some choice of
  $(e,\omega')$ if and only if $\uppextom(q,\omega) > 0$
  (resp. $\upplinkom(q,\omega)> 0$).

  We prove the following result.
  \begin{lemma}
    The curve
    \[
    \upplinkom(q,\omega)= 0,
    \]
    delimiting the region where linked orbits are possible,
    has two connected components, and coincides with the curve
    \[
    \{\lowintom(q,\omega)= 0\}\cup \{\lowextom(q,\omega)= 0\}.
    \]
    \label{lem:ulinkom}
  \end{lemma}
  
  \begin{proof}
    If $(q,\omega)$ is such that
    \begin{equation}
      Q'-\frac{q(1+\hat{e}_*)}{1+\hat{e}_*\cos\omega}= 0,
      \label{firstcaseom}
    \end{equation}
    then
    \[
    \frac{2q}{1-\cos\omega} - q' > \frac{2q}{1+\cos\omega} - Q' =
    \frac{2q}{1+\cos\omega} -
    \frac{q(1+\hat{e}_*)}{1+\hat{e}_*\cos\omega} \geq 0,
    \]
    because $\frac{q(1+e)}{1+e\cos\omega}$ is increasing with $e$.
    We prove that \eqref{firstcaseom} is equivalent to
    \begin{equation}
      Q'-q=0
      \label{straightline}
    \end{equation}
    From relations \eqref{estar}, \eqref{hatestar} we deduce
    that $\hat{e}_*(q,\omega)=0$ if and only if $p'\leq q(1-e'^2)$, that is if
    \begin{equation}
      q \geq \frac{q'}{1-e'}.
      \label{qminstraight}
    \end{equation}
    Since $q=Q'$ fulfills \eqref{qminstraight}, then
    \eqref{straightline} implies \eqref{firstcaseom}.  To prove the
    converse first we observe that
    relation \eqref{firstcaseom} implies
    \begin{equation}
      \frac{1}{2}Q'\leq q\leq Q'.
      \label{qconfin}
    \end{equation}
    If $\hat{e}_*=0$, i.e. if $e_*\leq 0$, then relations
    \eqref{firstcaseom} and \eqref{straightline} are the
    same. Otherwise $e_*>0$. If $e_*\in(0,1]$ then $e_*$ is defined so
that it satisfies $\dnod^+=-\dnod^-$ with $\omega'=\pi$, therefore using
\eqref{firstcaseom} we have
    \begin{equation}
      Q'-\frac{q(1+e_*)}{1+e_*\cos\omega} =
      \frac{q(1+e_*)}{1-e_*\cos\omega} - q' = 0,
      \label{doppiaeq}
    \end{equation}
    from which we obtain
    \[
    e_* = \frac{q'-Q'}{(q'+Q')\cos\omega}<0,
    \]
    giving a contradiction.  If $e_*>1$ then \eqref{estargt1}
    holds. Therefore, either $\cos\omega=1$ and relations
    \eqref{firstcaseom} and \eqref{straightline} are the same, or we
    have $4q\leq 2Q'\sin^2\omega< 2Q'$, that contradicts
    \eqref{qconfin}.
    
\noindent We conclude that, in this case, $\upplinkom(q,\omega) = Q' -
q$, and the curve $\upplinkom(q,\omega)=0$ has a connected component
corresponding to $\lowextom(q,\omega)=0$.

\smallbreak
    On the other hand, if $(q,\omega)$ is such that
    \[
    \frac{2q}{1-\cos\omega} - q' =0,
    \]
    then
    \[
    Q'-\frac{q(1+\hat{e}_*)}{1+\hat{e}_*\cos\omega} >
    q'-\frac{q(1+\hat{e}_*)}{1+\hat{e}_*\cos\omega} =
    \frac{2q}{1-\cos\omega} -
    \frac{q(1+\hat{e}_*)}{1+\hat{e}_*\cos\omega} \geq 0.
    \]
    Therefore, in this case, $\upplinkom(q,\omega) =
    \frac{2q}{1-\cos\omega} - q'$, and the curve
    $\upplinkom(q,\omega)=0$ has another connected component
    coinciding with $\lowintom(q,\omega)=0$.
    
\end{proof}

Now we describe the shape of the curve $\uppextom(q,\omega)=0$.
First we observe that
\begin{equation}
  \frac{2q}{1-\cos\omega} - \frac{p'}{1-\hat{\xi}_*'} = \max\{ F_1, F_2\},
\label{F1F2}
\end{equation}
where
\[
  F_1(q,\omega) = \frac{2q}{1-\cos\omega} -\frac{p'}{1-\xi_*'},
\qquad  F_2(q,\omega) = \frac{2q}{1-\cos\omega} - \frac{p'}{1-e'}.
\]
The relation
\begin{equation}
F_1(q,\omega) = 0
\label{F1eq0}
\end{equation}
corresponds to
\begin{equation}
  q=p'/2.
\label{qpp}
\end{equation}
In fact $\xi_*'$ is defined so that it satisfies $\dnod^+=\dnod^-$ with $e=1$,
therefore, if \eqref{F1eq0} holds, we have
\[
\frac{2q}{1-\cos\omega} -\frac{p'}{1-\xi_*'} = \frac{2q}{1+\cos\omega}
-\frac{p'}{1+\xi_*'} = 0,
\]
from which we obtain \eqref{qpp}.
On the other hand, substituting $q=p'/2$ into \eqref{xipstar} we obtain
\begin{equation}
\xi_*' = \frac{2\cos\omega}{\sin^2\omega + \sqrt{\sin^4\omega +
    4\cos^2\omega}} = \frac{2\cos\omega}{1-\cos^2\omega +
  \sqrt{(1-\cos^2\omega)^2 + 4\cos^2\omega}} = \cos\omega,
\label{xistarcasopart}
\end{equation}
that yields \eqref{F1eq0}.

\noindent Since $F_1(q,\pi/2) = 2q-p'$, by continuity we obtain
\begin{eqnarray}
&&F_1(q,\omega)>0 \qquad\mbox{if}\ \ q>p'/2,\label{F1gt0}\\
&&F_1(q,\omega)<0 \qquad\mbox{if}\ \ q<p'/2,\label{F1lt0}
\end{eqnarray}
for each $\omega\in[0,\pi/2]$.
We also note that
\begin{equation}
F_2(q,\omega) < 0 \qquad\mbox{if}\ \ q<p'/2, \ \omega\in[\arccos e',\pi/2].
\label{F2lt0}
\end{equation}

Using \eqref{F1F2}, \eqref{F1gt0} we obtain that
\[
 \frac{2q}{1-\cos\omega} - \frac{p'}{1-\hat{\xi}_*'}
 >0\qquad\mbox{if}\ \ q>p'/2,
\]
so that, for such values of $q$, $\uppextom=0$ corresponds to
$\frac{2q}{1+\cos\omega}-q'=0$.

On the other hand, we can prove that
\[
\uppextom(q,\omega) < 0 \qquad\mbox{if}\ \ q<p'/2,
\]
therefore the curve $\uppextom=0$ does not intersect the region with
$q<p'/2$.  In fact, by \eqref{F1lt0},
\eqref{F2lt0} we obtain that
\[
\max\{F_1,F_2\}<0 \qquad\mbox{if}\ \ \omega\in [\arccos e',\pi/2],
\]
and we can easily check that, for such values of $q$,
\[
 \frac{2q}{1+\cos\omega}-q'<0 \qquad\mbox{if}\ \ \omega \in [0,\arccos
 e').
 \]

Finally, we prove that, if $q=p'/2$, we have
\[
\uppextom(q,\omega) = 0\qquad \mbox{if and only if}\qquad
\omega\in[\arccos e',\pi/2].
\]
Assume that $q=p'/2$.  If $\omega\in(\arccos e',\pi/2]$ then
 $\uppextom(q,\omega) = 0$.  In fact, in this case, from
 \eqref{xistarcasopart} we obtain $\hat{\xi}_*' = \cos\omega$, so that
\[
\qquad \frac{2q}{1-\cos\omega}-\frac{p'}{1-\hat{\xi}_*'} =
0
\]
and
\[
\frac{2q}{1+\cos\omega} - q'= \frac{p'}{1+\cos\omega} - \frac{p'}{1+e'} >  0.
\]
On the other hand, if $\omega\in[0\arccos e')$ then
\[
\uppextom(q,\omega)<0,
\]
because in this case
\[
\frac{2q}{1+\cos\omega} - q' <  0.
\]
Finally, if $\omega=\arccos e'$, we have $\hat{\xi}_*'=\cos\omega =e'$, so that
\begin{equation}
\frac{2q}{1-\cos\omega}-\frac{p'}{1-\hat{\xi}_*'} =
\frac{2q}{1+\cos\omega} - q' = 0.
\label{AeqB}
\end{equation}  

We conclude that the curve $\uppextom(q,\omega)=0$ is composed by the
vertical segment $\{(q,\omega): q=p'/2, \ \omega\in[\arccos
  e',\pi/2]\}$ and by the curve $\{(q,\omega)\in{\cal
  D}_2:\frac{2q}{1+\cos\omega}-q'=0, \ q>p'/2\}$. These two portions
of the curve $\uppextom=0$ meet in the point $(q,\omega)=(p'/2,\arccos
e')$ and therefore they form a unique connected component.
The proof of Proposition~\ref{prop:ILEC_qom} is concluded.
       
\end{proof}

\begin{remark}
  There can not exist $(q,\omega)\in{\cal D}_2$ such that we have
  linked orbits for each $(e,\omega')\in{\cal D}_1$, unlike the case
  of internal and external nodes.
\end{remark}
\begin{proof}
  If $\deltalink(q,e,\omega,\omega')>0$ for each $(e,\omega')\in{\cal
    D}_1$ then in particular $\upplinkom(q,\omega)>0$, and this
  corresponds to $\lowintom,\lowextom<0$ at $(q,\omega)$, so that, by
  {\em ii)} of Lemma~\ref{lem:omegabounds}, there exists $(e,\omega')$
  corresponding to a crossing configuration, that yields a
  contradiction.

\end{proof}  

\noindent In Figure~\ref{fig:qom} we show the possible linking configurations
for $q'=1$ and $e'=0.2$.

\smallbreak
In the next statement we present the optimal lower and upper bounds for
$\deltanod$ as functions of $(q,e)$.
\begin{proposition}
Let ${\cal D}_3 =
\{(\omega,\omega'): 0\leq \omega\leq {\pi}/{2}, 0\leq \omega' <
\pi\}$, ${\cal D}_4 = \{(q,e): 0< q\leq \qmax, 0\leq e\leq 1\}$.  For
each choice of $(q,e)\in{\cal D}_4$ we have
  \begin{align}
    &\displaystyle \min_{(\omega,\omega')\in{\cal D}_3} \deltanod =
  \max\Bigl\{0, \lowinte, \lowexte\Bigr\}, \label{lbe}\\
  &\displaystyle \max_{(\omega,\omega')\in{\cal D}_3} \deltanod =
  \max\{\upplinke,|p' - q(1+e)|\},\label{ube}
  \end{align}
where\footnote{here $\lowinte(q,1)=-\infty$, and $\upplinke(q,1)=Q'-q$.}
\begin{align*}
  &\lowinte(q,e) =q'-\frac{q(1+e)}{1-e},\qquad \lowexte(q,e) = q-Q',\\
&\upplinke(q,e) = \min\left\{\frac{q(1+e)}{1-e}-q', Q'-q\right\}.
\end{align*}
\label{prop:deltanod_qe_gen}
\end{proposition}

\begin{proof}

We prove some preliminary facts.
\begin{lemma}
The following properties hold:
\begin{itemize}
\item[i)] for each $(q,e)\in {\cal D}_4$ and $(\omega,\omega')\in{\cal
  D}_3$ we have
  \begin{eqnarray}
  &&\deltaint(q,e,\omega,\omega') \geq \deltaint(q,e,0,\pi) = \lowinte(q,e),\label{b1e}\\ 
  &&\deltaext(q,e,\omega,\omega') \geq \deltaext(q,e,0,\pi) = \lowexte(q,e),\label{b2e} 
  \end{eqnarray}
  therefore, given $(q,e)\in{\cal D}_4$, we have internal
(resp. external) nodes for each $(\omega,\omega')\in{\cal D}_3$ if and only if
$\lowinte(q,e)>0$ (resp. $\lowexte(q,e)>0$);

\item[ii)] if $(q,e)$ is such that $\lowinte(q,e)\leq 0$ and
  $\lowexte(q,e)\leq 0$, then there exists $(\omega, \omega')\in{\cal
  D}_3$ such that $\dnod^+\dnod^- = 0$.
\end{itemize}
\label{lem:ebounds}
\end{lemma}
\begin{proof}
  We prove the bounds \eqref{b1e}, \eqref{b2e} by observing that for
  each $(q,e)\in{\cal D}_4$ and $(\omega,\omega')\in{\cal D}_3$ we
  have
  \[\begin{split}
  \deltaint&\geq \min\Bigl\{\min_{\omega'\in[0,\pi]}\rpiup -
  \max_{\omega\in[0,\pi/2]}\rpiu, \min_{\omega'\in[0,\pi]}\rmenop -
  \max_{\omega\in[0,\pi/2]}\rmeno\Bigr\}\cr &=\min\{\rpiup|_{\omega'=0} -
  \rpiu|_{\omega=\pi/2}, \rmenop|_{\omega'=\pi} - \rmeno|_{\omega=0}\} =
  q' - \frac{q(1+e)}{1-e}\cr
  \end{split}\]
  and
  \[\begin{split}
  \deltaext&\geq \min\Bigl\{\min_{\omega\in[0,\pi/2]}\rpiu - \max_{\omega'\in[0,\pi]}\rpiup,
  \min_{\omega\in[0,\pi/2]}\rmeno - \max_{\omega'\in[0,\pi]}\rmenop\Bigr\}\cr
  &=
  \min\{\rpiu|_{\omega=0} - \rpiup|_{\omega'=\pi},
  \rmeno|_{\omega=\pi/2} - \rmenop|_{\omega'=0}\} = q - Q'.\cr
  \end{split}\]
  We conclude the proof of {\em i)} using properties {\em a)}, {\em b)} in
  Lemma~\ref{lem:delta}.
To prove {\em ii)} we note that
\[
\deltaint(q,e,\pi/2,\pi/2) = q'-q, \qquad
\deltaext(q,e,\pi/2,\pi/2) = q-q'.
\]
Therefore, either they are both zero and there is a crossing for
$(\omega,\omega')=(\pi/2,\pi/2)$, or they are different from zero and
opposite and, since we are assuming that $\lowinte,\lowexte\leq 0$ at
$(q,e)$, by continuity there exists $(\omega,\omega')\in{\cal D}_3$
corresponding to a crossing configuration.

\end{proof}
We also prove the following result.
\begin{lemma}
Let us consider the function
\[
D(\xi,\xi';p,p') = \min\Bigl\{\frac{p'}{1+\xi'} -
  \frac{p}{1+\xi}, \frac{p'}{1-\xi'} - \frac{p}{1-\xi} \Bigr\}
\]
defined for $(\xi,\xi')\in{\cal D} := (-1,1)\times(-1,1)$, depending on the
parameters $p,p'>0$.  Then we have
  \[
  \sup_{(\xi,\xi')\in{\cal D}}D(\xi,\xi';p,p') = \left\{
  \begin{array}{ll}
     p'-p = D(0,0;p,p') &\mbox{ if }\  p'\geq p,\cr
    \displaystyle\frac{p'-p}{2} = \limsup_{(\xi,\xi')\to\pm(1,1)}D(\xi,\xi';p,p')  &\mbox{ if }\  p'< p.\cr
    \end{array}\right. 
  \]
\label{lem:symfunc}
\end{lemma}
\begin{proof}
Let us set
  \[
  D^+(\xi,\xi';p,p')=\frac{p'}{1+\xi'} - \frac{p}{1+\xi} \ \quad\mbox{and}\ \quad
  D^-(\xi,\xi';p,p')=\frac{p'}{1-\xi'} - \frac{p}{1-\xi} .
  \]
  For each $\xi\in(-1,1)$, $D^+$ is a non-increasing
  function of $\xi'$, while $D^-$ is non-decreasing.
  Moreover,
\[
\lim_{\xi'\to -1^+}D^+(\xi,\xi';p,p') = +\infty, \qquad \lim_{\xi'\to
  1^-}D^+(\xi,\xi';p,p') = \frac{p'}{2} - \frac{p}{1+\xi},
\]
and
\[
\lim_{\xi'\to -1^+}D^-(\xi,\xi';p,p') = \frac{p'}{2} - \frac{p}{1-\xi}, \qquad
\lim_{\xi'\to 1^-}D^-(\xi,\xi';p,p') = +\infty.
\]
Therefore, for each $\xi\in(-1,1)$, there exists a unique value of $\xi' = \xi_*'(\xi)\in(-1,1)$
such that
\begin{equation}
D^+(\xi,\xi_*'(\xi);p,p') = D^-(\xi,\xi_*'(\xi);p,p').
\label{xipdef}
\end{equation}
Its expression is given by
\[
  \xi_*'(\xi) =
\frac{2p\xi}{\sqrt{p'^2(1-\xi^2)^2 + 4p^2\xi^2} + p'(1-\xi^2)}.
\]
Moreover, for each $\xi\in(-1,1)$, the maximum value of the function
\[
(-1,1)\ni\xi'\mapsto D(\xi,\xi';p,p') =
\min\{D^+(\xi,\xi';p,p'),D^-(\xi,\xi';p,p')\}
\]
is attained at
$\xi_*'(\xi)$, see Figure~\ref{fig:maxintext}.
Substituting into $D(\xi,\xi';p,p')$ we obtain
\[
\begin{split}
  D_*(\xi;p,p') &:= D^+(\xi,\xi_*'(\xi);p,p')\cr
  & = \frac{p'\bigl(\sqrt{p'^2(1-\xi^2)^2 + 4p^2\xi^2} + p'(1-\xi^2) \bigr)}
  {\sqrt{p'^2(1-\xi^2)^2 + 4p^2\xi^2} + p'(1-\xi^2) + 2p\xi} - \frac{p}{1+\xi},\cr
  \end{split}
\]
where we have also used \eqref{xipdef}.  The function
$\xi\mapsto\xi_*'(\xi)$ is odd, so that $\xi\mapsto D_*(\xi;p,p')$ is
even, in fact
\[
D_*(-\xi;p,p') = D^+(-\xi,-\xi_*'(\xi);p,p') =
D^-(\xi,\xi_*'(\xi);p,p') = D_*(\xi;p,p').
\]

We compute the stationary points of $D_*$ in $(\xi,\xi')\in{\cal D}$
fulfilling the condition $D^+=D^-$ by Lagrange's multiplier method.
These points satisfy the relations
\[
\left\{
\begin{split}
  &(1-\lambda)\nabla_{(\xi,\xi')}D^+ = -\lambda \nabla_{(\xi,\xi')}D^-,\\
  &D^+ = D^-,\\
  \end{split}
\right.
\]
for some $\lambda\in\R$, so that the determinant
\[
\det\left[
  \begin{array}{cc}
    \frac{p}{(1+\xi)^2} &-\frac{p}{(1-\xi)^2}\cr
    -\frac{p'}{(1+\xi')^2} &\frac{p'}{(1-\xi')^2}\cr
  \end{array}\right] =
    4pp'\frac{(\xi'-\xi)(1-\xi\xi')}{(1-\xi^2)^2(1-\xi'^2)^2}
\]
must vanish when we set $\xi'=\xi_*'(\xi)$.
This happens for $\xi=0$ or for $p'=p$.


To conclude the proof of this lemma we evaluate $D_*$ at
$\xi=0$ and compute the limit of $D_*$ for $\xi\to 1^-$:
\[
D_*(0;p,p') = p'-p, \qquad \lim_{\xi\to 1^-}D_*(\xi;p,p') = \frac{p'-p}{2}.
\]
Using the fact that $D_*$ is even we see that
\begin{itemize}
\item[i)] if $p'\geq p$, then $p'-p$ is the maximal value of $D_*$
  over $(-1,1)$, attained at $\xi=0$;
\item[ii)] if $p'< p$, then $(p'-p)/2$ is the supremum of $D_*$
  over $(-1,1)$, attained in the limit for $\xi\to 1^-$ and for
  $\xi\to -1^+$.
\end{itemize}

\end{proof}

We continue the proof of Proposition~\ref{prop:deltanod_qe_gen}.

{\em Lower bound:} we prove relation \eqref{lbe} observing that, by {\em i)} of
Lemma~\ref{lem:ebounds}, if $\lowinte(q,e)>0$ we can have only
internal nodes. Therefore $\min_{(\omega,\omega')\in{\cal D}_3}\deltanod(q,e) =
\min_{(\omega,\omega')\in{\cal D}_3}\deltaint(q,e) = \lowinte(q,e)$ and
$\deltaext(q,e,\omega,\omega')$, $\deltalink(q,e,\omega,\omega') <0$
for each $(\omega,\omega')\in{\cal D}_3$. In particular we have
$\lowexte(q,e)<0$.
In a similar way, if $\lowexte(q,e)>0$ we can have only external
nodes, therefore $\min_{(\omega,\omega')\in{\cal D}_3}\deltanod(q,e) = \min_{(\omega,\omega')\in{\cal
    D}_3}\deltaext(q,e) = \lowexte(q,e)$ and
$\deltaint(q,e,\omega,\omega')$, $\deltalink(q,e,\omega,\omega') <0$
for each $(\omega,\omega')\in{\cal D}_3$. In particular we have
$\lowinte(q,e)<0$.

Finally, if $\lowinte(q,e)\leq 0$ and $\lowexte(q,e)\leq 0$, by
{\em ii)} of Lemma~\ref{lem:ebounds} there exists $(\omega,\omega')\in{\cal
  D}_3$ corresponding to a crossing configuration, therefore
$\min_{(\omega,\omega')\in{\cal D}_3}\deltanod(q,e)=0$.
The previous discussion yields relation \eqref{lbe}.

  
{\em Upper bound:} given $(q,e)\in{\cal D}_4$ we can consider
$\dnod^+, \dnod^-$ as functions of $\xi=e\cos\omega,
\xi'=e'\cos\omega'$, with $\xi\in[0,e]$, $\xi'\in[-e',e']$.  From
Lemma~\ref{lem:symfunc} we obtain that the maximal value of
$\deltaint$ over ${\cal D}_3$ is
\begin{equation}
  \uppinte(q,e) =
  \left\{
  \begin{array}{ll}
    p'-p &\mbox{ if } p'\geq p\cr
    m_{\rm int}(q,e) &\mbox{ if } p'< p\cr
  \end{array}\right.
  \label{uppinte}
\end{equation}
for some $m_{\rm int}<(p'-p)/2<0$.
On the other hand, for each $\xi\in(-1,1)$ we have
\begin{equation}
\sup_{\xi'\in(-1,1)}\min\{-\dnod^+,-\dnod^-\} =
-\sup_{\xi'\in(-1,1)}\min\{\dnod^+,\dnod^-\},
\label{maxintext}
\end{equation}
see Figure~\ref{fig:maxintext}.
\begin{figure}[h!]
\centerline{\includegraphics[width=8cm]{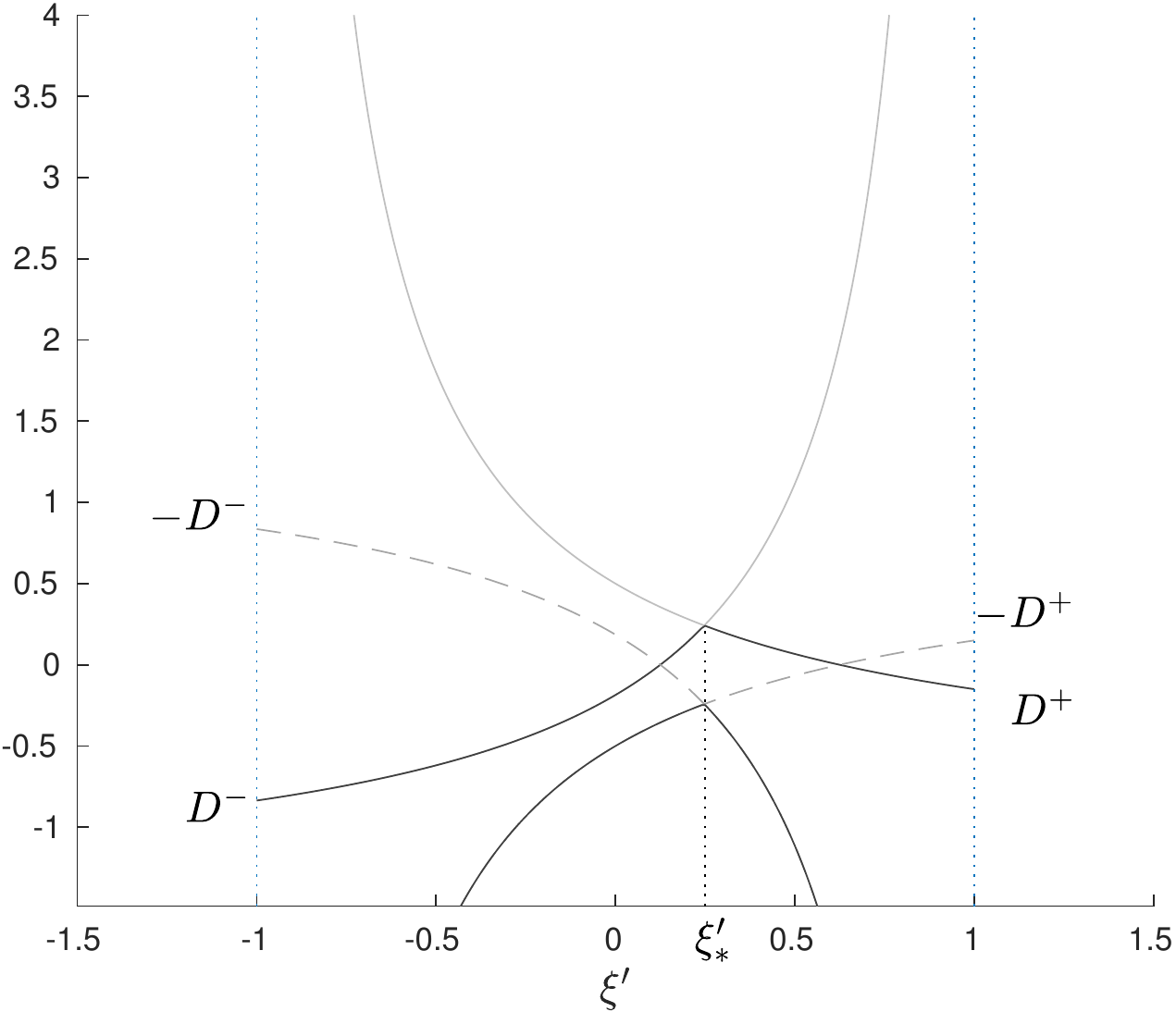}}
  \caption{Illustration of relation \eqref{maxintext}. Here $D^\pm$
    denote $\dnod^\pm$ regarded as functions of $\xi, \xi'$.}
  \label{fig:maxintext}
\end{figure}

\noindent Thus we conclude that the maximal value of $\deltaext$ over
${\cal D}_3$ is
\begin{equation}
\uppexte(q,e) =
\left\{
\begin{array}{ll}
 p-p' &\mbox{ if } p'\leq p \cr
 m_{\rm ext}(q,e)  &\mbox{ if } p'>p\cr  
  \end{array}
\right.,
\label{uppexte}
\end{equation}
for some $m_{\rm ext}<(p-p')/2 < 0$.
Therefore, for each $(q,e)\in{\cal D}_4$ we obtain
\begin{equation}
\max\{\max_{(\omega,\omega')\in{\cal D}_3}\deltaint,
\max_{(\omega,\omega')\in{\cal D}_3}\deltaext\} = |p'-q(1+e)|.
\label{maxmax}
\end{equation}

\begin{figure}[t!]
  \centerline{\epsfig{figure=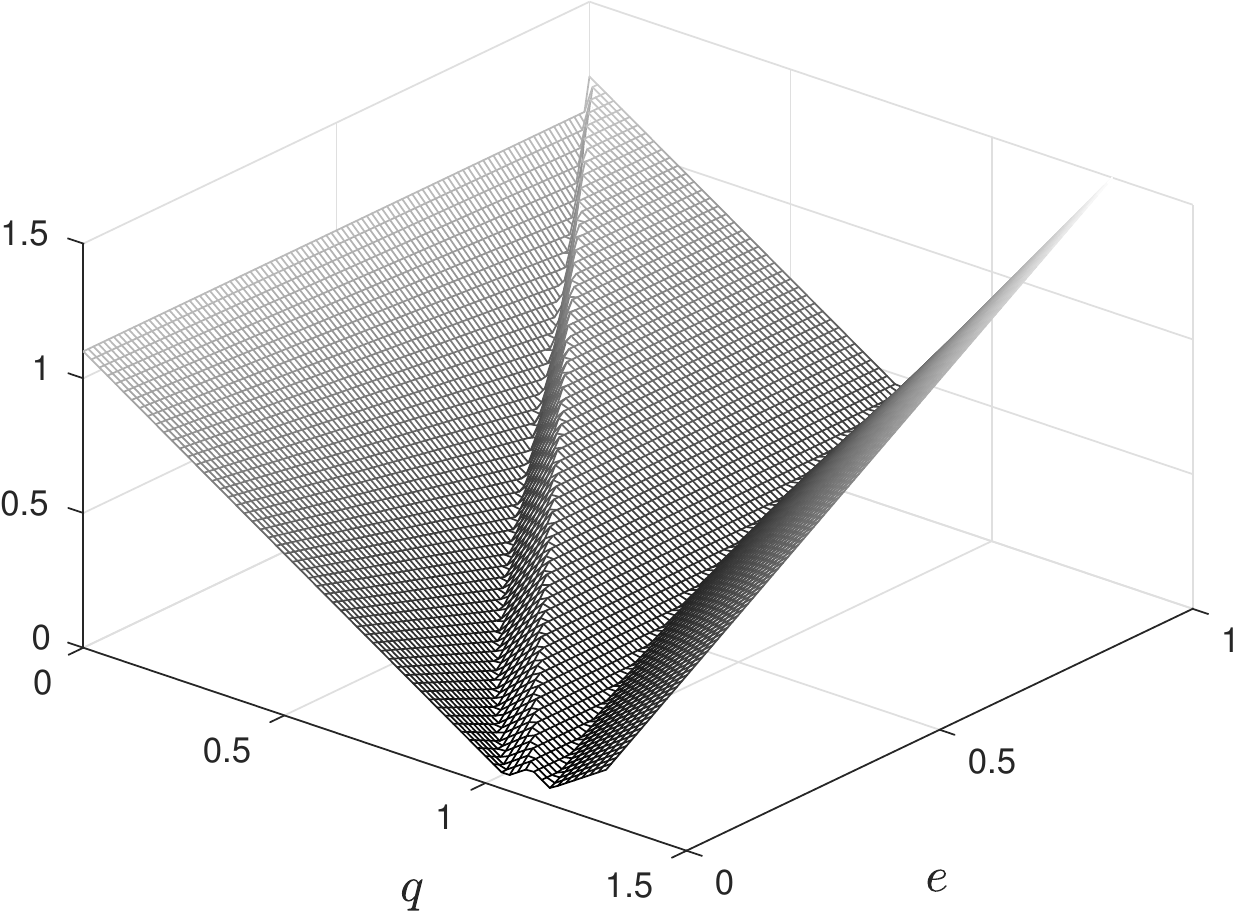,width=7.5cm}
    \hskip 0.2cm
    \epsfig{figure=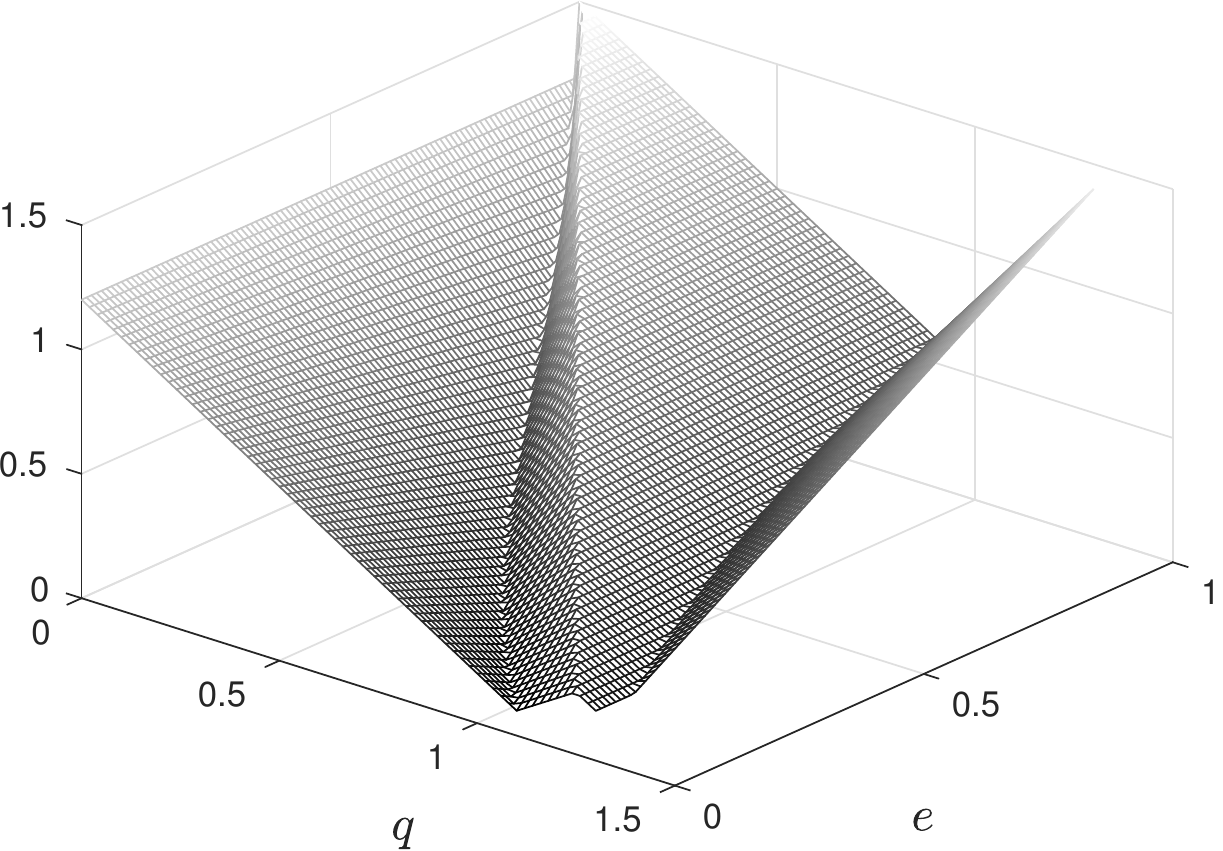,width=7.5cm}}
    \centerline{\epsfig{figure=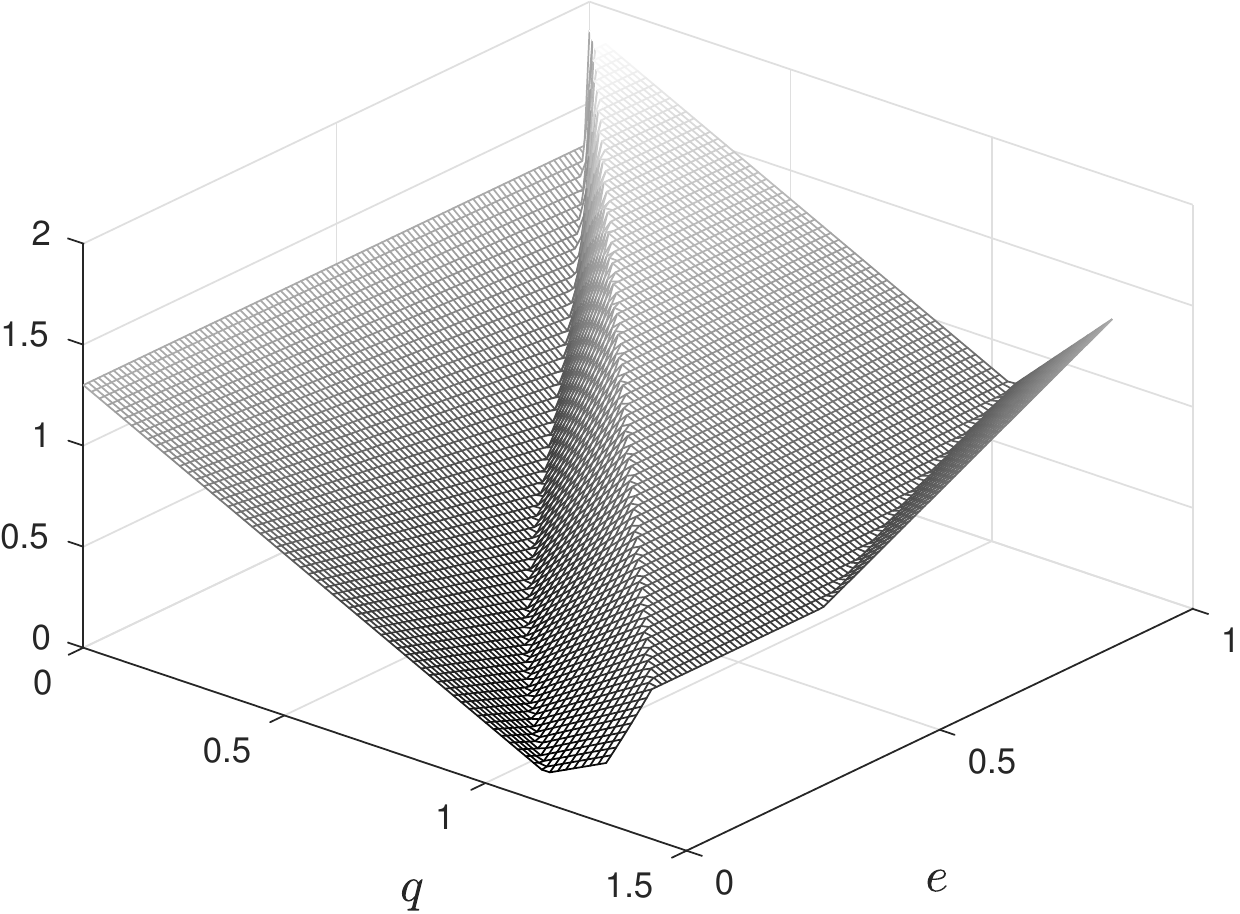,width=7.5cm}
    \hskip 0.2cm
    \epsfig{figure=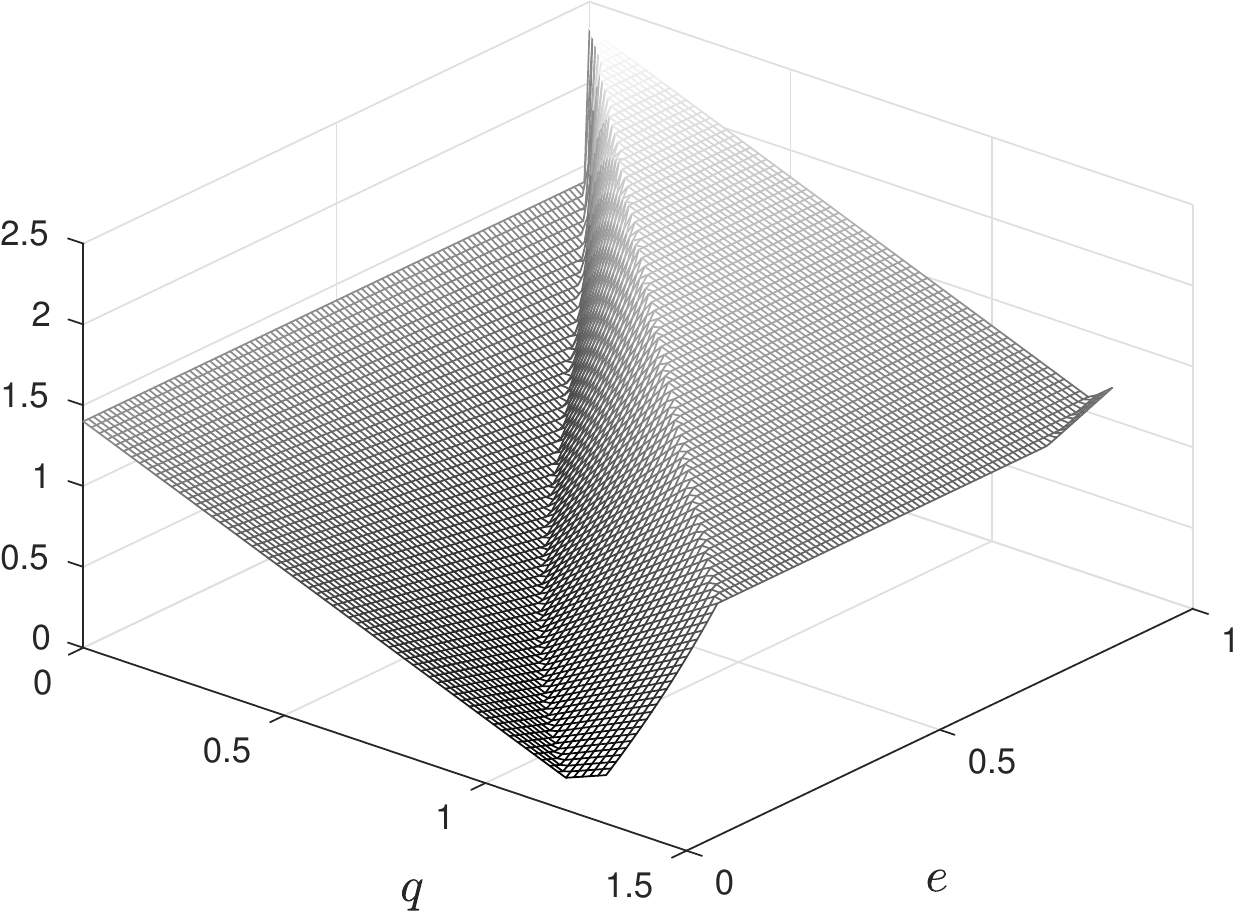,width=7.5cm}}
    \caption{Graphic of $(q,e)\mapsto \max_{(\omega,\omega')\in{\cal D}_3}
      \deltanod(q,e)$ for $e'=0.1$ (top left), $e'= 0.2$ (top right),
      $e'=0.3$ (bottom left), $e'=0.4$ (bottom right). Here we set $q'=1$.}
\label{maxdeltanodqe}
\end{figure}
Finally, we consider the function $\deltalink$ and examine
$\deltalink^{(i)}$ and $\deltalink^{(ii)}$ separately. By
Lemma~\ref{der_om_omp_e}, both $-\dnod^+$ and $\dnod^-$ are
non-decreasing functions of $\omega$ and non-increasing functions of
$\omega'$, therefore also $\deltalink^{(i)}$ is.  For each fixed value
of $(q,e)$, the maximal value of $\deltalink^{(i)}$ over ${\cal D}_3$
is attained for $\omega={\pi}/{2}$, $\omega'=0$ and is
\[
\min\{q(1+e)-q', Q' - q(1+e)\}.
\]
In a similar way we prove that, for each fixed value of $(q,e)$, the maximal value of
$\deltalink^{(ii)}$ over ${\cal D}_3$ is attained for $\omega=0$,
$\omega'=\pi$ and is
\[
\min\Bigl\{Q' - q, \frac{q(1+e)}{1-e} - q'\Bigr\}.
\]
Therefore the maximal value of $\deltalink$ over ${\cal D}_3$ is attained by $\deltalink^{(ii)}$ and corresponds to
\[
\upplinke(q,e) 
 =\ \min\Bigl\{\frac{q(1+e)}{1-e} - q', Q' - q\Bigr\}.
\]
We conclude the proof of relation \eqref{ube} using \eqref{deltanod_max},
\eqref{maxmax} and the optimal bound
\[
\deltalink(q,e,\omega,\omega) \leq \upplinke(q,e).
\]

\end{proof}

In Figure~\ref{maxdeltanodqe} we show the graphic of
$\max_{(\omega,\omega')\in{\cal D}_3}\deltanod(q,e)$ for different
values of $e'$, with $q'=1$.

\begin{proposition}
The zero level curves of $\lowinte, \lowexte, p'-q(1+e)$ divide the
plane $(q,e)$ into regions where different linking configurations are allowed.
\label{prop:ILEC_qe}
\end{proposition}
\begin{figure}[h!]
  \begin{center}
    \includegraphics[width=10cm]{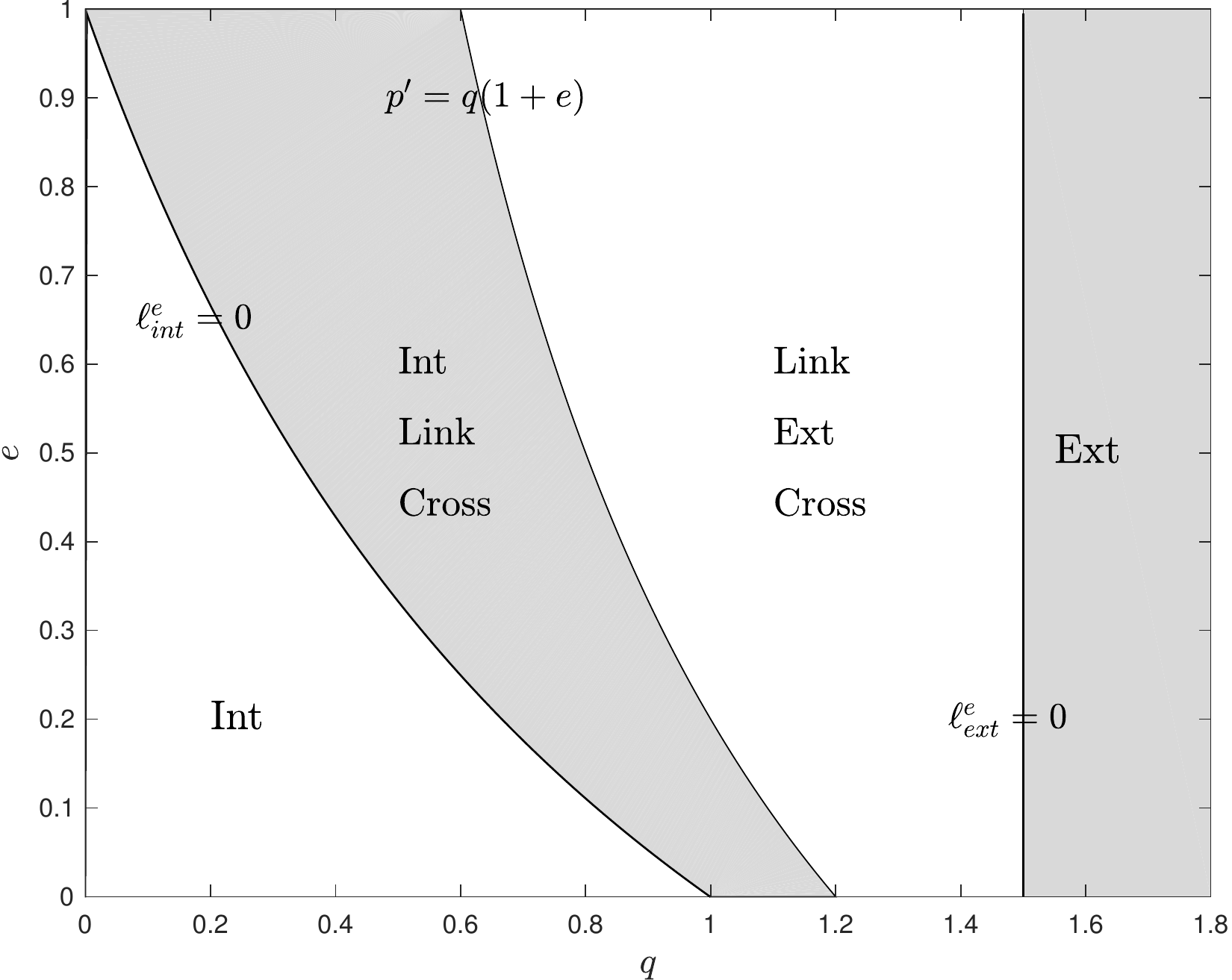}
  \end{center}
  \caption{Regions with different linking configurations in the plane $(q,e)$ for $q'=1$ and $e'=0.2$.}
  \label{fig:qe}
\end{figure}
\begin{proof}
  By Lemma~\ref{lem:ebounds}, given $(q,e)\in{\cal D}_4$, we have
  internal nodes for each $(\omega,\omega')\in{\cal D}_3$ if and only
  if $\lowinte(q,e)>0$, therefore the region where only internal nodes
  are possible is delimited on the right by the curve
  $\lowinte(q,e)=0$.  In a similar way, the region with only external
  nodes is delimited on the left by $\lowexte(q,e)=0$.
  
  Moreover, given $(q,e)$, we have internal nodes (resp. external
  nodes) for some choice of $(\omega,\omega')$ if and only if
  $\uppinte(q,e) > 0$ (resp. $\uppexte(q,e) > 0$).  From relations
  \eqref{uppinte}, \eqref{uppexte} we obtain that both the curves
  $\uppinte(q,e)=0$ and $\uppexte(q,e)=0$ correspond to $p'-q(1+e)=0$.
  Therefore, we can not have both the cases of internal and external
  nodes with the same value of $(q,e)$.
  
  In a similar way, given $(q,e)$, we have linked orbits for some choice of
  $(\omega,\omega')$ if and only if $\upplinke(q,e)> 0$.  We note that
  the curve
  \[
  \upplinke(q,e)= 0,
  \]
  delimiting the region where linked orbits are possible, coincides
  with the curve
  \[
  \{\lowinte(q,e)= 0\}\cup \{\lowexte(q,e)= 0\}.
  \]
  
\end{proof}

\begin{remark}
  There can not exist $(q,e)\in{\cal D}_4$ such that we have linked orbits
  for each $(\omega,\omega')\in{\cal D}_3$.
\end{remark}
\begin{proof} If there exists $(q,e)\in{\cal D}_4$ such that
  $\deltalink(q,e,\omega,\omega')>0$ for each
  $(\omega,\omega')\in{\cal D}_3$, then in particular
  $\upplinke(q,e)>0$, and this corresponds to $\lowinte,\lowexte<0$ at
  $(q,e)$, so that, by {\em ii)} of Lemma~\ref{lem:ebounds}, there
  exists $(\omega,\omega')$ corresponding to a crossing configuration,
  that yields a contradiction.

\end{proof}

\noindent In Figure~\ref{fig:qe} we show the possible linking configurations for
$q'=1$ and $e'=0.2$.

\smallbreak
Next we present optimal bounds for $\deltanod$ as functions of
$(q,\omega')$. To this aim, we let $\omega$ vary in $[0,\pi]$ and
$\omega'$ in $[0,\pi/2]$, which is a different choice with respect to
\eqref{deltanod_domain}, however it also allows us to get all the
possible values of $\deltanod$.
\begin{proposition}
Let ${\cal D}_5 = \{(e,\omega): 0\leq e\leq 1, 0\leq \omega\leq
\pi\}$, ${\cal D}_6 = \{(q,\omega'): 0< q\leq \qmax, 0\leq \omega'\leq
\pi/2\}$.  For each choice of $(q,\omega')\in{\cal D}_6$ we have
  \begin{align}
    &\displaystyle \min_{(e,\omega)\in{\cal D}_5} \deltanod =
  \max\bigl\{0, \lowextomp\bigr\}, \label{lbomp}\\
  &\displaystyle \max_{(e,\omega)\in{\cal D}_5} \deltanod =
  \max\Bigl\{\upplinkomp,\uppextomp\Bigr\},\label{ubomp}
  \end{align}
where
\begin{align*}
&\lowextomp(q,\omega') =  q-\frac{p'}{1-e'\cos\omega'},\\
&\upplinkomp(q,\omega') = \frac{p'}{1-e'\cos\omega'} - q,
\end{align*}
and
\[
\uppextomp(q,\omega') = \frac{2q}{1+\cos\omega_*} - \frac{p'}{1+e'\cos\omega'},
\]
with
\[
\cos\omega_* =
\frac{p'e'\cos\omega'}{\sqrt{q^2(1-e'^2\cos^2\omega')^2
    + (p'e'\cos\omega')^2} + q(1-e'^2\cos^2\omega')}.
%
\]
\label{prop:deltanod_qomp_gen}
\end{proposition}

We prove some preliminary facts.
\begin{lemma}
  The following properties hold:
  \begin{itemize}
    \item[i)] for each $(q,\omega')\in{\cal D}_6$ and $(e,\omega)\in{\cal
      D}_5$ we have
  \begin{eqnarray}
    &&\inf_{(e,\omega)\in{\cal D}_5}\deltaint(q,e,\omega,\omega') = -\infty,\label{minusinf}\\
    &&\deltaext(q,e,\omega,\omega') \geq
    \deltaext(q,0,\omega,\omega') = \lowextomp(q,\omega').\label{lowerbomp}
  \end{eqnarray}
\item[ii)] If $(q,\omega')$ is such that $\lowextomp(q,\omega')\leq
  0$, then there exists $(e,\omega)\in{\cal D}_5$ such that
  $\dnod^+\dnod^-=0$.
  \end{itemize}
  \label{lem:omegapbounds}
\end{lemma}
\begin{proof}
 Setting $e_k =1-\frac{1}{k},
  \omega_k=\frac{1}{k}$, $k\in\N$, we have
  \[
\lim_{k\to\infty}\deltaint(q,e_k,\omega_k,\omega') = -\infty
\]
for each $(q,\omega')\in{\cal D}_6$.
We prove the bound \eqref{lowerbomp} by observing that
for each $(q,\omega')\in{\cal D}_6$ and $(e,\omega)\in{\cal D}_5$ we have
\[
\begin{split}
\deltaext&\geq \min\biggl\{\min_{(e,\omega)\in{\cal D}_5}\rpiu - \rpiup,
\min_{(e,\omega)\in{\cal D}_5}\rmeno - \rmenop\biggr\}\cr
&=
\min\{\rpiu|_{e=0} - \rpiup,
\rmeno|_{e=0} - \rmenop\} = q - \frac{p'}{1-e'\cos\omega'},\cr
\end{split}
\]
where the last equality holds because $\omega'\in[0,\pi/2]$.

To prove $ii)$ we observe that by Lemma~\ref{der_om_omp_e}, for each
$(q,\omega')\in{\cal D}_6$, the maximal value of $\deltaext$ is
attained at $e=1$, whatever the value of $\omega$. By the same lemma,
$-\dnod^+$ is a non-decreasing function of $\omega$, while $-\dnod^-$
is non-increasing, whatever the value of $e$. Since
\[
\bigl.\lim_{\omega\to 0^+}\dnod^+\bigr|_{e=1} = \frac{p'}{1+\xi'}-q, \qquad
\bigl.\lim_{\omega\to \pi^-}\dnod^+\bigr|_{e=1} = -\infty,
\]
and
\[
\bigl.\lim_{\omega\to 0^+}\dnod^-\bigr|_{e=1} = -\infty, \qquad
\bigl.\lim_{\omega\to
  \pi^-}\dnod^-\bigr|_{e=1} = \frac{p'}{1-\xi'} - q,
\]
there is always a value $\omega_*$ of $\omega\in[0,\pi]$ such that
\begin{equation}
\bigl.\dnod^+\bigr|_{e=1,\omega=\omega_*}=\bigl.\dnod^-\bigr|_{e=1,\omega=\omega_*},
\label{dpeqdm}
\end{equation}
and this is given by relation
\begin{equation}
\cos\omega_* = \frac{p'e'\cos\omega'}{\sqrt{q^2(1-e'^2\cos^2\omega')^2
    + (p'e'\cos\omega')^2} + q(1-e'^2\cos^2\omega')}.
\label{cosomegastar}
\end{equation}
We conclude that the maximal value of $\deltaext$ over ${\cal D}_5$
is given by
\begin{equation}
  \uppextomp(q,\omega')
    %
    = \frac{2q}{1+\cos\omega_*} - \frac{p'}{1+e'\cos\omega'} =
    \frac{2q}{1-\cos\omega_*} - \frac{p'}{1-e'\cos\omega'}.
    \label{uppextomp}
  \end{equation}
 If $\uppextomp(q,\omega')\geq 0$, then there exists
 $(e,\omega)\in{\cal D}_5$ corresponding to a crossing configuration
 because we are assuming $\lowextomp(q,\omega')\leq 0$.
 On the other hand, if
 $\uppextomp(q,\omega')< 0$ we have $\deltaext(q,e,\omega,\omega')<0$
 for each $(e,\omega)\in{\cal D}_5$. However, this assumption yields a
 contradiction, in fact one of the following cases
 holds:
 \begin{itemize}
 \item[a)] $\deltaint(q,e,\omega,\omega') > 0$ for some
   $(e,\omega)\in{\cal D}_5$;
 \item[b)] $\deltaint(q,e,\omega,\omega') < 0$ for each
   $(e,\omega)\in{\cal D}_5$, that is,
   \[
   \uppintomp(q,\omega')=\max_{(e,\omega)\in{\cal
       D}_5}\deltaint(q,\omega') < 0.
   \]
 \end{itemize}
 If a) holds, then by relation \eqref{minusinf} and the continuity of
 $\deltaint$ there exists $(e,\omega)\in{\cal D}_5$ yielding a
 crossing configuration. Instead, if b) holds, from
 $\uppintomp<0$ and $\uppextomp <0$ we obtain that $\dnod^+\dnod^-<0$ at $(q,\omega')$ for
 each $(e,\omega)\in{\cal D}_5$, that is, for the considered pair
 $(q,\omega')$ we always have linked orbits. However, this contradicts
 relation \eqref{dpeqdm}.
 
\end{proof}

We continue the proof of Proposition~\ref{prop:deltanod_qomp_gen}.

\begin{proof}

  {\em Lower bound:} \eqref{lbomp} follows from
  Lemma~\ref{lem:omegapbounds}.
  
  {\em Upper bound:} By Lemma~\ref{der_om_omp_e}
  we obtain
  \[
  \deltaint(q,e,\omega,\omega')\leq \deltaint(q,0,\omega,\omega')
  =\min\Bigl\{\frac{p'}{1+e'\cos\omega'}-q, \frac{p'}{1-e'\cos\omega'}-q \Bigr\}
  \]
  for each $(q,\omega')\in{\cal D}_6$, and each $(e,\omega)\in{\cal D}_5$. We
  conclude that the maximal value of $\deltaint$ over ${\cal D}_5$ is
  \begin{equation}
    \uppintomp(q,\omega') =
    \frac{p'}{1+e'\cos\omega'} - q.
    \label{uppintomp}
  \end{equation}
  
  The maximal value of $\deltaext$ over ${\cal D}_5$ has been computed
  in Lemma~\ref{lem:omegapbounds} and is given in \eqref{uppextomp}.
  
\begin{figure}[t!]
  \centerline{\epsfig{figure=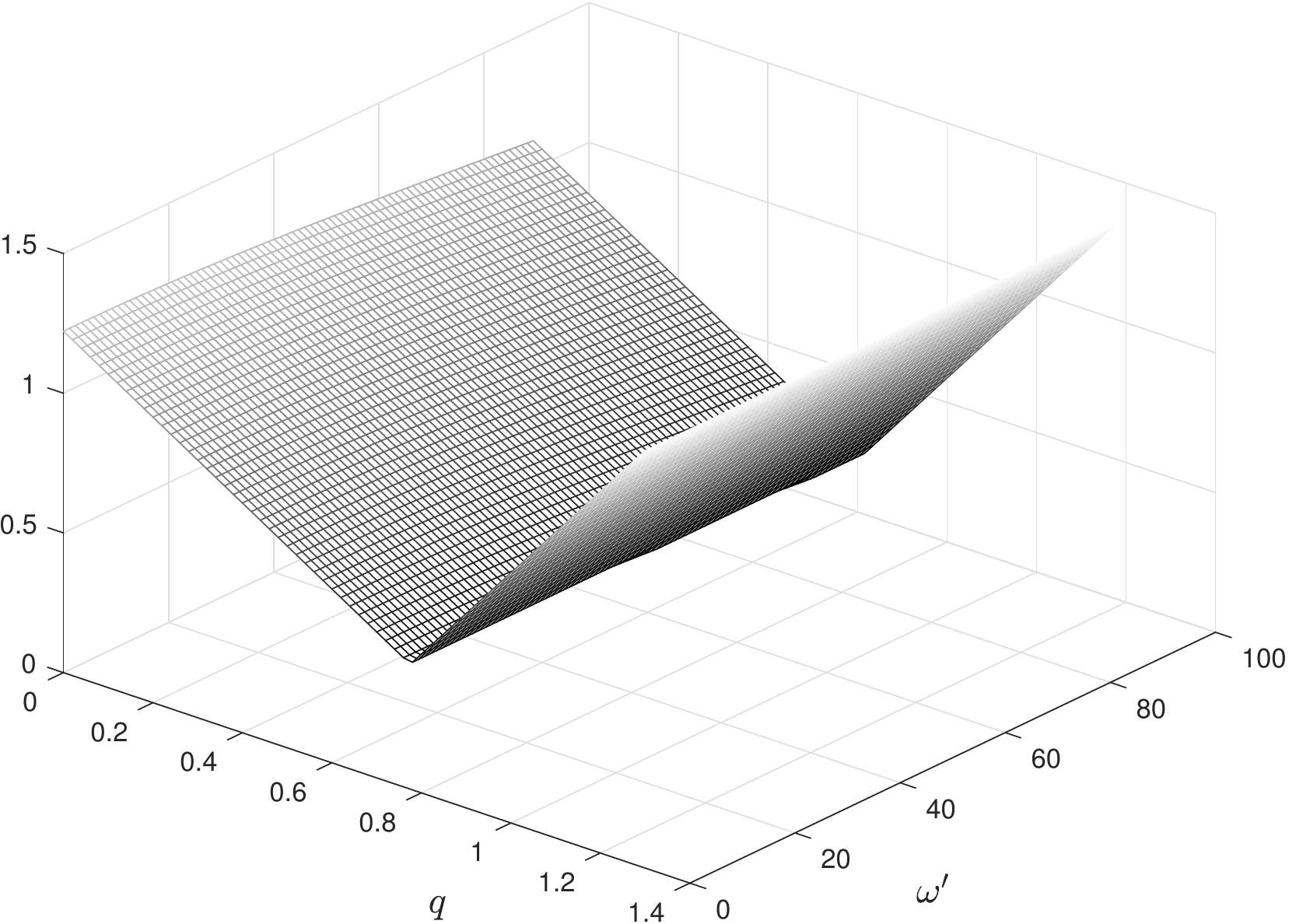,width=7.5cm}
    \hskip 0.2cm
    \epsfig{figure=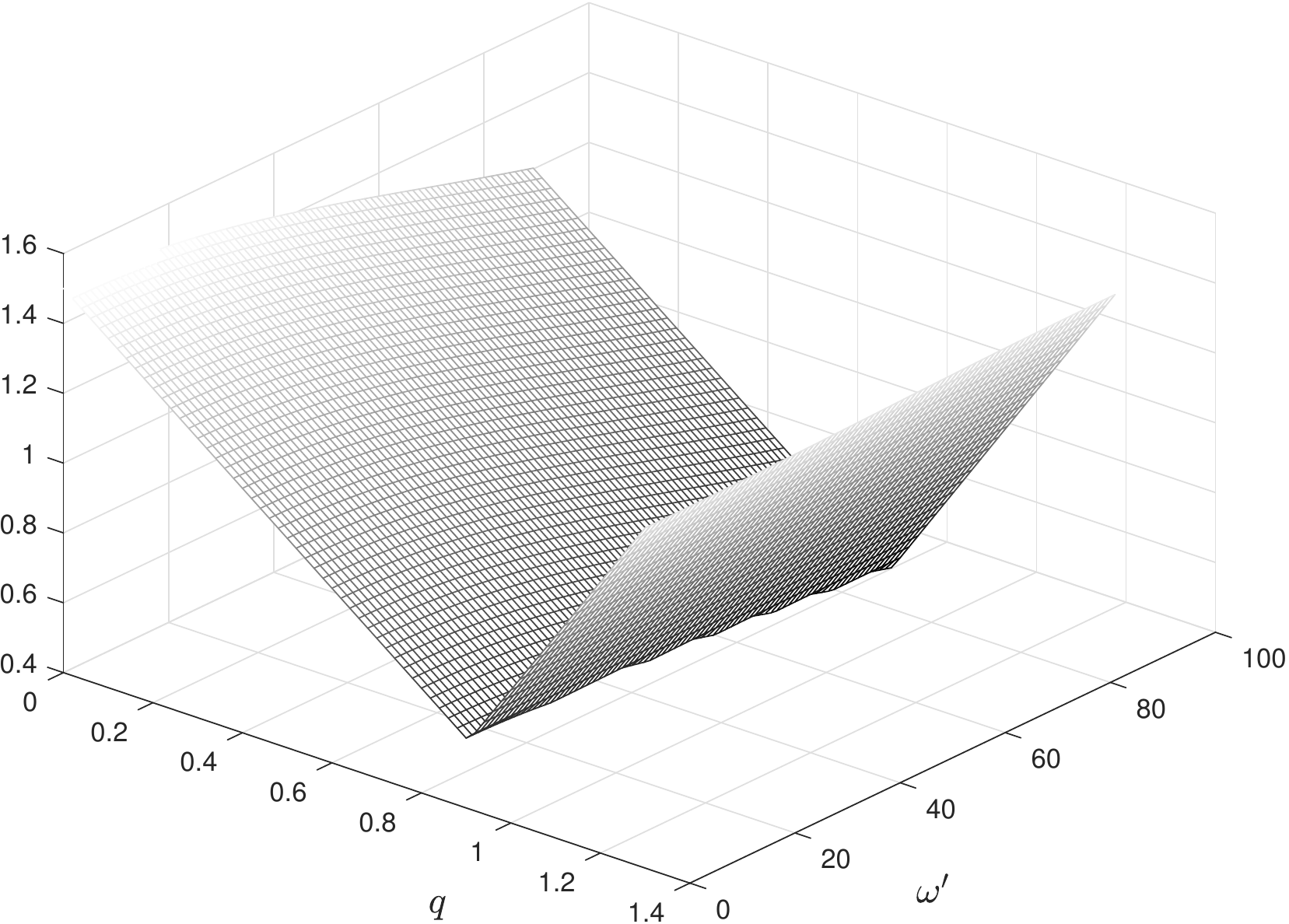,width=7.5cm}}
    \centerline{\epsfig{figure=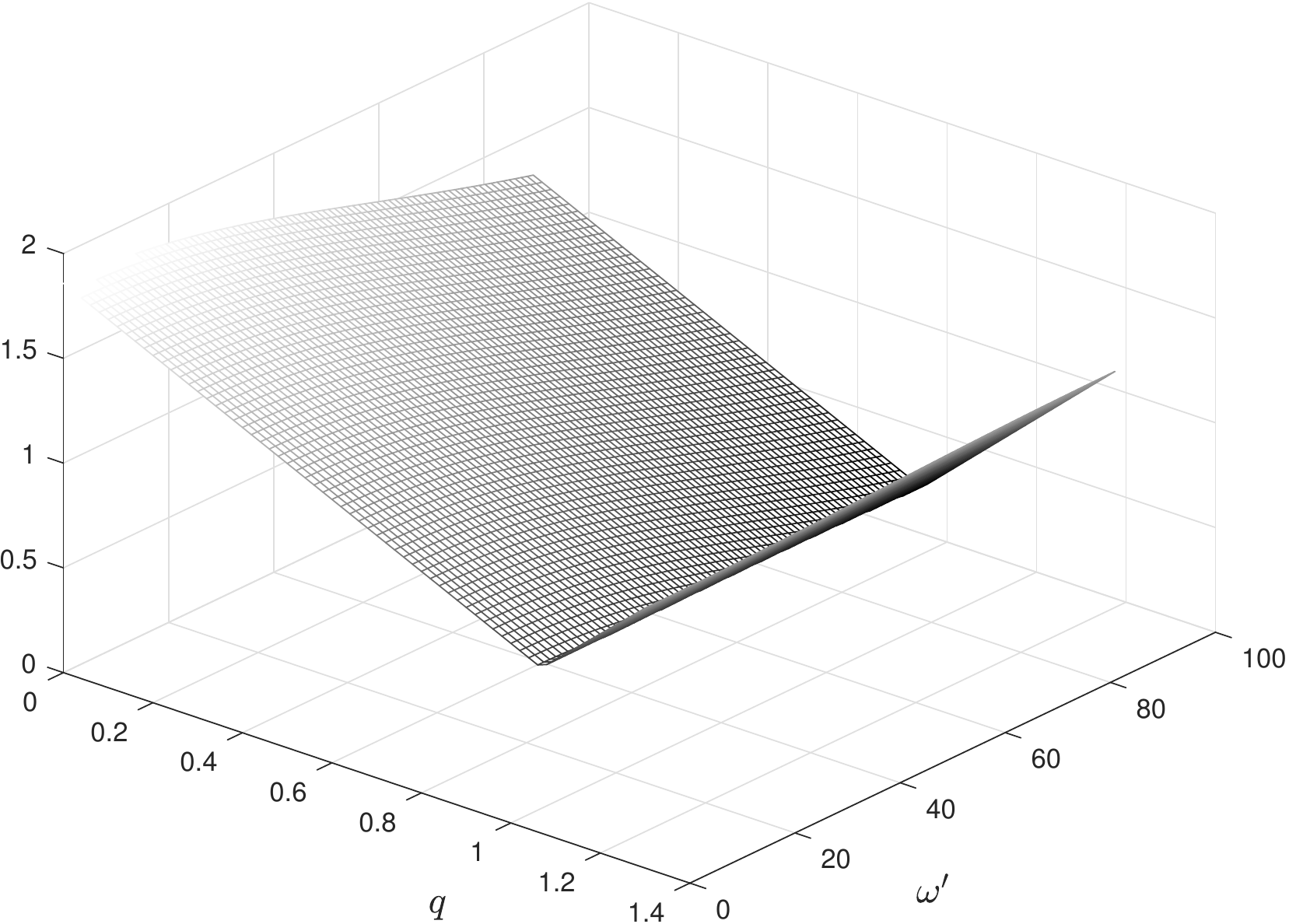,width=7.5cm}
    \hskip 0.2cm
    \epsfig{figure=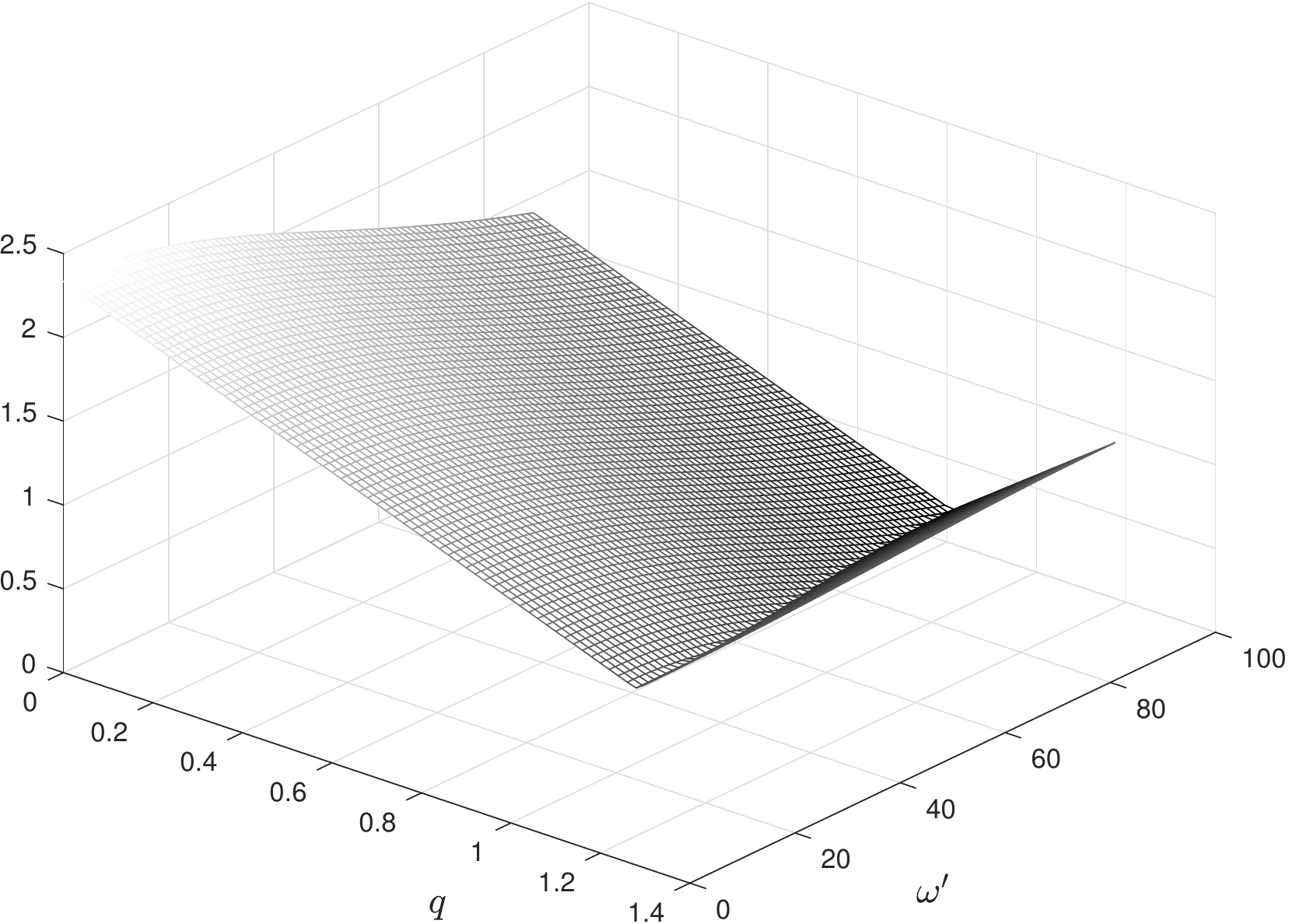,width=7.5cm}}
    \caption{Graphic of $(q,\omega')\mapsto
      \max_{(e,\omega)\in{\cal D}_5} \deltanod(q,\omega')$ for
      $e'=0.1$ (top left), $e'= 0.2$ (top right), $e'=0.3$ (bottom
      left), $e'=0.4$ (bottom right). Here we set $q'=1$.}
\label{maxdeltanodqomp}
\end{figure}

By Lemma~\ref{der_om_omp_e}, for each $(q,\omega')\in{\cal D}_6$ the
  maximal value of $\deltalink^{(i)}$ is attained at $\omega = \pi$
  and the maximal value of $\deltalink^{(ii)}$ is attained at $\omega
  = 0$.  We note that
  \[
  \begin{split}
    \deltalink^{(i)}|_{\omega=\pi} &=\min\Bigl\{ \frac{q(1+e)}{1-e} -
    \frac{p'}{1+e'\cos\omega'}, \frac{p'}{1-e'\cos\omega'} - q\Bigr\}\cr
    &\geq \min\Bigl\{\frac{p'}{1+e'\cos\omega'} -q, \frac{q(1+e)}{1-e} -
    \frac{p'}{1-e'\cos\omega'} \Bigr\}  = \deltalink^{(ii)}|_{\omega=0} 
\end{split}
  \]
  for each $e\in[0,1]$. Therefore the maximal value of $\deltalink$
  over ${\cal D}_5$ is obtained by $\deltalink^{(i)}$.
  By Lemma~\ref{der_om_omp_e}, $-\dnod^+|_{\omega=\pi}$ is
  non-decreasing with $e$, while $\dnod^-|_{\omega=\pi}$ is constant.
  Since
  \[
\lim_{e\to 1^-}-\dnod^+|_{\omega=0} = +\infty
  \]
  the maximal value of $\deltalink$ over ${\cal D}_5$
  is given by
  \begin{equation}
    \upplinkomp(q,\omega') = \dnod^-|_{\omega=\pi} =
     \frac{p'}{1-e'\cos\omega'} - q.
  \label{upplinkomp}
  \end{equation}
  Finally, we note that $\uppintomp\leq\upplinkomp$, therefore
  \[
  \max_{(e,\omega)\in{\cal D}_5}\deltanod = \max\{\uppintomp,\upplinkomp,\uppextomp\}
  = \max\{\upplinkomp,\uppextomp\}.
  \]

\end{proof}

In Figure~\ref{maxdeltanodqomp} we show the graphic of
$\max_{(e,\omega)\in{\cal D}_5}\deltanod(q,\omega')$ for different
values of $e'$, with $q'=1$.  Using Remark~\ref{rem:sym} we can extend
by symmetry the graphic of $\max_{(e,\omega)\in{\cal
    D}_5}\deltanod(q,\omega')$ to the set $(0,\qmax]\times[0,2\pi)$.

\begin{proposition}
The zero level curves of $\lowextomp, \uppintomp, \uppextomp$ divide
the plane $(q,\omega')$ into regions where different linking
configurations are allowed. Moreover, the curve $\uppextomp=0$ corresponds to the straight line $q=p'/2$.
\label{prop:ILEC_qomp}
\end{proposition}
\begin{figure}[h!]
  \begin{center}
    \includegraphics[width=10cm]{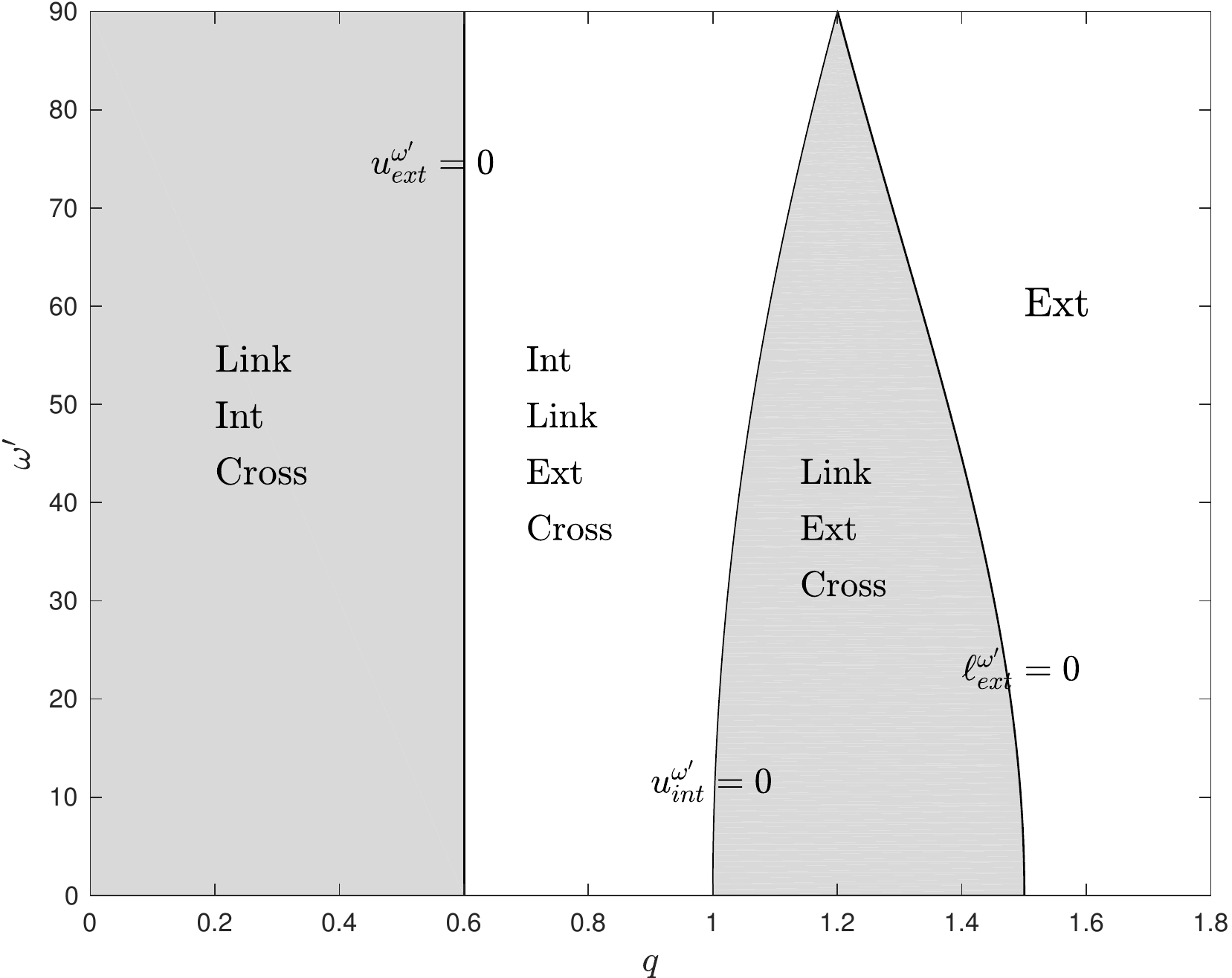}
  \end{center}
  \caption{Regions with different linking configurations in the plane $(q,\omega')$ for $q'=1$ and $e'=0.2$.}
  \label{fig:qomp}
\end{figure}
\begin{proof}
  By relation \eqref{minusinf}, given $(q,\omega')\in{\cal D}_6$, we
  can not have internal nodes for each $(e,\omega)\in{\cal D}_5$.
  Moreover, we have only external nodes if and only if
  $\lowextomp(q,\omega')>0$, i.e. for $q>p'/(1-e'\cos\omega')$. On the
  other hand, we have internal nodes for some choice of $(e,\omega)$
  if and only if $\uppintomp(q,\omega')>0$, i.e. for
  $q<p'/(1+e'\cos\omega')$.  Moreover, we have external nodes
  (resp. linked orbits) for some choice of $(e,\omega)$ if and only if
  $\uppextomp(q,\omega')>0$ (resp. $\upplinkomp(q,\omega')>0$).

\smallbreak We describe the shape of the curve $\uppextomp(q,\omega')=0$.

\noindent Eliminating $\cos\omega_*$ from equations
  \[
  \frac{p'}{1+e'\cos\omega'} - \frac{2q}{1+\cos\omega_*} =
  \frac{p'}{1-e'\cos\omega'} - \frac{2q}{1-\cos\omega_*} = 0
  \]
  we obtain
  \begin{equation}
  q = \frac{p'}{2}.
  \label{vertline}
  \end{equation}
  Vice versa, substituting $q=p'/2$ into \eqref{cosomegastar} we obtain
  \[
   \cos\omega_* = e'\cos\omega',
  \]
  so that $\uppextomp(q,\omega') = 0$ for each $\omega'\in[0,\pi/2]$.
  Therefore the curve $\uppextomp=0$ is the straight line defined by
  \eqref{vertline}.

\end{proof}

\begin{remark}
  There can not exist $(q,\omega')\in{\cal D}_6$ such that we have
  linked orbits for each $(e,\omega)\in{\cal D}_5$.
\end{remark}
\begin{proof} If
  $\deltalink(q,e,\omega,\omega')>0$ for each $(e,\omega)\in{\cal
    D}_5$ then in particular $\upplinkomp(q,\omega')>0$ and this
  corresponds to $\lowextomp(q,\omega')<0$ because
  $\upplinkomp=-\lowextomp$ . Therefore, by {\em ii)} of
  Lemma~\ref{lem:omegapbounds}, there exists $(e,\omega)$
  corresponding to a crossing configuration, that yields a
  contradiction.

\end{proof}

\noindent In Figure~\ref{fig:qomp} we show the possible linking
configurations for $q'=1$ and $e'=0.2$.

\subsection{Bounds for $\deltanod$ when $e'=0$}

In this section we consider the particular case $e'=0$, where ${\cal
  A}'$ is circular.
We recall some results proved in \cite{GV2013} concerning the orbit
distance $\dmin$, that is the distance between the sets ${\cal A}$ and
${\cal A}'$, and compare them with the corresponding results for the
nodal distance $\deltanod$, that can be obtained by setting $e'=0$ in
the statements of Propositions~\ref{prop:deltanod_qom_gen},
\ref{prop:deltanod_qe_gen}, \ref{prop:deltanod_qomp_gen}.

Assume $q'>0$ is given and let $e'=0$. The following proposition, proved in
\cite{GV2013}, gives optimal bounds for $\dmin$ as functions of
$(q,\omega)$.
\begin{proposition}
Set ${\cal D}_1' = \{(e,I): 0\leq e\leq 1, 0\leq I\leq \pi/2\}$ and
${\cal D}_2 = \{(q,\omega): 0< q\leq \qmax, 0\leq \omega\leq \pi/2\}$.
For each choice of $(q,\omega)\in{\cal D}_2$ we have
\[
\begin{array}{l}
\displaystyle \min_{(e,I)\in{\cal D}_1'} \dmin = \max\{0,q-q'\},\\
\displaystyle \max_{(e,I)\in{\cal D}_1'} \dmin = \max\{q'-q,\delta_{\omega}(q,\omega)\},\\
\end{array}
\]
where $\delta_{\omega}(q,\omega)$ is the distance 
between ${\cal A}'$ and ${\cal
  A}$ with $e=1, I=\pi/2$:
\begin{equation}
\delta_{\omega}(q,\omega) =  
\sqrt{(\xi-q'\sin\omega)^2 + 
\Bigl(\frac{\xi^2 - 4q^2}{4q} + q'\cos\omega\Bigr)^2}\,,
\label{deltaomega}
\end{equation}
with $\xi=\xi(q,\omega)$ the unique real solution of
\[
x^3 + 4q(q+\cos\omega)x - 8q'q^2\sin\omega = 0.
\]
\label{prop:qom}
\end{proposition}

We compare the above result with the following.
\begin{proposition}
Set ${\cal D}''_1 = \{e: 0\leq e\leq 1\}$ and
${\cal D}_2 = \{(q,\omega): 0< q\leq \qmax, 0\leq \omega\leq {\pi}/{2}\}$.
For each choice of $(q,\omega)\in{\cal D}_2$ we have
\begin{eqnarray}
&&\displaystyle \min_{e\in{\cal D}_1''} \deltanod = \max\left\{0, q'-\frac{2q}{1-\cos\omega}, q-q'\right\}\label{lbom_ep0},\\
&&\displaystyle \max_{e\in{\cal D}_1''} \deltanod =\max\{q'-q, \frac{2q}{1+\cos\omega} - q'\}.\label{ubom_ep0}
\end{eqnarray}
\label{prop:deltanod_qom}
\end{proposition}
\begin{proof}

We consider the statement of Proposition~\ref{prop:deltanod_qom_gen}
for $e'=0$, so that $Q'=p'=q'$.
By Lemma~\ref{der_om_omp_e} we obtain
\[
\upplinkom(q,\omega) \leq Q' -
\frac{q(1+\hat{e}_*)}{1+\hat{e}_*\cos\omega} = p' -
\frac{q(1+\hat{e}_*)}{1+\hat{e}_*\cos\omega} \leq p'-q = \uppintom(q,\omega).
\]
Moreover, for $e'=0$ we have $\hat{\xi}'_* = 0$, therefore
\[
\frac{2q}{1-\cos\omega} - \frac{p'}{1-\hat{\xi}_*'} =
\frac{2q}{1-\cos\omega} - q',
\]
so that
\[
\uppextom(q,\omega) = \frac{2q}{1+\cos\omega} - q',
\]
and \eqref{lbom},  \eqref{ubom} reduce to \eqref{lbom_ep0}, \eqref{ubom_ep0}.

\end{proof}

\begin{figure}[t!]
\centerline{\includegraphics[width=7cm]{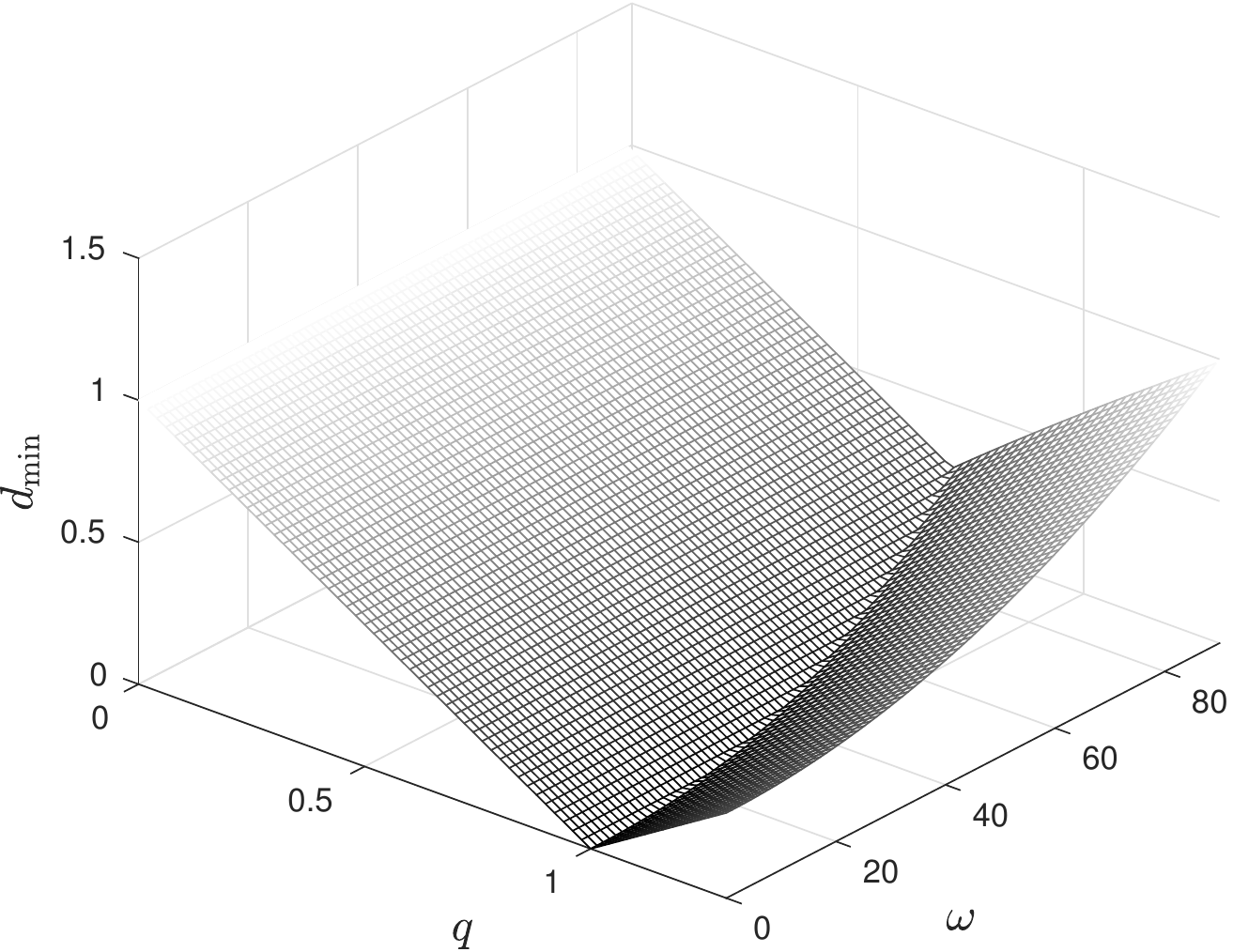}\hskip
  0.5cm
  \includegraphics[width=7cm]{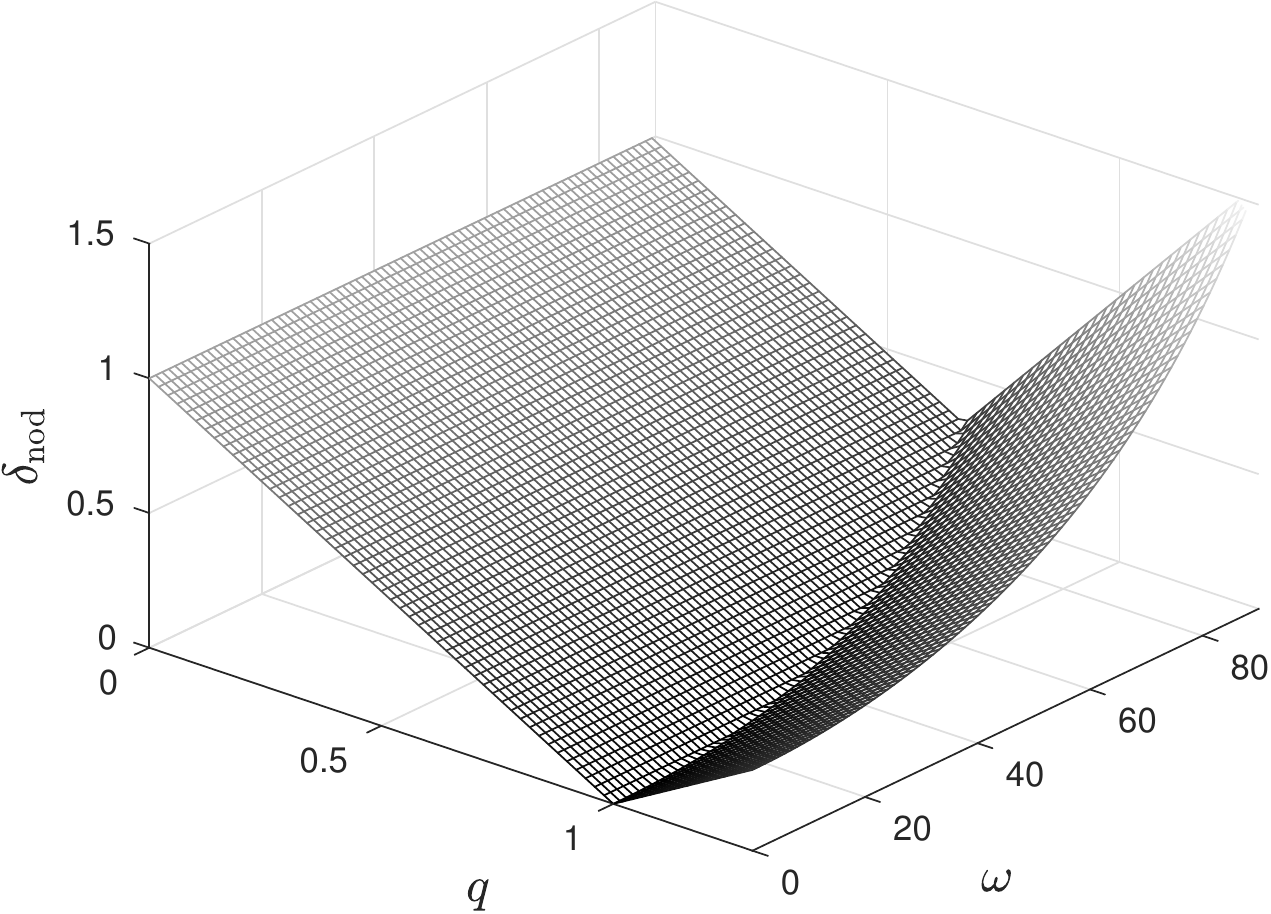}}
\caption{Left: $\max_{(e,I)\in{\cal D}_1'}\dmin(q,\omega)$. Right:
  $\max_{e\in{\cal D}_1''}\deltanod(q,\omega)$.}
\label{fig:bounds_qom_circ}
\end{figure}

\noindent In Figure~\ref{fig:bounds_qom_circ}, for $q'=1$, we show the graphics of
$\max_{(e,I)\in{\cal D}_1'}\dmin(q,\omega)$ on the left, and of
$\max_{e\in{\cal D}_1''}\deltanod(q,\omega)$ on the right.

\begin{figure}[t!]
\centerline{\includegraphics[width=9cm]{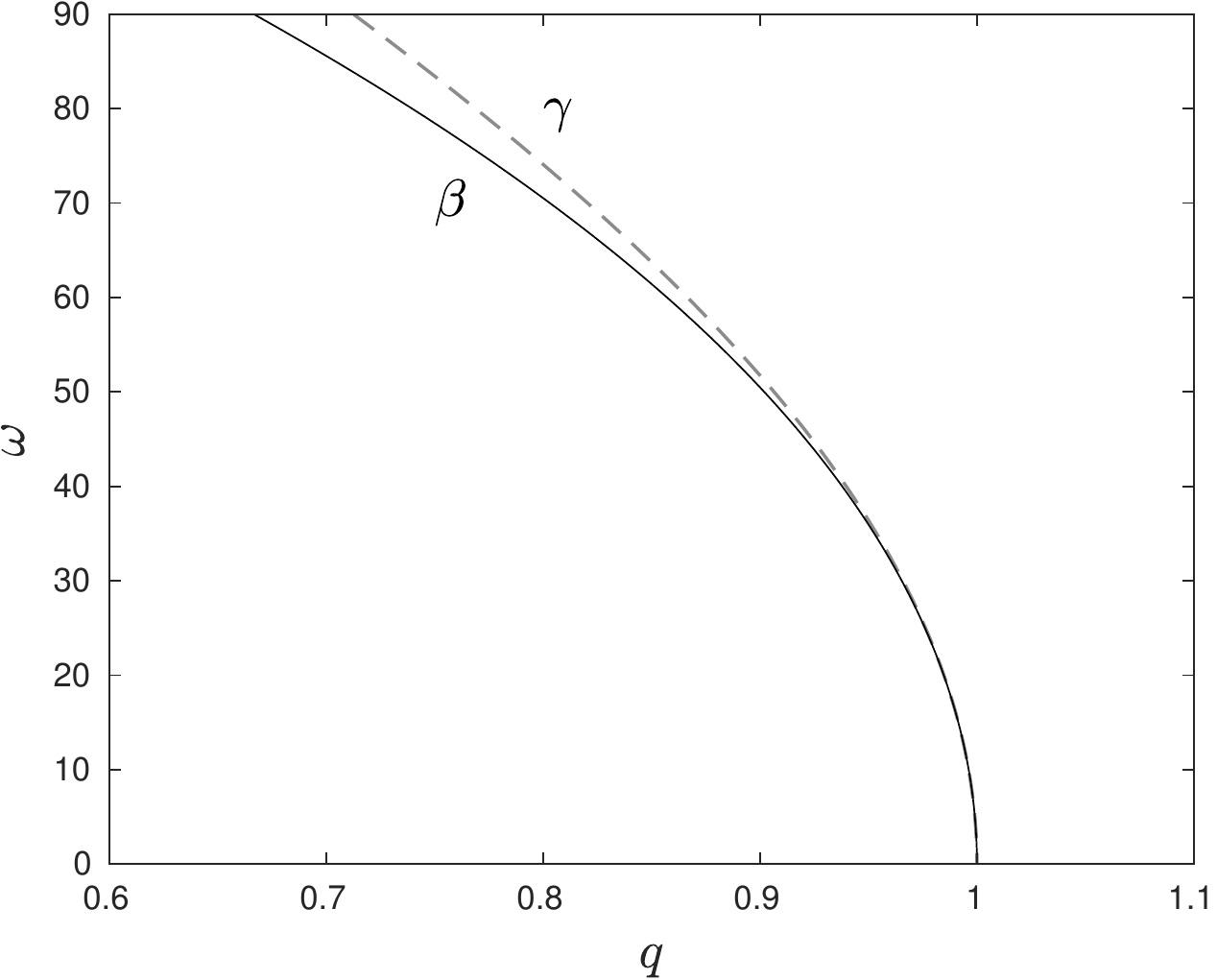}}
\caption{Comparison between the curves $\gamma$ and $\beta$.}
\label{fig:gamma_beta}
\end{figure}

In \cite{GV2013} the authors introduced the equation of a curve,
denoted by $\gamma$, which separates the region in the plane
$(q,\omega)$ where the trajectories maximizing $\dmin$ over ${\cal D}_1'$ have $e=0$, from
the region where such trajectories have $e=1$, that is, $\gamma$ is the set
of points $(q,\omega)$ where $q'-q$ and $\delta_\omega(q,\omega)$, defined in \eqref{deltaomega},
assume the same values. This equation is
\begin{equation}
  \begin{split}
&2q^4 + 2q'(-5+7y)q^3 -2q'^2(3y+22)(y-1)q^2 +\cr 
&  + q'^3(y^3 + 13y^2 + 9y - 27)q -2q'^4y^3 = 0,\cr
\end{split}
  \label{gamma_curve}
\end{equation}
with $y=\cos\omega$. The analogous equation for $\deltanod$ is
\begin{equation}
qy + 3q - 2q'y - 2q' = 0,
\label{beta_curve}
\end{equation}
that is easily obtained by equating $q'-q$ with
$\frac{2q}{1+\cos\omega}-q'$.  We denote by $\beta$ the curve defined
by \eqref{beta_curve}.  In Figure~\ref{fig:gamma_beta} we plot both
curves for comparison.

We also recall the following result (see \cite{GV2013}), stating
optimal bounds for the orbit distance $\dmin$ as functions of
$(q,e)$.\footnote{Here we state the result presented in \cite{GV2013} with a formula that is not
  singular for $e=1$.}
\begin{figure}[h!]
\centerline{\includegraphics[width=7cm]{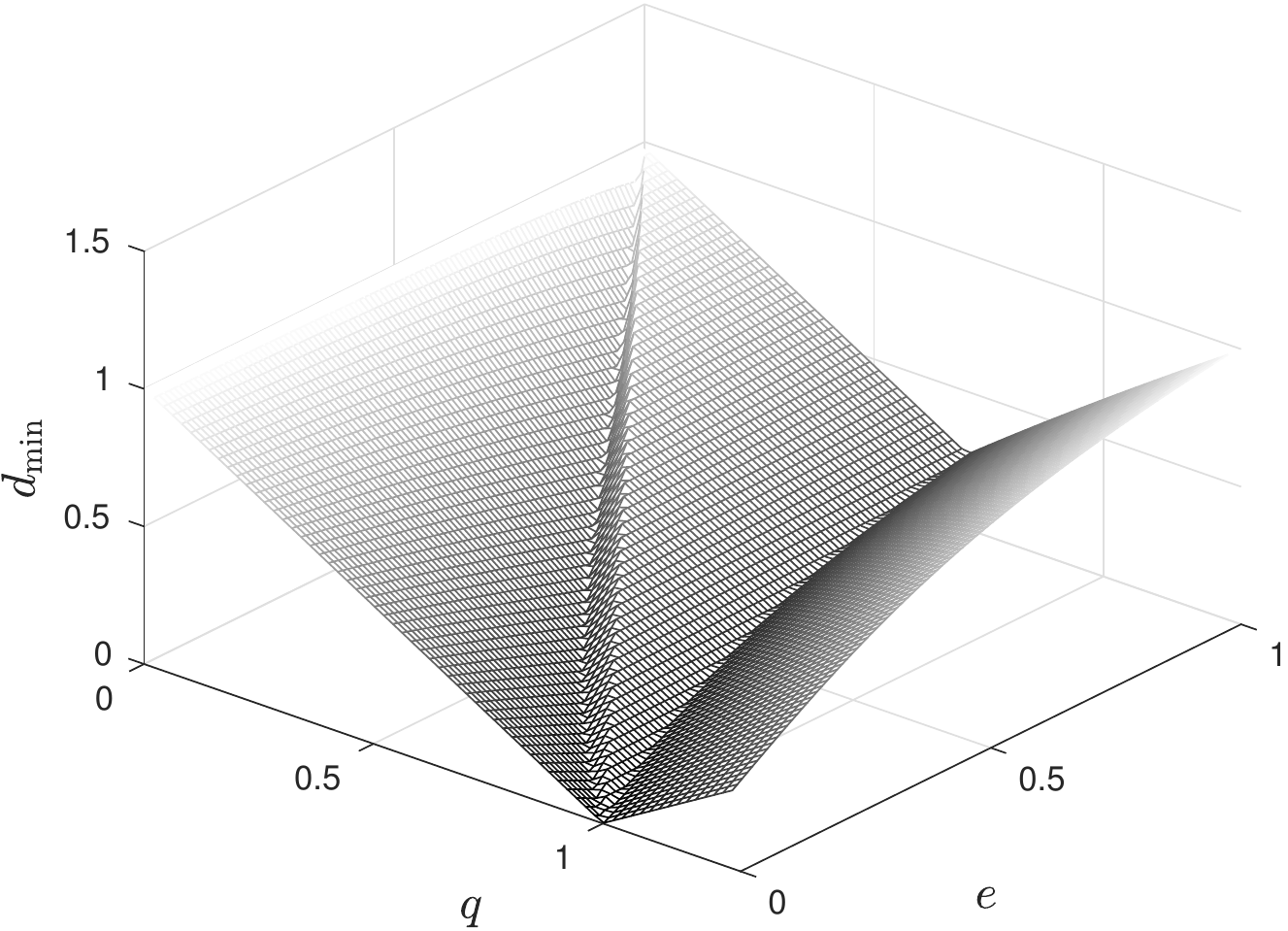}\hskip 0.5cm
\includegraphics[width=7cm]{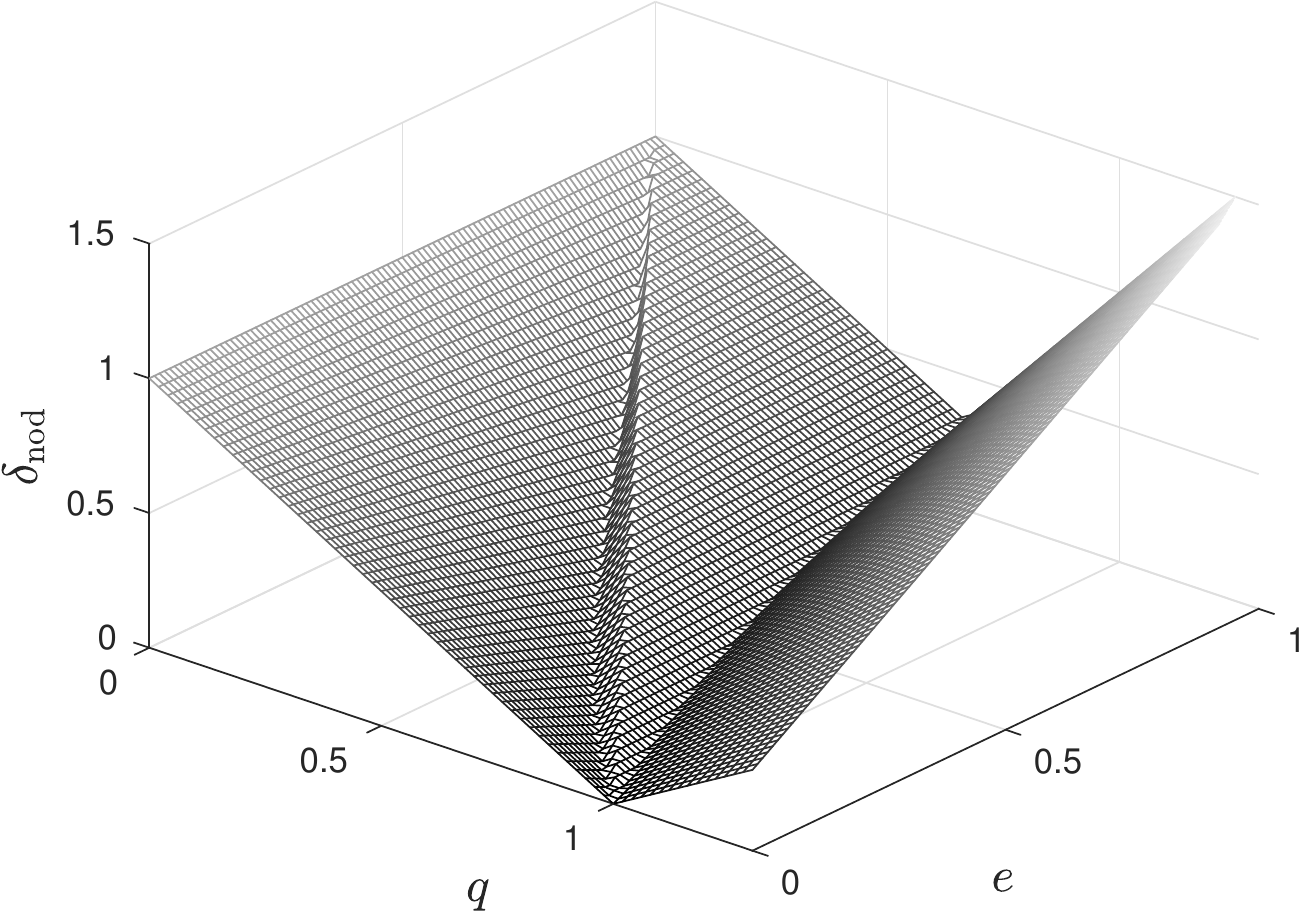}}
\caption{Left: $\max_{(I,\omega)\in{\cal D}_3'}\dmin(q,e)$. Right:
  $\max_{\omega\in{\cal D}_3''}\deltanod(q,e)$.}
\label{fig:bounds_qe_circ}
\end{figure}

\begin{proposition}
Set ${\cal D}_3' = \{(I,\omega): 0\leq I\leq \pi/2, 
0\leq \omega\leq \pi/2\}$ and
${\cal D}_4 = \{(q,e): 0< q\leq \qmax, 0\leq e\leq 1\}$.
For each choice of $(q,e)\in{\cal D}_4$ we have
\[
\begin{array}{l}
\displaystyle \min_{(I,\omega)\in{\cal D}_3'} \dmin =
\max\{0,q'-Q,q-q'\},\\
\displaystyle \max_{(I,\omega)\in{\cal D}_3'} \dmin = 
\max\{\min\{q'-q,Q-q'\},\delta_{e}(q,e)\},\\
\end{array}
\]
where $Q=q(1+e)/(1-e)$ is the (possibly infinite) apocenter distance
and $\delta_{e}(q,e)$ is the distance between ${\cal A}'$ and ${\cal
  A}$ with $I=\pi/2,\omega=\pi/2$:
\[
\delta_{e}(q,e) = \sqrt{(\xi-q')^2 + \biggl(\frac{\xi^2-q^2(1+e)^2}{qe(1+e) +\sqrt{(1+e)(q^2-\xi^2(1-e))}}\biggr)^2}
\]
where $\xi=\xi(q,e)$ is the unique real positive solution of
\[
\begin{split}
&e^4 x^4 + 2q'e^2(1-e^2)x^3 + (1+e)^2(q'^2(1-e)^2+q^2e^2)x^2 \\
-\: &2q'q^2e^2(1+e)^2x - q'^2q^2(1-e^2)(1+e)^2 = 0 .\\
\end{split}
\]
\label{prop:qe}
\end{proposition}

We compare the above result with the following.
\begin{proposition}
Set ${\cal D}_3'' = \{\omega: 0\leq \omega\leq
\pi/2\}$ and ${\cal D}_4 = \{(q,e): 0< q\leq \qmax, 0\leq e\leq 1\}$.
For each choice of $(q,e)\in{\cal D}_4$ we have
\begin{equation}
\left\{
\begin{array}{l}
\displaystyle \min_{\omega\in{\cal D}_3''} \deltanod = \max\{0,q'-Q,q-q'\}\\
\displaystyle \max_{\omega\in{\cal D}_3''} \deltanod =\max\{
\min\{q'-q, Q-q'\}, |q'-q(1+e)|\}.\\
\end{array}
\right.
\end{equation}
\label{prop:deltanod_qe}
\end{proposition}

\begin{proof}
  The result follows immediately by setting $e'=0$ in relations \eqref{lbe},
  \eqref{ube}.
  
\end{proof}

\noindent In Figure~\ref{fig:bounds_qe_circ}, for $q'=1$, we show the
graphics of $\max_{(I,\omega)\in{\cal D}_3'}\dmin(q,e)$ on the left,
and of $\max_{\omega\in{\cal D}_3''}\deltanod(q,e)$ on the right.

\section{Applications to the discovery of near-Earth asteroids}
\label{s:appl}

In Figure~\ref{known_NEAs_qom} we show the distribution of the known
population of near-Earth asteroids with absolute magnitude $H>22$
(faint NEAs) in the plane $(q,\omega)$. We have used the database of
NEODyS ({\tt https://newton.spacedys.com/neodys}) to the date of July 23,
  2019. On the left of the curve $\lowintom=0$, computed for $q'=1$ au
  and $e'=0$ and prolonged by symmetry, we can have only internal
  nodes (see also Figure~\ref{fig:qom}), therefore asteroids with
  those values of $(q,\omega)$ are difficult to be observed because
  they are always on the side of the Sun. This explains why this
  region appears depopulated. On the other hand, we can see that
  several asteroids are concentrated in a neighborhood of the curve
  $\beta$, defined by equation \eqref{beta_curve} and prolonged by
  symmetry, which represents the set of pairs $(q,\omega)$ where the
  value of $\deltanod$ can not be too large, whatever the value of
  $e$. In \cite{GV2013} the concentration of faint NEAs along the
  curve $\gamma$ defined by equation \eqref{gamma_curve} had already
  been noticed and explained by the same geometrical argument
  employing the orbit distance $\dmin$ instead of $\deltanod$. Here we
  observe that the curve $\beta$ is close to $\gamma$ (see
  Figure~\ref{fig:gamma_beta}), but it has a much simpler expression,
  therefore it can be easily used for a quick computation.

\begin{figure}[t!]
  \centerline{\includegraphics[width=10cm]{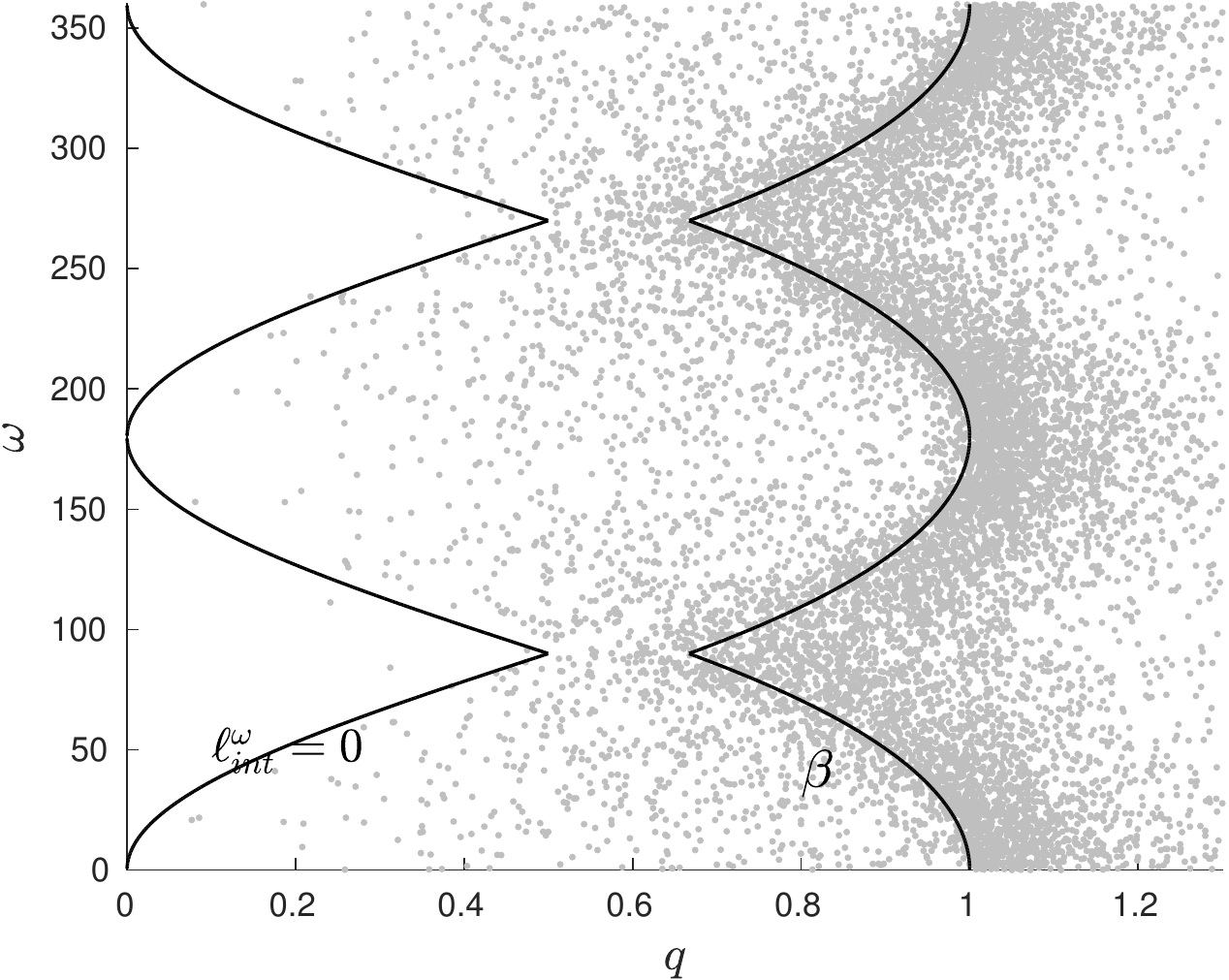}}
  \caption{Orbital distribution of the known NEAs in the plane
    $(q,\omega)$. The gray dots correspond to faint asteroids
    ($H>22$).}
  \label{known_NEAs_qom}
\end{figure}

%

\section{Comparison with the orbit distance $\dmin$}
\label{s:dmin}

In this section we discuss the analogies and the differences between
the upper bounds found for $\deltanod$ in
Propositions~\ref{prop:deltanod_qom_gen}, \ref{prop:deltanod_qe_gen},
\ref{prop:deltanod_qomp_gen} and similar upper bounds for $\dmin$,
computed by numerical methods.

In the mutual reference frame the coordinates of a point of ${\cal
  A}$ and another of ${\cal A}'$ are given by
\begin{equation}
\left\{
\begin{array}{l}
x = r\cos(f+\omega)\cr
y = r\sin(f+\omega)\cos I\cr
z = r\sin(f+\omega)\sin I\cr
\end{array}
\right.
\hskip 1cm
\left\{
\begin{array}{l}
x' = r'\cos(f'+\omega')\cr
y' = r'\sin(f'+\omega')\cr
z' = 0
\end{array}
\right.
\end{equation}
where
\[
r = \frac{q(1+e)}{1+e\cos f},
\hskip 1cm
r' = \frac{q'(1+e')}{1+e'\cos f'},
\]
with $f,f'\in[0,2\pi)$.
  Therefore, the squared distance between these two points is
\begin{eqnarray*}
d^2 &=& (x-x')^2 + (y-y')^2 + z^2 =\\
&=& r^2 + r'^2 - 2rr'
[\cos(f+\omega)\cos(f'+\omega') + \sin(f+\omega)\sin(f'+\omega')\cos I] =\\
&=& \frac{q^2(1+e)^2}{(1+e\cos f)^2} + \frac{q'^2(1+e')^2}{(1+e'\cos f')^2}
-\\
&-& 2\frac{q(1+e)}{1+e\cos f}\frac{q'(1+e')}{1+e'\cos f'}[\cos(f+\omega)\cos(f'+\omega') + \sin(f+\omega)\sin(f'+\omega')\cos I]
\end{eqnarray*}

From the expression above we see that we get all the possible values of the distance even if we restrict to the following ranges for $I,\omega,\omega'$:
\[
0\leq I\leq \pi/2, \qquad 0\leq \omega\leq\pi/2, \qquad
0\leq \omega' < 2\pi.
\]
or
\[
0\leq I\leq\pi/2, \qquad 0\leq \omega < 2\pi, \qquad 0\leq \omega'\leq \pi/2.
\]

\begin{figure}[t!]
  \centerline{\epsfig{figure=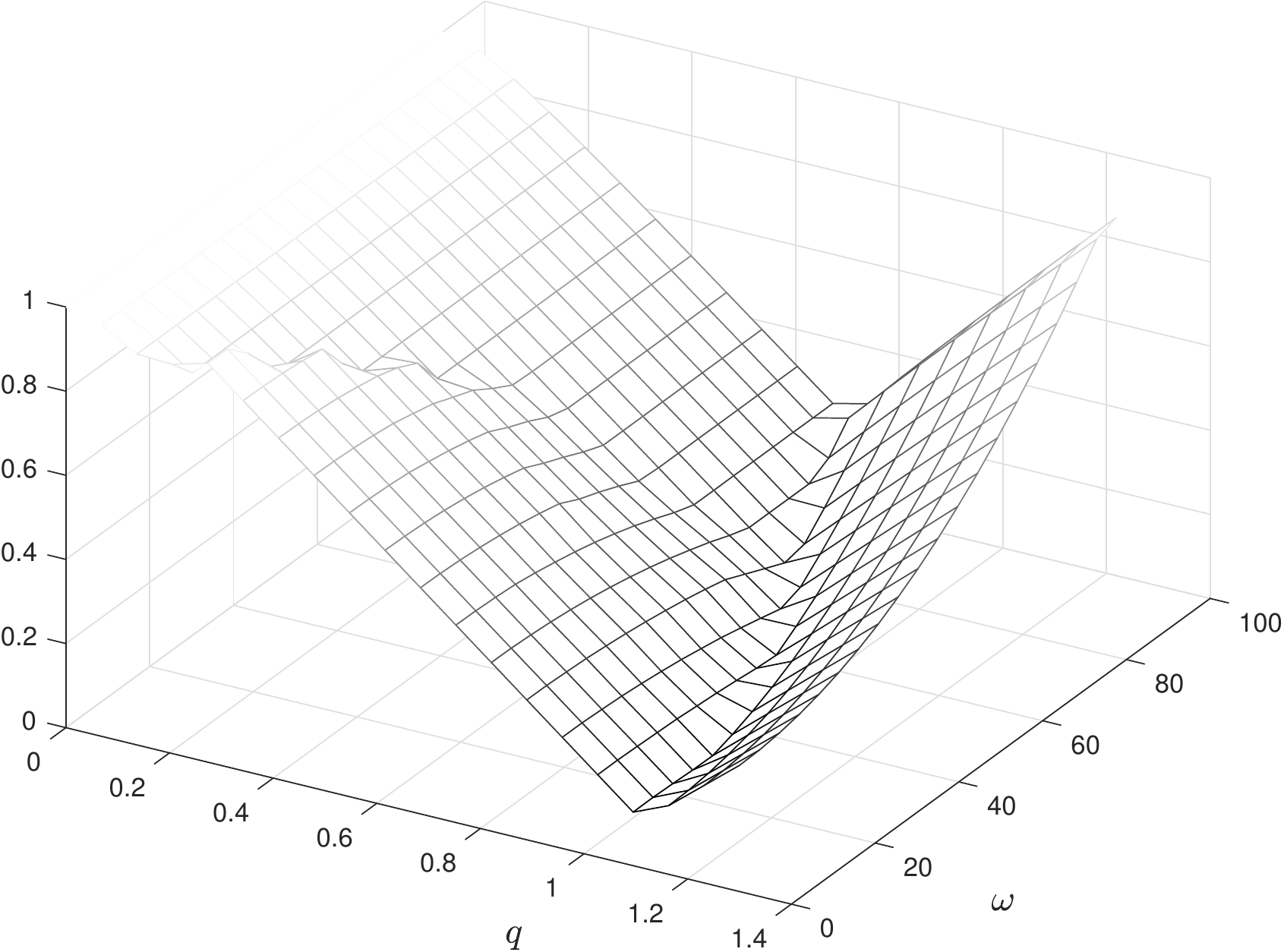,width=7cm}
    \hskip 0.2cm
    \epsfig{figure=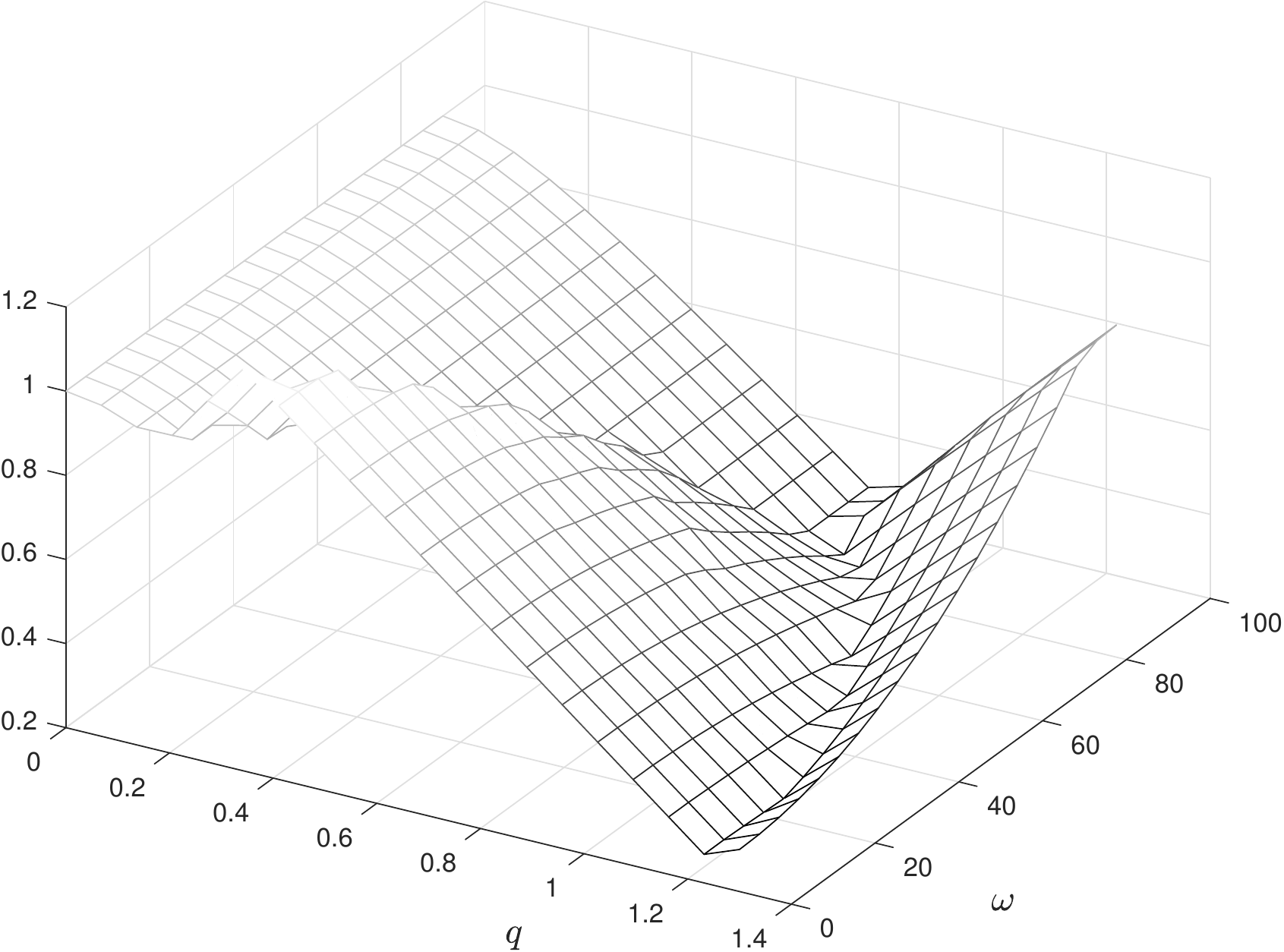,width=7cm}}
    \centerline{\epsfig{figure=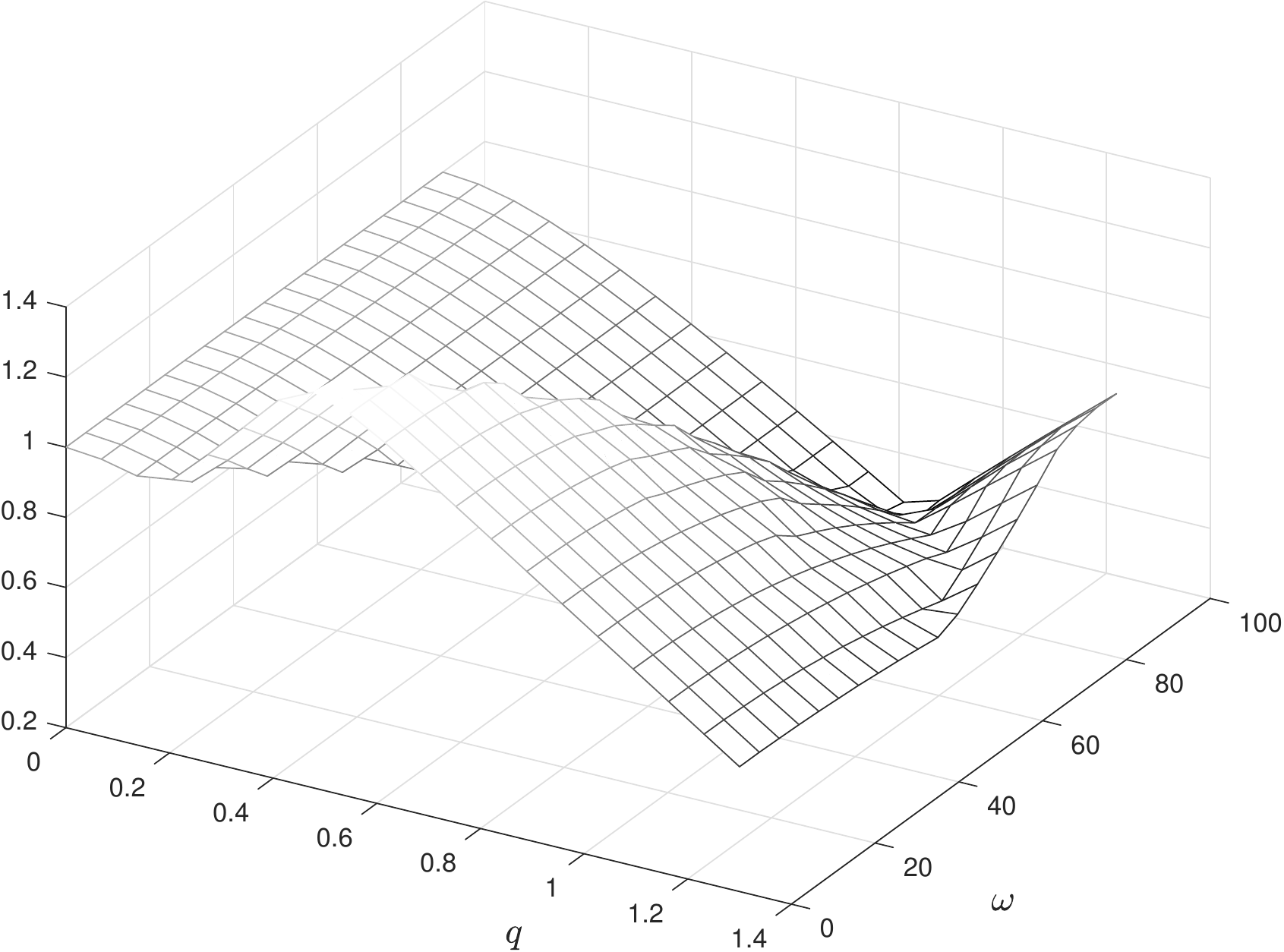,width=7cm}
    \hskip 0.2cm
    \epsfig{figure=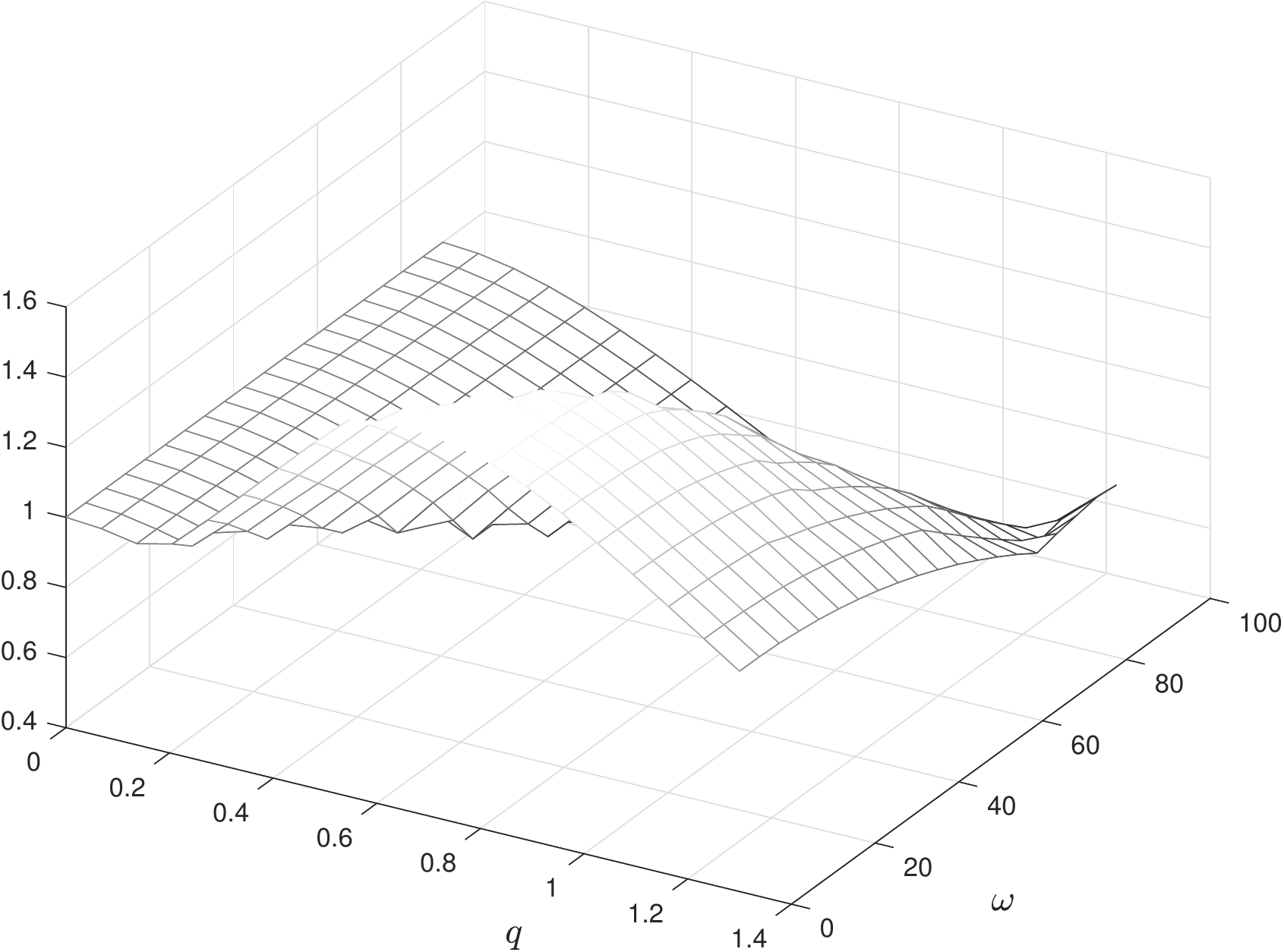,width=7cm}}
\caption{Graphic of $\max_{\widetilde{\cal D}_1}
  \dmin(q,\omega)$ for $e'=0.1$ (top left), $e'=0.2$ (top right),
  $e'=0.3$ (bottom left), $e'=0.4$ (bottom right). Here we set $q'=1$.}
\label{maxdminqom}
\end{figure}

In Figures~\ref{maxdminqom}, \ref{maxdminqe} we show, for different
values of $e'>0$, the graphics of $\max_{\widetilde{\cal D}_1}\dmin
(q,\omega)$ and $\max_{\widetilde{\cal D}_3}\dmin (q,e)$, where
\[
\begin{split}
\widetilde{\cal D}_1 &= \{(e,I,\omega'): 0\leq e\leq 1, \ 0\leq I\leq \pi/2, \ 0\leq \omega'\leq
2\pi \},\cr
\widetilde{\cal D}_3 &= \{(I, \omega,\omega'): 0\leq I\leq \pi/2, \ 0< \omega\leq \pi/2, \ 0\leq \omega'\leq
2\pi \}.\cr
\end{split}
\]
In both these cases we see that the graphics are similar to those in
Figures~\ref{maxdeltanodqom}, \ref{maxdeltanodqe}. In particular, the
{\em bulges} appearing in the graphics of $\max_{{\cal D}_1}\deltanod$
when $e'>0$ appear also in the graphics of
$\max_{\widetilde{\cal D}_1}\dmin$.

\begin{figure}[t!]
  \centerline{\epsfig{figure=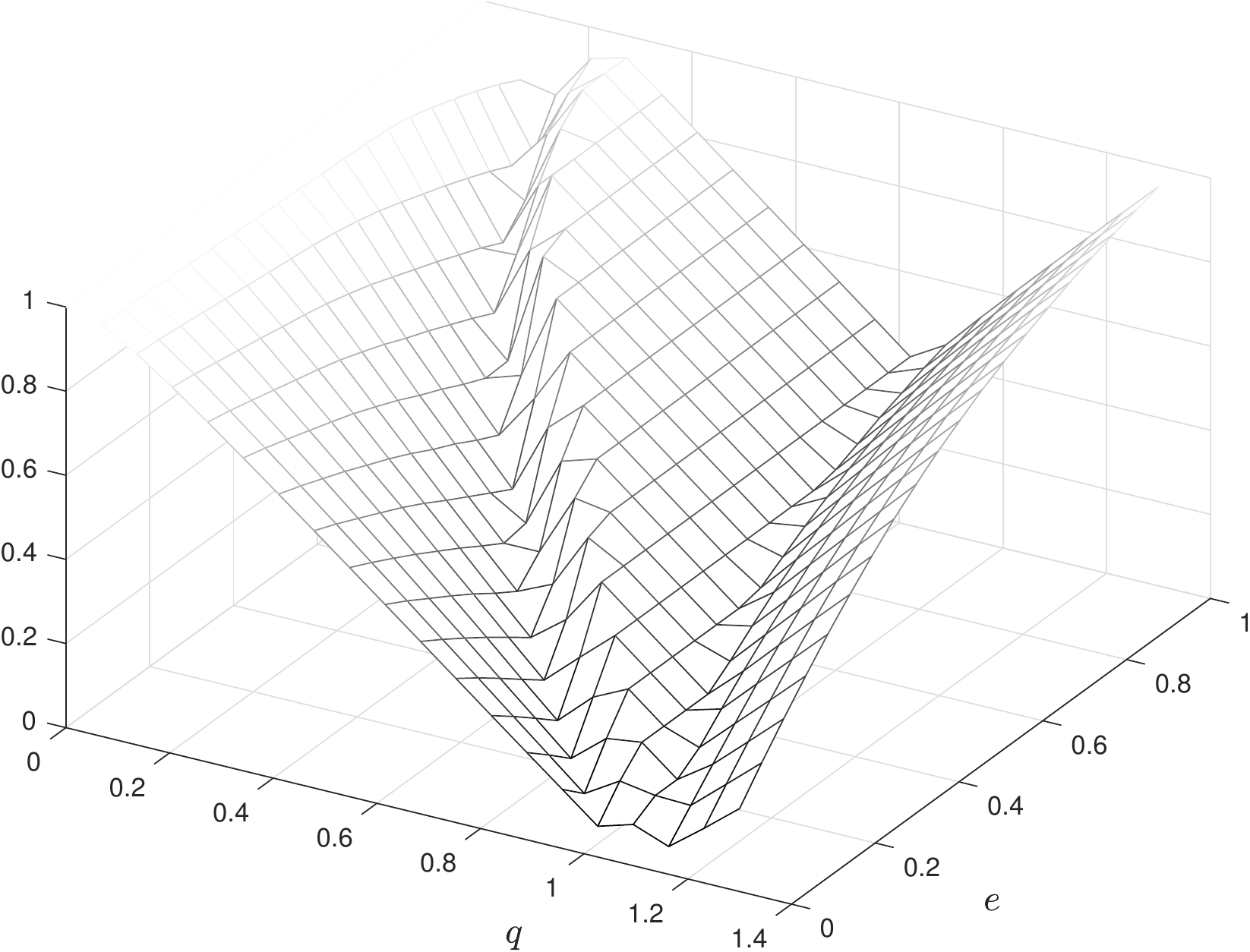,width=7cm}
    \hskip 0.2cm
    \epsfig{figure=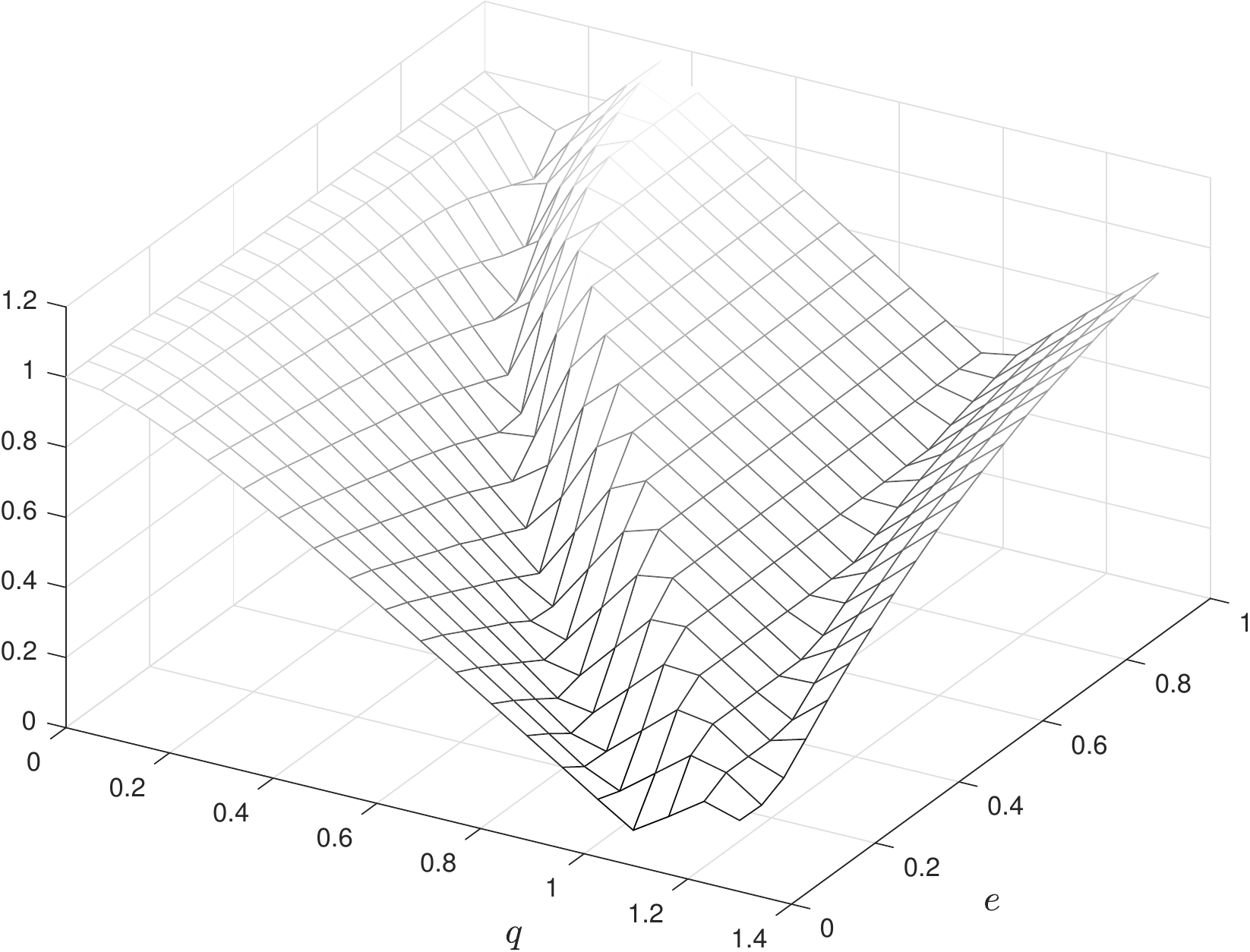,width=7cm}}
    \centerline{\epsfig{figure=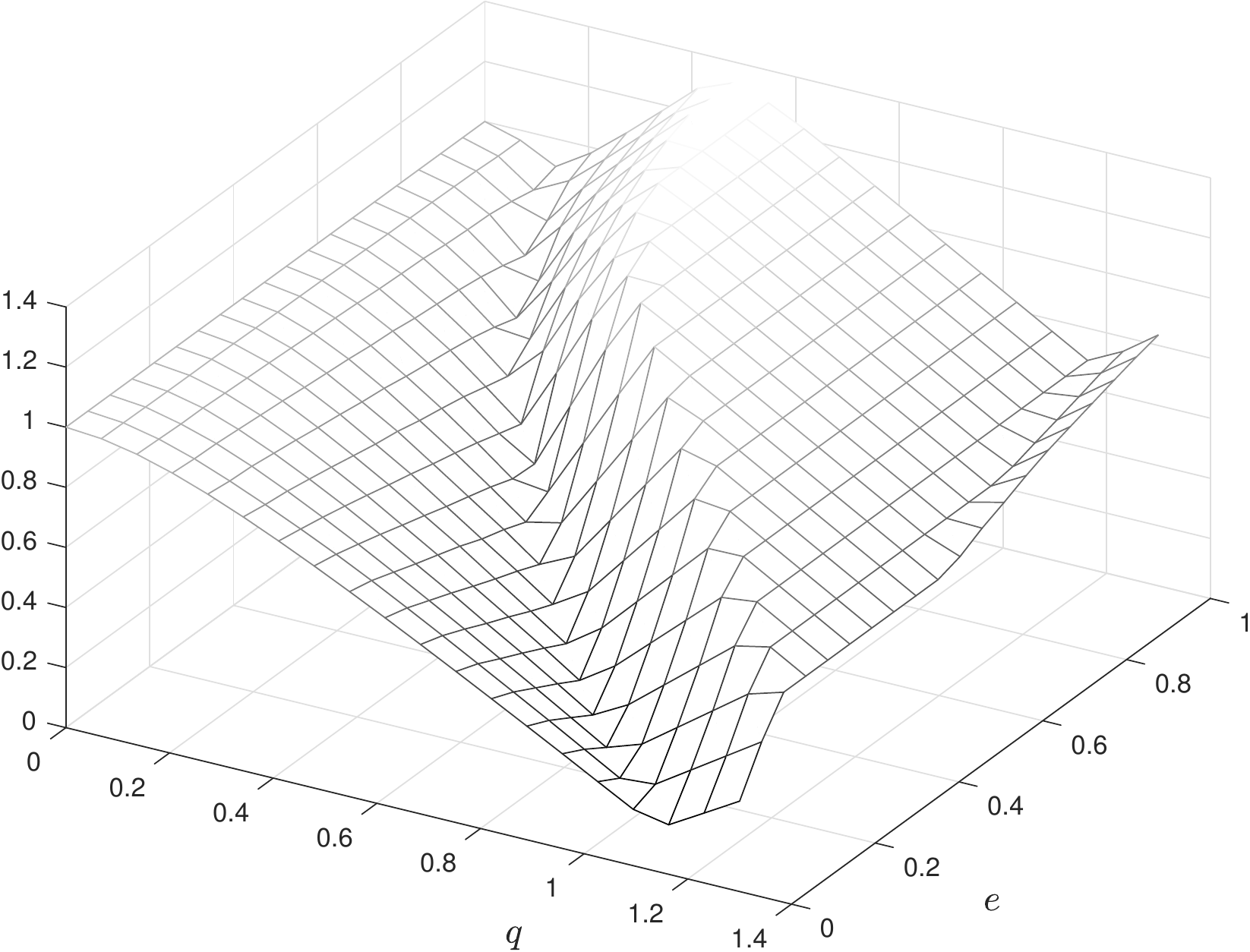,width=7cm}
    \hskip 0.2cm
        \epsfig{figure=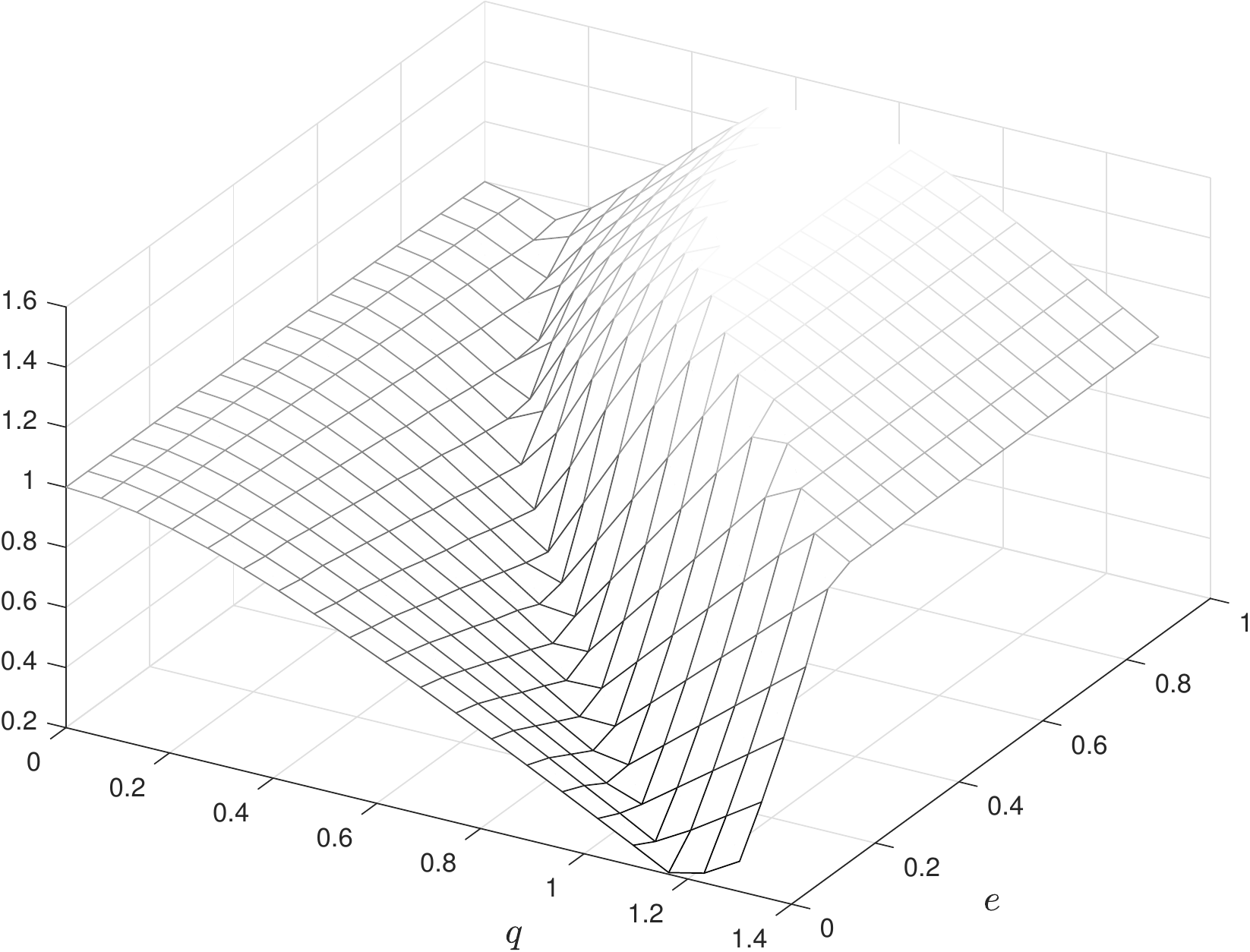,width=7cm}}
\caption{Graphic of $\max_{\widetilde{\cal D}_3} \dmin(q,e)$ for
  $e'=0.1$ (top left), $e'=0.2$ (top right), $e'=0.3$ (bottom left),
  $e'=0.4$ (bottom right). Here we set $q'=1$.}
\label{maxdminqe}
\end{figure}

\begin{figure}[t!]
  \centerline{\epsfig{figure=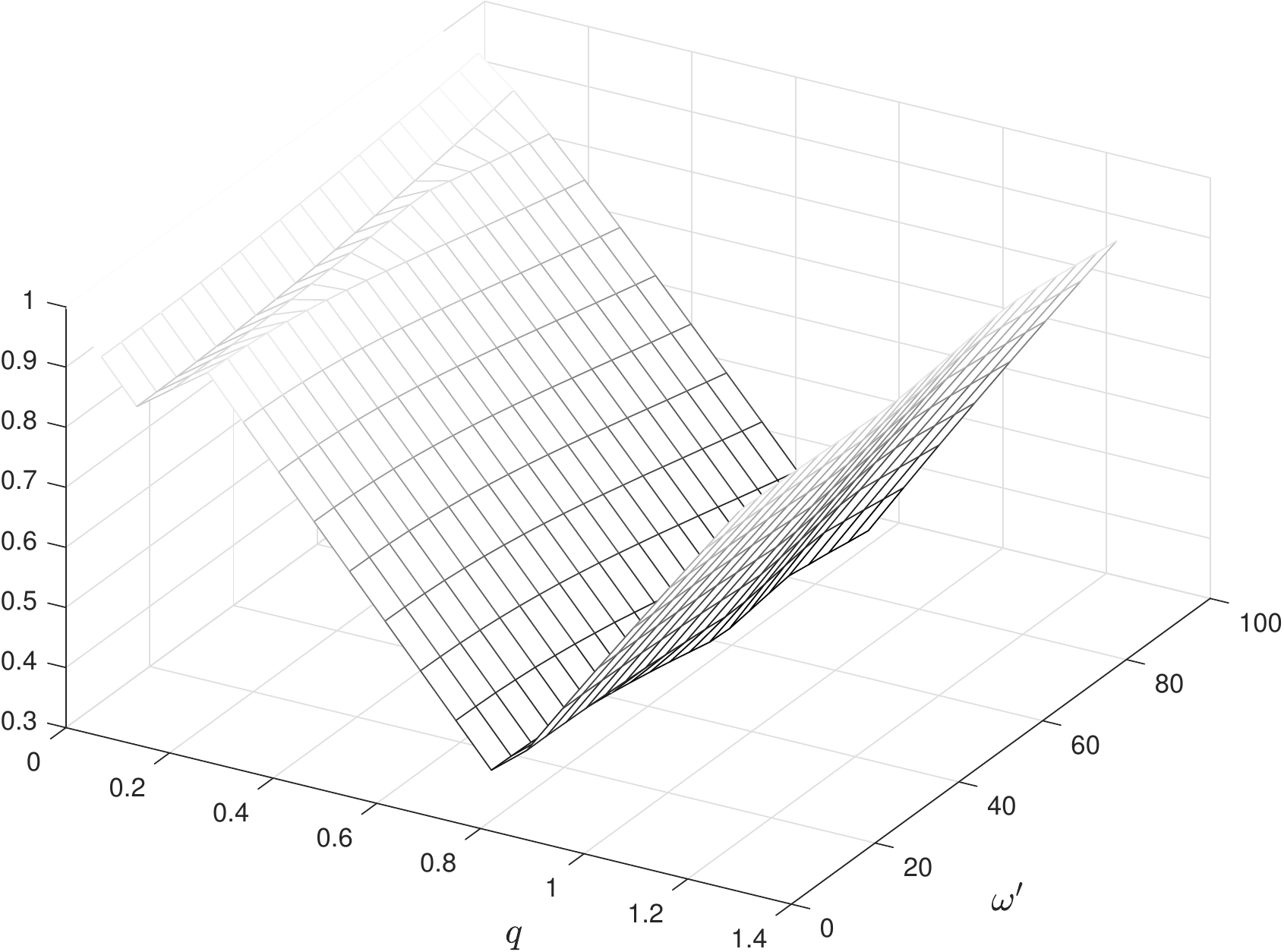,width=7cm}
    \hskip 0.2cm
    \epsfig{figure=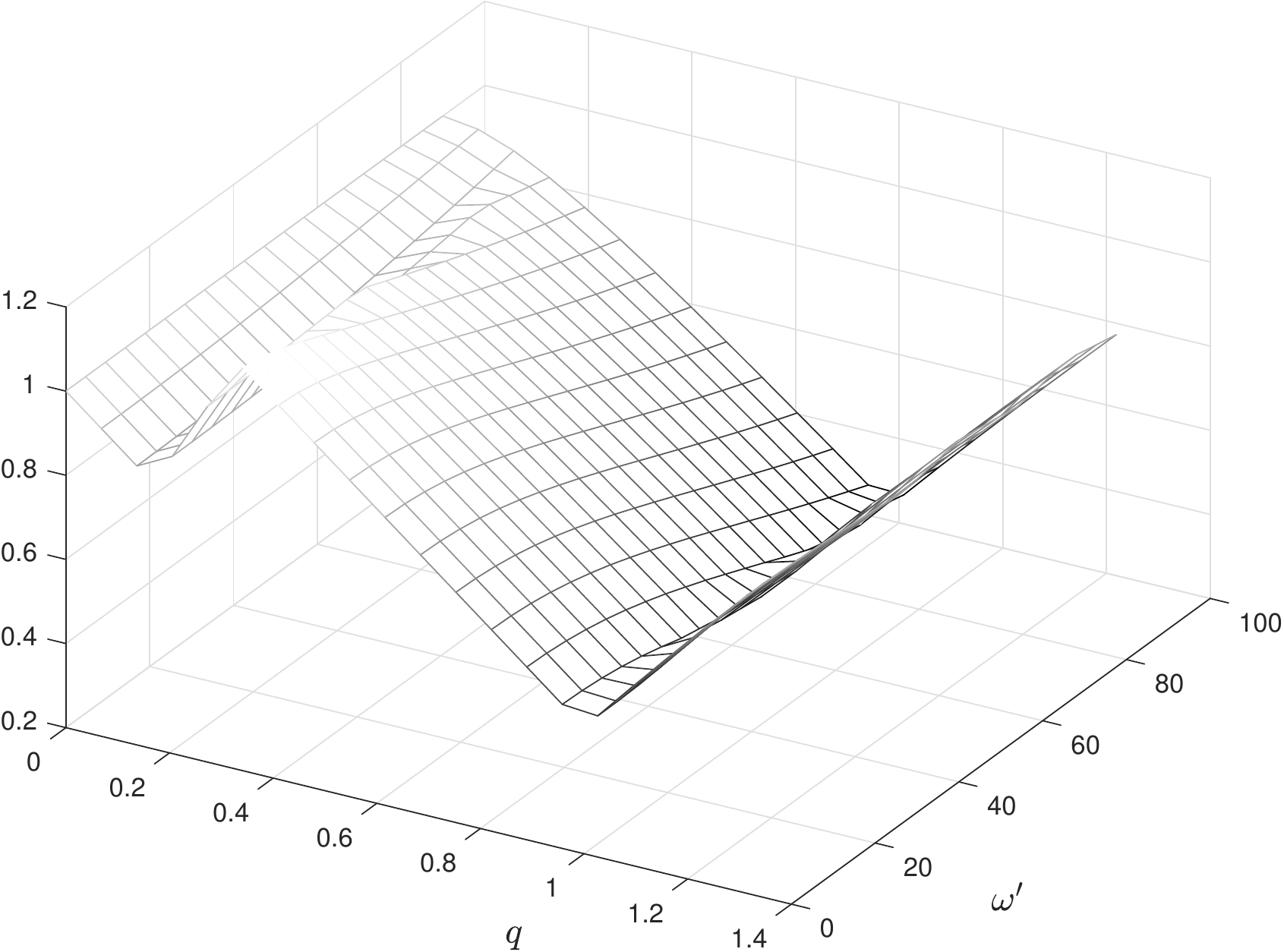,width=7cm}}
    \centerline{\epsfig{figure=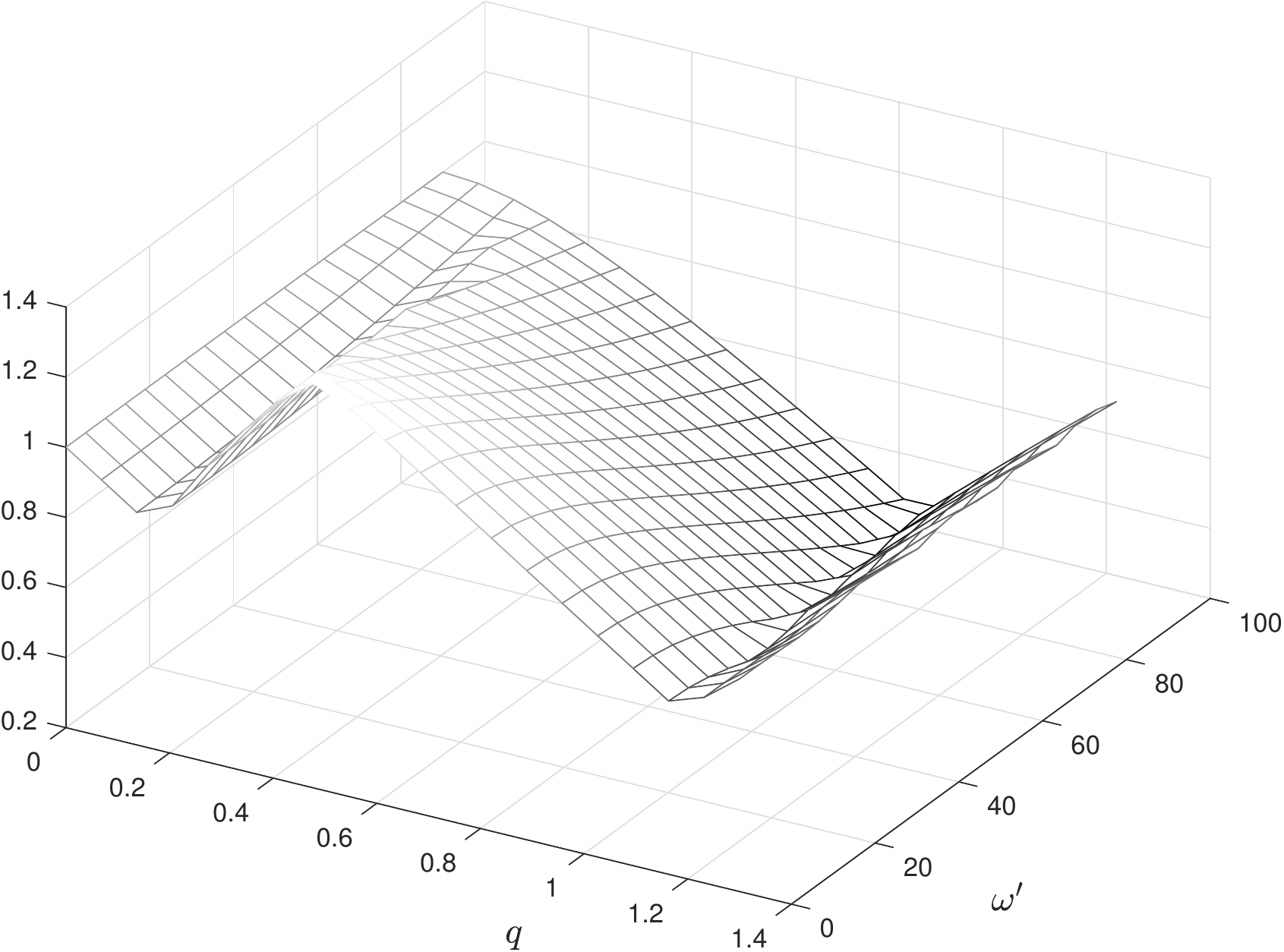,width=7cm}
    \hskip 0.2cm
        \epsfig{figure=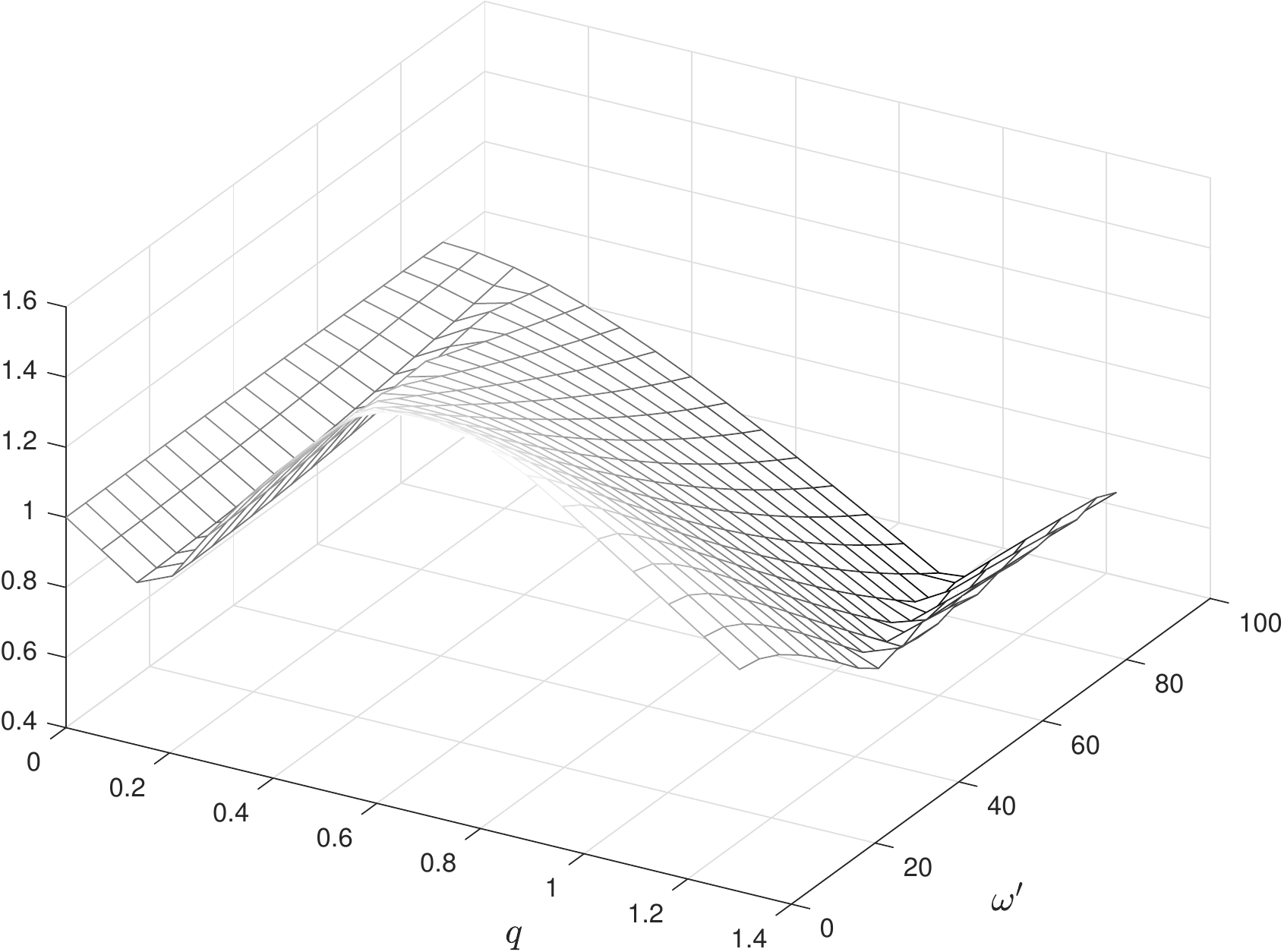,width=7cm}}
\caption{Graphic of $\max_{\widetilde{\cal D}_5} \dmin(q,\omega')$ for
  $e'=0.1$ (top left), $e'=0.2$ (top right), $e'=0.3$ (bottom left),
  $e'=0.4$ (bottom right). Here we set $q'=1$.}
\label{maxdminqomp}
\end{figure}

In Figure~\ref{maxdminqomp} we show, for the same values of $e'$, the
graphics of $\max_{\widetilde{\cal D}_5}\dmin (q,\omega')$, where
\[
\widetilde{\cal D}_5 = \{(e,I,\omega): 0\leq e\leq 1, \ 0\leq I\leq \pi/2,
\ 0\leq \omega\leq 2\pi \}.
\]
In this case the optimal bounds for $\deltanod$ displayed in
Figure~\ref{maxdeltanodqomp} has not the same features appearing here:
in fact the bulges appearing in the graphics of
$\max_{\widetilde{\cal D}_5}\dmin$ are not reproduced in the graphic of
$\max_{{\cal D}_5}\deltanod$.

\section{Conclusions}
\label{s:final}

We have introduced optimal bounds for the nodal distance $\deltanod$
between a given bounded Keplerian trajectory ${\cal A}'$ and another
Keplerian trajectory ${\cal A}$, with a focus in common with the
former, whose mutual orbital elements may vary.  Besides being
interesting in itself, this work aims at understanding how similar
bounds can be stated and proved for the orbit distance $\dmin$.
The conclusion is that the behavior of the upper bounds for $\deltanod$
given in Propositions~\ref{prop:deltanod_qom_gen},
\ref{prop:deltanod_qe_gen}, as functions of $(q,\omega)$ and $(q,e)$,
is similar to that for $\dmin$, obtained here by numerical computations.
On the other hand, the upper bound for $\deltanod$ given in
Proposition~\ref{prop:deltanod_qomp_gen}, as function of
$(q,\omega')$, is qualitatively different from that for $\dmin$.  As a
by-product of these results we have also found the equations of the
curves dividing the planes with coordinates $(q,\omega)$, $(q,e)$,
$(q,\omega')$ into regions where different linking configurations are
allowed.

\section{Acknowledgements}
Part of this work has been done during a visiting period of
G.F. Gronchi at the {\em Institut de m\'ecanique c\'eleste et de
  calcul des \'eph\'em\'erides (IMCCE), Observatoire de Paris}.
The same author also acknowledges the project MIUR-PRIN 20178CJA2B titled
"New frontiers of Celestial Mechanics: theory and applications".

\end{document}